\documentclass{article}

\usepackage{microtype}
\usepackage{graphicx}
\usepackage{subfigure}
\usepackage{booktabs}% professional tables

\usepackage[utf8]{inputenc}
\usepackage[T1]{fontenc}

\usepackage{xcolor}
\usepackage{textgreek}
\definecolor{sthlmLightBlue}{RGB}{214,237,252}
\definecolor{sthlmBlue}{RGB}{0,110,191}
\definecolor{sthlmLightGreen}{RGB}{213,247,244}
\definecolor{sthlmGreen}{RGB}{0,134,127}
\definecolor{sthlmLightGrey}{RGB}{213,217,225}
\definecolor{sthlmGrey}{RGB}{245,243,238}
\definecolor{sthlmDarkGrey}{RGB}{51,51,51}
\definecolor{sthlmLightOrange}{RGB}{255,215,210}
\definecolor{sthlmOrange}{RGB}{221,74,44}
\definecolor{sthlmLightPurple}{RGB}{241,230,252}
\definecolor{sthlmPurple}{RGB}{93,35,125}
\definecolor{sthlmLightRed}{RGB}{254,222,237}
\definecolor{sthlmRed}{RGB}{196,0,100}
\definecolor{sthlmYellow}{RGB}{252,191,10}

\usepackage[accepted]{icml2026}

\usepackage{amsmath, amssymb, amsfonts, mathtools, amsthm}
\usepackage[normalem]{ulem}
\usepackage{thmtools}
\usepackage[cal=boondox]{mathalfa}
\usepackage[mathcal]{euscript}
\usepackage{stmaryrd}

\usepackage{array,multirow,makecell}
\usepackage{pifont}
\usepackage{paralist}
\usepackage{enumerate}
\usepackage{wrapfig}
\usepackage{caption}
\usepackage{subcaption}
\usepackage{colortbl}
\usepackage{multirow}
\usepackage{array}

\usepackage{algorithm}
\usepackage{algorithmic}

\usepackage{pgfplots}
\pgfplotsset{compat=1.18}
\usepgfplotslibrary{fillbetween}
\usepackage{tikz}
\usetikzlibrary{shapes,shadows,shadows.blur,shapes,shapes.callouts,shapes.geometric,arrows,automata,fit,shadows,trees,decorations,decorations.pathreplacing,decorations.pathmorphing,decorations.text,positioning,patterns,backgrounds,matrix,arrows.meta,calc,petri,chains,bending,scopes}
\tikzset{arrow/.style={-{Stealth[length=2mm]},thick}}

\usepackage{url}
\usepackage[colorlinks]{hyperref}
\hypersetup{
  linkcolor=sthlmRed,
  citecolor=sthlmGreen,
  urlcolor=sthlmBlue
}

\usepackage[capitalize,noabbrev]{cleveref}

\usepackage{todonotes} % comment out before camera-ready if desired

\usepackage{xspace}
\def\ie{{\em i.e.,}\xspace}

\def\cf{{\em cf.}\xspace}

\definecolor{RUGRed}{HTML}{C00016}
\definecolor{P1}{HTML}{2B8CBE}
\definecolor{P2}{HTML}{EB6A3D}
\definecolor{P3}{HTML}{2E8540}
\definecolor{P4}{HTML}{9B4D96}
\definecolor{P5}{HTML}{800080}
\definecolor{PanelGray}{RGB}{245,247,250}

\newcommand{\myrowcolour}{\rowcolor[gray]{0.925}}

\newcommand{\highest}[1]{\textcolor{sthlmRed}{\textbf{#1}}}

\DeclareMathOperator*{\argmax}{arg\,max}
\DeclareMathOperator*{\argmin}{arg\,min}

\theoremstyle{plain}
\newtheorem{theorem}{Theorem}[section]
\newtheorem{proposition}[theorem]{Proposition}
\newtheorem{lemma}[theorem]{Lemma}

\theoremstyle{definition}
\newtheorem{definition}[theorem]{Definition}

\theoremstyle{remark}

\newcommand{\defeq}{\doteq}

\tikzset{
 rightrobot/.pic={
code={
 \tikzset{scale=10/10}, 
 
 \draw[fill=gray!30, thick] (0,4) circle (0.5);

 \draw[fill=gray!30, thick]  (-0.5,4)--(-0.5,2.5)  --  (0.5,2.5) -- (0.5,4);

 \draw[fill=white, thick] (-.65,3.889) rectangle (-.265,3.959);
 \draw[fill=white, thick] (-.6,4) rectangle (-.265,3.85);

 \draw[thick] (-.85,3.2) -- (0,3.55);
 \draw[thick] (-.85,3.2) -- (-.85,2.95);
 \draw[thick] (-.85,3.2) -- (-1.1,3.2);

 \draw[fill=white, thick] (0,2.5) circle (0.25);
 
 \draw[fill=white] (0.1,4.155) ellipse (.295cm and .25cm) node[scale=10/10] at (0.1,4.155) {}; %$m^\mathtt{A}$$\pi(s)$
}
  }
}

\tikzset{
 leftrobot/.pic={
code={
 \tikzset{scale=10/10}, 
 
 \draw[fill=gray!30, thick] (0,4) circle (0.5);

 \draw[fill=gray!30, thick]  (-0.5,4)--(-0.5,2.5)  --  (0.5,2.5) -- (0.5,4);

 \draw[fill=white, thick] (+.65,3.889) rectangle (+.385,3.959);
 \draw[fill=white, thick] (.6,4) rectang le (+.335,3.85);

 \draw[thick] (.85,3.2) -- (0,3.55);
 \draw[thick] (.85,3.2) -- (.85,2.95);
 \draw[thick] (.85,3.2) -- (1.1,3.2);

 \draw[fill=white, thick] (0,2.5) circle (0.25);
 
 \draw[fill=white] (-.1,4.155) ellipse (.295cm and .25cm) node[scale=10/10] at (-.1,4.155) {$\pi(s)$}; 
}
  }
}

\tikzset{
 passiveleftrobot/.pic={
code={ 
 \draw[fill=gray!30, thick] (0,4) circle (0.5);

 \draw[fill=gray!30, thick]  (-0.5,4) -- (-0.5,2.5) -- (0.5,2.5) -- (0.5,4);

 \draw[fill=white, thick] (+.65,3.889) rectangle (+.385,3.959);
 \draw[fill=white, thick] (.6,4) rectangle (+.335,3.85);

 \draw[thick] (.85,3.2) -- (0,3.55);
 \draw[thick] (.85,3.2) -- (.85,2.95);
 \draw[thick] (.85,3.2) -- (1.1,3.2);

 \draw[fill=white, thick] (0,2.5) circle (0.25);
 
 \draw[fill=white] (-.1,4.155) ellipse (.295cm and .25cm) node at (-.1,4.155) {}; 
}
  }
}

\tikzset{
 passiverightrobot/.pic={
code={ 
 \draw[fill=gray!30, thick] (0,4) circle (0.5);

 \draw[fill=gray!30, thick]  (-0.5,4) -- (-0.5,2.5) -- (0.5,2.5) -- (0.5,4);

 \draw[fill=white, thick] (-.65,3.889) rectangle (-.265,3.959);
 \draw[fill=white, thick] (-.6,4) rectangle (-.265,3.85);

 \draw[thick] (-.85,3.2) -- (0,3.55);
 \draw[thick] (-.85,3.2) -- (-.85,2.95);
 \draw[thick] (-.85,3.2) -- (-1.1,3.2);

 \draw[fill=white, thick] (0,2.5) circle (0.25);
 
 \draw[fill=white] (-.1,4.155) ellipse (.295cm and .25cm) node at (-.1,4.155) {}; 
}
  }
}

\icmltitlerunning{An \(\varepsilon\)-Optimal Sequential Approach for Solving zs-POSGs}

\begin{document}
\interdisplaylinepenalty=1000
\setlength{\abovedisplayskip}{6pt}
\setlength{\belowdisplayskip}{6pt}

\twocolumn[
\icmltitle{An \(\varepsilon\)-Optimal Sequential Approach for Solving zs-POSGs}
\begin{icmlauthorlist}
\icmlauthor{Erwan C. Escudie}{groningen}
\icmlauthor{Matthia Sabatelli}{groningen}
\icmlauthor{Jilles S. Dibangoye}{groningen}
\end{icmlauthorlist}

\icmlaffiliation{groningen}{
Bernoulli Institute, University of Groningen,
Nijenborgh 4, NL-9747 AG Groningen, Netherlands}

\icmlcorrespondingauthor{Erwan C. Escudie}{e.c.escudie@rug.nl}
\icmlcorrespondingauthor{Matthia Sabatelli}{m.sabatelli@rug.nl}
\icmlcorrespondingauthor{Jilles S. Dibangoye}{j.s.dibangoye@rug.nl}

\icmlkeywords{}

\vskip 0.3in
]

\printAffiliationsAndNotice

\begin{abstract}
While recent reductions of zero-sum partially observable stochastic games (zs-POSGs) to transition-independent stochastic games (TI-SGs) theoretically admit dynamic programming, practical solutions remain stifled by the inherent non-linearity and exponential complexity of the simultaneous minimax backup. In this work, we surmount this computational barrier by rigorously recasting the simultaneous interaction as a sequential decision process via the principle of separation. We introduce distinct sufficient statistics for valuation and execution—the sequential occupancy state and the private occupancy family—which reveal a latent geometry in the optimal value function. This structural insight allows us to linearise the backup operator, reducing the update complexity from exponential to polynomial while enabling the direct extraction of safe policies without heuristic bookkeeping. Experimental results demonstrate that algorithms leveraging this sequential framework significantly outperform state-of-the-art methods, effectively rendering previously intractable domains solvable.
\end{abstract}

\section{Introduction}

The principle of optimality, introduced by \citeauthor{bellman} \citep{bellman} in the 1950s, remains the cornerstone of sequential decision-making. While its rigorous application to fully observable Markov Decision Processes (MDPs) and zero-sum Stochastic Games (SGs) \citep{shapley} is now settled theory, extending this level of control to simultaneous-move zero-sum partially observable stochastic games (zs-POSGs) has long been regarded as one of the hardest open problems in algorithmic game theory. The combination of partial observability, adversarial dynamics, and simultaneous moves induces a severe curse of dimensionality that has historically defied efficient solution methods \citep{HansenBZ04}.

Significant progress has been made in the adjacent field of cooperative POSGs (Dec-POMDPs), where the adoption of a centralised planning perspective allowed for the reformulation of the problem as a continuous-state MDP \citep{nayyar2013decentralized,oliehoek2013sufficient,DibangoyeABC13}. This approach facilitated the transfer of powerful dynamic programming techniques \citep{Dibangoye:OMDP:2016,pmlr-v80-dibangoye18a} and heuristic search algorithms \citep{Szer05,OliehoekSDA10,dibangoye2012scaling,peralezsolving,Peralez_Delage_Castellini_Cunha_Dibangoye_2025}, effectively rendering complex cooperative domains tractable. In stark contrast, the adversarial POSG landscape has remained fragmented. Approaches relying on heuristic search \citep{horak2017heuristic,HorBos-aaai19,Buffet2020,Buffet2020bis}, deep reinforcement learning \citep{BrownBLG20,Moravck2017DeepStackEA}, or regularised public-belief abstractions \citep{SokotaDL0KB23} often struggle to balance scalability with game-theoretic guarantees. While recent work has explored convex value functions \citep{Wiggers16,cunha2023convex,7963513}, these methods typically grapple with intricate bookkeeping \citep{DelBufDibSaf-DGAA-23} or restrictive structural assumptions to extract safe policies.

Recent advances have begun to crack this barrier. Notably, \citet{escudie2025varepsilonoptimally} achieved a structural breakthrough by demonstrating that general zs-POSGs can be losslessly reduced to transition-independent zero-sum stochastic games (TI-zs-SGs). This reduction is profound: it lifts the nebulous space of partial observability into a structured Markovian game played over \emph{occupancy states} (beliefs over states and joint histories), theoretically opening the door to dynamic programming in the spirit of \citeauthor{shapley}. However, in practice, the resulting computational burden remains prohibitive. The fundamental bottleneck lies in the simultaneous nature of the backup operator. In the TI-zs-SG formulation, the ``actions'' available to the planner are decision rules—distributions over private actions conditioned on history. Even in the reduced game, the standard Bellman backup requires solving a matrix game (minimax optimisation) over these decision rules at every state. Since the space of decision rules grows exponentially with the horizon, the simultaneous backup operator is inherently non-linear and computationally ruinous. Thus, despite the structural reduction, the community remains stifled by a computational wall.

\textbf{The Sequential Paradigm Shift.} In this work, we argue that the intractability of the simultaneous backup is not an intrinsic property of the game, but a consequence of how we structure the optimisation. We propose a paradigm shift: we show that optimal reasoning in simultaneous games can be rigorously recast as a sequential decision process, effectively decoupling the agents' choices in the planner's reasoning phase without altering the simultaneous semantics of the execution phase. Therefore, in this paper, we investigate the following research question:

\tikzstyle{mybox} = [draw=black, very thick, rectangle, rounded corners, inner ysep=10pt, inner xsep=10pt]
\vspace{-12pt}
\begin{center}
\begin{tikzpicture}
\node [mybox] (box){
\begin{minipage}{.93\linewidth}
\emph{How can the computationally ruinous simultaneous minimax backup be decomposed into tractable sequential updates without sacrificing game-theoretic optimality?}
\end{minipage}
};
\end{tikzpicture}
\end{center}
\vspace{-8pt}

Our answer rests on the principle of separation \citep{DibangoyeB015}, which suggests that the statistics required to compute \emph{values} do not need to be identical to those required to extract \emph{policies}. By formalising this distinction, we introduce two novel information structures: the \emph{sequential occupancy state}, a value-sufficient statistic that allows us to decompose the joint value function, and the \emph{private occupancy family}, a policy-sufficient statistic that enables the direct extraction of safe strategies. This structural dichotomy allows us to reveal the geometry of the optimal value function, fundamentally altering the optimisation landscape. We demonstrate that the sequential backup operator becomes polynomial, transforming the intractable non-linear simultaneous update into a sequence of tractable linear programs. By disentangling the players' decision variables, we collapse the complexity of the Bellman update from exponential to polynomial complexity. Finally, we integrate these findings into a \emph{point-based value iteration} (PBVI) framework \citep{pineau2003point}, demonstrating that our sequential approach solves established benchmarks from the literature \citep{Wiggers16} with significantly higher numerical stability and lower exploitability than previous simultaneous methods \citep{escudie2025varepsilonoptimally,DelBufDibSaf-DGAA-23,Tammelin14}, effectively bridging the gap between theoretical solvability and practical tractability.

\section{Preliminaries}
\label{sec:preliminaries}

In this section, we formally define zero-sum partially observable stochastic games (zs-POSGs) and establish the notation used throughout our work. We then review the prevailing simultaneous-move central-planner formulation and show why, despite yielding a theoretically sound reduction, it encounters a fundamental computational barrier. This will motivate the sequential paradigm shift developed in the remainder of the paper.

\subsection{Simultaneous-Move zs-POSGs}

A simultaneous-move zs-POSG, denoted \(\mathcal{M}\), is defined by the tuple \(\mathcal{M} = (\mathcal{S}, \mathcal{A}_{\textcolor{sthlmRed}{1}}, \mathcal{A}_{\textcolor{sthlmRed}{2}}, \mathcal{Z}_{\textcolor{sthlmRed}{1}}, \mathcal{Z}_{\textcolor{sthlmRed}{2}}, p, r, b, \gamma, \ell)\). The game proceeds over a finite horizon \(\ell\), involving two competing agents, player \(\textcolor{sthlmRed}{1}\) (maximiser) and player \(\textcolor{sthlmRed}{2}\) (minimiser). The finite set \(\mathcal{S}\) represents the hidden states of the environment. At each stage \(t\), both players simultaneously select private actions \(a_{\textcolor{sthlmRed}{1}} \in \mathcal{A}_{\textcolor{sthlmRed}{1}}\) and \(a_{\textcolor{sthlmRed}{2}} \in \mathcal{A}_{\textcolor{sthlmRed}{2}}\). These actions, combined with the current state \(s\), trigger a stochastic transition to a next state \(s'\) and the generation of observations according to the joint probability kernel \(p(s', z_{\textcolor{sthlmRed}{1}}, z_{\textcolor{sthlmRed}{2}} \mid s, a_{\textcolor{sthlmRed}{1}}, a_{\textcolor{sthlmRed}{2}})\). Here, \(z_{\textcolor{sthlmRed}{i}} \in \mathcal{Z}_{\textcolor{sthlmRed}{i}}\) denotes the private observation received by player \(\textcolor{sthlmRed}{i}\). The immediate payoff \(r(s, a_{\textcolor{sthlmRed}{1}}, a_{\textcolor{sthlmRed}{2}})\) rewards player \(\textcolor{sthlmRed}{1}\) at the expense of player \(\textcolor{sthlmRed}{2}\). The game initiates from a distribution \(b \in \Delta(\mathcal{S})\) and discounts future rewards by \(\gamma \in [0,1)\). 

Unlike in fully observable games, players here act based on imperfect information. We define the private history \(h_{\textcolor{sthlmRed}{i},t}\) as the sequence of actions and observations available to player \(\textcolor{sthlmRed}{i}\) up to stage \(t\). To manage this uncertainty, players employ \emph{policies} \(\pi_{\textcolor{sthlmRed}{i}} = (d_{\textcolor{sthlmRed}{i},0}, \dots, d_{\textcolor{sthlmRed}{i},\ell-1})\), which are sequences of history-dependent decision rules. A decision rule \(d_{\textcolor{sthlmRed}{i},t}: \mathcal{H}_{\textcolor{sthlmRed}{i},t} \to \Delta(\mathcal{A}_{\textcolor{sthlmRed}{i}})\) maps private histories to probability distributions over actions. We denote the set of all such decision rules at stage \(t\) by \(\mathcal{D}_{\textcolor{sthlmRed}{i},t}\). The value of the game under a joint policy \((\pi_{\textcolor{sthlmRed}{1}}, \pi_{\textcolor{sthlmRed}{2}})\) is the expected cumulative discounted return from the initial belief \(b\), denoted \(v_{\pi_{\textcolor{sthlmRed}{1}}, \pi_{\textcolor{sthlmRed}{2}}}(b)\). Under the minimax theorem \citep{Neumann1928,Sion1958}, the game admits a unique optimal value \(v_{\mathrm{sim}}^*(b) = \max_{\pi_{\textcolor{sthlmRed}{1}}} \min_{\pi_{\textcolor{sthlmRed}{2}}} v_{\pi_{\textcolor{sthlmRed}{1}}, \pi_{\textcolor{sthlmRed}{2}}}(b)\).

\subsection{The Computational Barrier of Simultaneity}

Extending \citeauthor{bellman}'s principle of optimality to zs-POSGs requires an \emph{information state} that makes the interaction Markovian while respecting players' information constraints. Since no shared state or joint history is commonly known, recent work has pursued \emph{reductions} that embed a zs-POSG into a fully observable stochastic game and then apply dynamic programming (DP) in the reduced model. Following \citet{sanjari2023isomorphism}, a reduction is \emph{lossless} if it satisfies three criteria: \emph{value preservation} (the expected return of any joint policy is unchanged), \emph{equilibrium correspondence} (original and new equilibria correspond under a surjective mapping), and \emph{information-structure equivalence} (the reduction neither adds spurious information nor collapses strategically relevant distinctions). These criteria ensure that planning in the reduced game is both optimal and game-theoretically interpretable for the original zs-POSG.

Several reductions fall short of this standard or remain computationally intractable. Public-belief reductions exploit rich public signals: in POSGs with public observations, the common signal supports a Markov game on public beliefs \citep{ghosh2004zero,HorBos-aaai19,horak2017heuristic}. The PuB-AMG of \citet{BrownBLG20,SokotaDL0KB23} extend this idea via public policy commitments and a stochastic game over public beliefs, solving a regularised minimax problem. However, the added announcements modify the original information structure, and the resulting strategies can be highly exploitable in the underlying zs-POSG; there is no guarantee of value preservation or exact equilibrium recovery, so the reduction is not lossless for general zs-POSGs. Occupancy Markov games (zs-OMGs) instead adopt a central planner acting on \emph{occupancy states}, posterior distributions over hidden states, and joint histories \citep{Wiggers16,Buffet2020,DelBufDibSaf-DGAA-23}. While this supports Bellman-style DP, the planner chooses joint decision rules as if it observed the latent occupancy state, breaking equilibrium correspondence. Ad hoc policy-tracking can patch execution but does not restore a clean lossless reduction, and, crucially, each backup involves linear programs over exponentially many history-dependent decision rules, leading to exponential complexity in the horizon and information structure.

\citet{escudie2025varepsilonoptimally} resolve the \emph{correctness} side by giving a \emph{lossless} reduction from any zs-POSG to a transition-independent zero-sum stochastic game whose states are occupancy states and player-specific occupancy sets. The reduction preserves value, equilibria, and information structure, and it restores a clean dynamic-programming view: Bellman-style optimality equations become available, so point-based methods such as \texttt{PBVI} can be transferred in a principled way.
The remaining difficulty is \emph{algorithmic}. A direct Bellman backup at the central-planning level is a simultaneous minimax over \emph{history-dependent decision rules}, i.e., a stage game whose rows and columns are entire decision rules rather than primitive actions—an exponential object that is dead on arrival. \citet{escudie2025varepsilonoptimally} avoid this by trading “quantify over all opponent decision rules” for “optimise against an explicit, finite value representation”: the opponent’s continuation values are cached as a family of finite envelopes, and greedy updates are computed by a linear program whose constraints are indexed by cached envelopes and their witness vectors, not by opponent decision rules. The exponential quantification thus disappears from the \emph{formulation}.

However, this trade shifts the burden onto what is stored and how it is optimised. As \texttt{PBVI} iterates, new witness vectors are generated and the cache grows; meanwhile, the greedy step is realised by a \emph{single monolithic} linear program that tightly couples cache indices with the opponent’s action--observation branching. In practice, this coupling makes the program large quickly, so the backup remains the dominant runtime cost even though it is polynomial in explicit representation size. The present work targets precisely this bottleneck by reformulating the central-planner optimisation itself: we replace the monolithic greedy backup with a \emph{sequential} backup that factorises the update into smaller, structured programs, retaining the same lossless foundation while markedly reducing the coupling paid at each iteration.

\section{The Sequential Reformulation}
\label{sec:sequential_approach}

The computational barrier in \emph{simultaneous zs-POSGs} stems from a single source: the Bellman backup operator couples the players through a minimax optimisation over the joint space of decision rules \(\mathcal{D}_{\textcolor{sthlmRed}{1},t}\times \mathcal{D}_{\textcolor{sthlmRed}{2},t}\). Since the size of these spaces is exponential in the history length, solving this matrix game at every step becomes intractable. However, this coupling is a constraint of the \emph{simultaneous formulation}, not intrinsic to the optimal value function itself. To dismantle this barrier, we propose a topological transformation: we expand each simultaneous decision stage into two internal \emph{sub-stages}. Instead of selecting a joint profile in one shot, the planner first selects a decision rule for player~\textcolor{sthlmRed}{1}, and then—conditioned on this choice—selects a decision rule for player~\textcolor{sthlmRed}{2}. Crucially, this sequentialisation occurs only within the planner's optimisation process; the physical game remains simultaneous.

\subsection{Lifting to Sequential Occupancy States}
\label{subsec:sequential_states}

To support this two-step reasoning, we must lift the state space to capture the system's status at these distinct sub-stages. We introduce \emph{sequential occupancy states}: the first ($x_{\textcolor{sthlmRed}{1},t}$) represents the standard simultaneous state, while the second ($x_{\textcolor{sthlmRed}{2},t}$) captures the intermediate state after player~\textcolor{sthlmRed}{1}'s commitment but before player~\textcolor{sthlmRed}{2}'s response.

\begin{definition}%[Sequential Occupancy States]
\label{def:sequential:occupancy:state}
Let \(\theta_{\textcolor{sthlmRed}{1},t} \doteq (d_{\textcolor{sthlmRed}{1},0}, d_{\textcolor{sthlmRed}{2},0}, \dots, d_{\textcolor{sthlmRed}{1},t-1}, d_{\textcolor{sthlmRed}{2},t-1})\) and \(\theta_{\textcolor{sthlmRed}{2},t} \doteq (\theta_{\textcolor{sthlmRed}{1},t}, d_{\textcolor{sthlmRed}{1},t})\) denote the exhaustive information states (histories of joint decision rules) of the sequential planner at sub-stages \((\textcolor{sthlmRed}{1},t)\) and \((\textcolor{sthlmRed}{2},t)\), respectively.
The sequential occupancy states are the posteriors:
\begin{align*}
x_{\textcolor{sthlmRed}{1},t}(s,h_{\textcolor{sthlmRed}{1},t},h_{\textcolor{sthlmRed}{2},t})
&\defeq \Pr(s,h_{\textcolor{sthlmRed}{1},t},h_{\textcolor{sthlmRed}{2},t}\mid \theta_{\textcolor{sthlmRed}{1},t}),\\
x_{\textcolor{sthlmRed}{2},t}(s,h_{\textcolor{sthlmRed}{1},t},h_{\textcolor{sthlmRed}{2},t},a_{\textcolor{sthlmRed}{1}})
&\defeq \Pr(s,h_{\textcolor{sthlmRed}{1},t},h_{\textcolor{sthlmRed}{2},t},a_{\textcolor{sthlmRed}{1}}\mid \theta_{\textcolor{sthlmRed}{2},t}).
\end{align*}
Moreover, the internal transition $(\textcolor{sthlmRed}{1},t)\!\to\!(\textcolor{sthlmRed}{2},t)$ induced by $d_{\textcolor{sthlmRed}{1},t}$ is the deterministic pushforward
\(
x_{\textcolor{sthlmRed}{2},t}(s,h_{\textcolor{sthlmRed}{1},t},h_{\textcolor{sthlmRed}{2},t},a_{\textcolor{sthlmRed}{1}})
=
x_{\textcolor{sthlmRed}{1},t}(s,h_{\textcolor{sthlmRed}{1},t},h_{\textcolor{sthlmRed}{2},t})\, d_{\textcolor{sthlmRed}{1},t}(a_{\textcolor{sthlmRed}{1}}\mid h_{\textcolor{sthlmRed}{1},t})
\).
\end{definition}

\begin{figure*}
 \centering
 \begin{tikzpicture}[->,>=stealth',semithick,auto,node distance=3.75cm, square/.style={regular polygon,regular polygon sides=4}] %
  \tikzstyle{every state}=[draw=black,text=black,inner color= white,outer color= white,draw= black,text=black, drop shadow]
  \tikzstyle{place}=[thick,draw=sthlmBlue,fill=blue!20,minimum size=12mm, opacity=.5]
  \tikzstyle{red place}=[square,place,draw=sthlmRed,fill=sthlmLightRed,minimum size=20mm]
  \tikzstyle{green place}=[diamond,place,draw=sthlmGreen,fill=sthlmLightGreen,minimum size=15mm]

 %-=-=-=-=-=-=-=-=-=-=-=-=-=-=-=-=-=-=-=-=-=-=-=-=
 %	FRAME: discrete time systems
 %-=-=-=-=-=-=-=-=-=-=-=-=-=-=-=-=-=-=-=-=-=-=-=-=

  \node[fill=white, scale=.75] (T) at (-3.5,-3.75) {};
  \node[fill=white, scale=.75] (T0) at (0,-3.75) {$0$};
  \node[fill=white, scale=.75]     (T1) [right of=T0,node distance=3.75cm, fill=white] {$t$};
  \node[fill=white, scale=.75]     (T2) [right of=T1,node distance=3.75cm] {$\textcolor{sthlmRed}{\cdots}$};
  \node[fill=white, scale=.75]     (T3) [right of=T2,fill=white, node distance=3.75cm] {$t+1$};
  \node[fill=white, scale=.75]     (T4) [right of=T3,fill=white, node distance=1.75cm] {$\ldots$};
  \draw[->,>=stealth',semithick,dashed, color=black,anchor=mid] (T3) -- (T4); 
  \draw[->,>=stealth',semithick, color=black, anchor=mid] (T2) -- (T3); 
  \draw[->,>=stealth',semithick, color=black, anchor=mid] (T1) -- (T2); 
  \draw[->,>=stealth',semithick,dashed, color=black,anchor=mid] (T0) -- (T1); 
  \draw[->,>=stealth',semithick,rounded corners, color=black,anchor=mid] (T) node[fill=white, scale=.75]{Time} -- (T0); 
 %

 %-=-=-=-=-=-=-=-=-=-=-=-=-=-=-=-=-=-=-=-=-=-=-=-=
 %	FRAME: Markov chains
 %-=-=-=-=-=-=-=-=-=-=-=-=-=-=-=-=-=-=-=-=-=-=-=-=
 \node[state,place, scale=.75] (S0)            {$x_{\textcolor{sthlmRed}{1},0}$};
  \node[state,place, scale=.75]     (S1) [right of=S0] {$x_{\textcolor{sthlmRed}{1},t}$};
  \node[state,place, scale=.75]     (S2) [ right of=S1] {$x_{\textcolor{sthlmRed}{2},t}$};
  %\node[state,place, scale=.75]     (S23) [right of=S2,node distance=1.875cm] {$\cdots$};
  \node[state,place, scale=.75]     (S3) [ right of=S2] {$x_{\textcolor{sthlmRed}{1},t+1}$};
  \node[, scale=.75]     (S4) [ right of=S3,node distance=1.75cm] 		 {$\cdots$};

 % State path
 \path[sthlmBlue] (S0) edge[dashed] node[midway,text=black,draw=none,scale=.7, above=-5pt] 		{} (S1) 
      (S1) edge node[midway,text=black,draw=none,scale=.7, above=-5pt] 		{} (S2) 
      (S2)   edge[dashed] (S3) 
      (S3) edge[dashed] (S4);

 %-=-=-=-=-=-=-=-=-=-=-=-=-=-=-=-=-=-=-=-=-=-=-=-=
 %	FRAME: controllable Markov chains
 %-=-=-=-=-=-=-=-=-=-=-=-=-=-=-=-=-=-=-=-=-=-=-=-=
{
 \node[, scale=.75]     (A_SLACK) [below of=S0,node distance=2.6cm] {};

 \node[state,red place, scale=.65]     (A0) [below right of=A_SLACK,node distance=2cm] {$d_{\textcolor{sthlmRed}{2},0}$};
 \node[state,red place, scale=.65]     (A1) [right of=A0,node distance=4.25cm]    {$d_{\textcolor{sthlmRed}{1},t}$};
 %\node[state,red place, scale=.65]     (A12) [right of=A1,node distance=2.125cm]    {$\cdots$};
 \node[state,red place, scale=.65]     (A2) [right of=A1,node distance=4.25cm]    {$d_{\textcolor{sthlmRed}{2},t}$};
 \node[, scale=.75]     (A3) [right of=A2,node distance=5.5cm]    {};

 % Decision rule-State path
 \path[sthlmRed] (A0) edge [out=90, in=-155, dashed] node[sthlmGreen, scale=.75] {} (S1)
     (A1) edge [out=90, in=-155] node[sthlmGreen, scale=.75] {} (S2)
     (A2) edge [out=90, in=-155] node[sthlmGreen, scale=.75] {} (S3);   
}

 %-=-=-=-=-=-=-=-=-=-=-=-=-=-=-=-=-=-=-=-=-=-=-=-=
 %	FRAME: partially observable controllable Markov chains
 %-=-=-=-=-=-=-=-=-=-=-=-=-=-=-=-=-=-=-=-=-=-=-=-=
{
 \node[state,green place, scale=.75]     (O0) [below of=S0,node distance=2cm] {$r_{\textcolor{sthlmRed}{1},0}$};
 \node[state,green place, scale=.75]     (O1) [below of=S1,node distance=2cm] {$r_{\textcolor{sthlmRed}{1},t}$};
 \node[state,green place, scale=.75]     (O2) [below of=S2,node distance= 2cm] {$r_{\textcolor{sthlmRed}{2},t}$};
 \node[state,green place, scale=.75]     (O3) [below of=S3,node distance= 2cm] {$r_{\textcolor{sthlmRed}{1},t+1}$};
  \node[, scale=.75]     (O4) [ right of=O3,node distance=1.75cm] 		 {$\cdots$};

  % Decision rule-State-Observation path
 \path[sthlmGreen] (S0) edge  node {} (O0)
 	 (S1) edge  node {} (O1)
      (A0) edge[out=90, in=-180, dashed] node {} (O1)
     	(S2) edge  node {} (O2)
     (A1) edge[out=90, in=-180] node {} (O2)
  	(S3) edge  node {} (O3)
     (A2) edge[out=90, in=-180] node {} (O3);
     }
   
     \node[inner sep=0pt] (prof) at (-3.35,-1.5) {};

  \draw[->,>=stealth',semithick,rounded corners, color=sthlmBlue,anchor=mid] (prof) -- (-3.35,0) --node[draw=none, fill=white, scale=.75] {infers} (S0);	
  {\draw[->,>=stealth',semithick,rounded corners, color=sthlmGreen,anchor=mid] (prof) --node[draw=none, fill=white, scale=.75] {receives} (O0); }	
  { \draw[->,>=stealth',semithick,rounded corners, color=sthlmRed,anchor=mid] (prof) -- (-3.35,-2.85) --node[draw=none, fill=white, scale=.75] {selects} (A0);}	

 \node[inner sep=0pt] (russell) at (-3.25,-1.35) {\includegraphics[width=.1\textwidth]{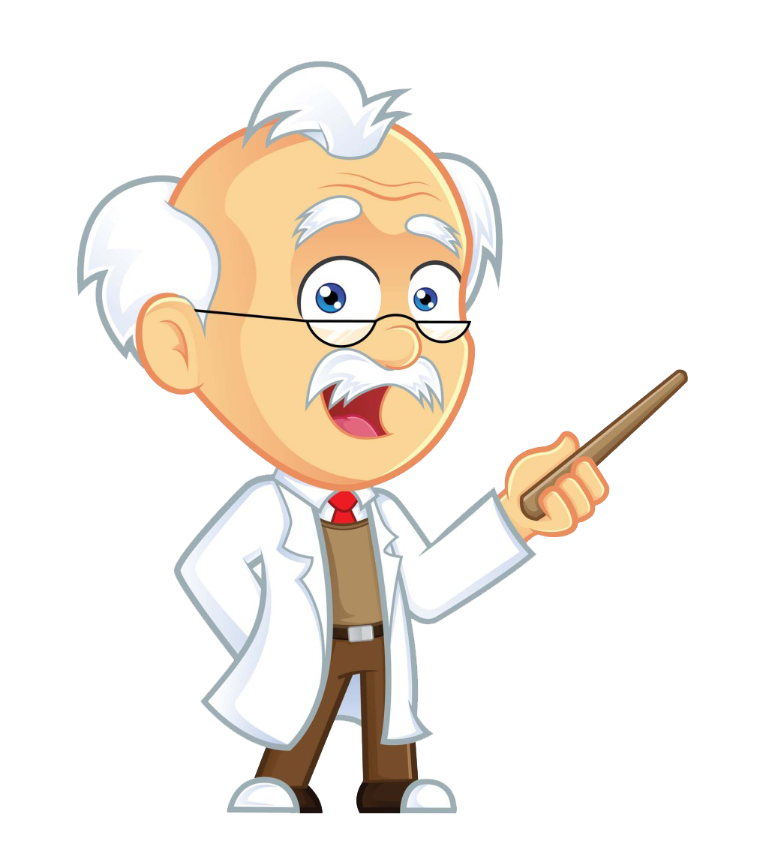}};
 \end{tikzpicture}
\caption{The influence diagram for the sequential occupancy game from the sequential central planner.}
\label{fig:influence}
\end{figure*}
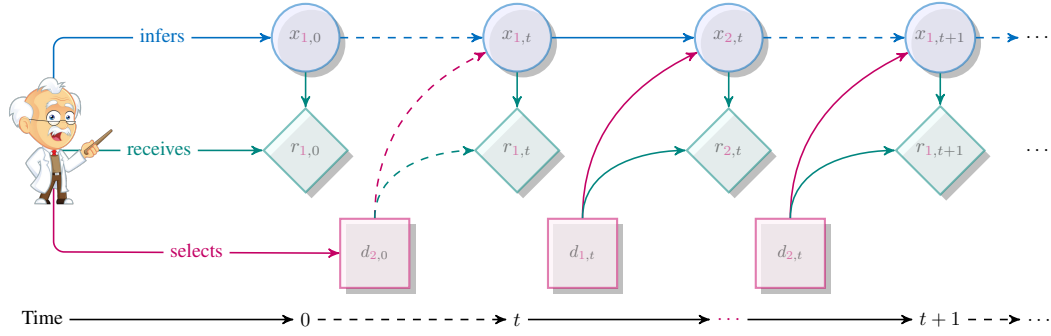

\noindent Figure~\ref{fig:influence} illustrates this topological change. The simultaneous decision node is \emph{unzipped} into the sequence
\[x_{\textcolor{sthlmRed}{1},t} \xrightarrow{d_{\textcolor{sthlmRed}{1}}} x_{\textcolor{sthlmRed}{2},t} \xrightarrow{d_{\textcolor{sthlmRed}{2}}} x_{\textcolor{sthlmRed}{1},t+1}.\]
The intermediate node \(x_{\textcolor{sthlmRed}{2},t}\) is not a new physical state of the world, but a computational snapshot that explicitly includes player~\textcolor{sthlmRed}{1}'s mixed action. We write \(\mathcal{X}_{\mathrm{seq}}\) for the set of all sequential occupancies.

\subsection{Sequential Dynamics and Rewards}
\label{subsec:sequential_primitives}

At the occupancy level, only two kinds of updates can occur, matching the two sub-stages.
At \((\textcolor{sthlmRed}{1},t)\), selecting \(d_{\textcolor{sthlmRed}{1},t}\) only injects player~\textcolor{sthlmRed}{1}'s mixed action into the occupancy; the environment is not yet resolved, so the immediate payoff is zero.
At \((\textcolor{sthlmRed}{2},t)\), selecting \(d_{\textcolor{sthlmRed}{2},t}\) completes the original simultaneous stage: the joint action is resolved, the true reward is realised, and the next stage-boundary occupancy is produced.

We denote these two updates by occupancy-level primitives \(\tau_{\mathrm{seq}}\) (transition) and \(\rho_{\mathrm{seq}}\) (reward).
Concretely, for the first sub-stage we set \(\rho_{\mathrm{seq}}(x_{\textcolor{sthlmRed}{1},t},d_{\textcolor{sthlmRed}{1},t})\equiv 0\) and define \(\tau_{\mathrm{seq}}(x_{\textcolor{sthlmRed}{1},t},d_{\textcolor{sthlmRed}{1},t}) \doteq x_{\textcolor{sthlmRed}{2},t}\) via the deterministic pushforward.
For the second sub-stage, we define \(\tau_{\mathrm{seq}}(x_{\textcolor{sthlmRed}{2},t},d_{\textcolor{sthlmRed}{2},t}) \doteq x_{\textcolor{sthlmRed}{1},t+1}\) via the standard environment kernel, and the reward \(\rho_{\mathrm{seq}}(x_{\textcolor{sthlmRed}{2},t},d_{\textcolor{sthlmRed}{2},t}) \doteq \mathbb{E}[r_t\mid x_{\textcolor{sthlmRed}{2},t},d_{\textcolor{sthlmRed}{2},t}]\).

\begin{definition}[Sequential Occupancy Game]
\label{def:seq_omg_occ}
The sequential reformulation induces a \emph{sequential occupancy Markov game}
\(\mathcal{M}_{\mathrm{seq}} \defeq (\mathcal{X}_{\mathrm{seq}},\mathcal{F}_{\mathrm{seq}}, \mathcal{D}, \tau_{\mathrm{seq}}, \rho_{\mathrm{seq}}, \gamma, \ell)\),
where \(\mathcal{D} \doteq \mathcal{D}_{\textcolor{sthlmRed}{1}} \cup \mathcal{D}_{\textcolor{sthlmRed}{2}}\) is the combined space of decision rules acting on the sequential state space.
\end{definition}

\subsection{Recovering Simultaneity via Two-Step Recursion}
\label{subsec:sequential_bellman}

The central property of this reformulation is that composing the two sequential updates perfectly reconstructs the original simultaneous backup.
We define the sequential optimal value function \(v^{*}_{\mathrm{seq}}\) via standard Bellman equations on the unzipped graph, using a stage-dependent operator \(\mathrm{opt}_{\textcolor{sthlmRed}{i}}\) (defined as \(\max\) for player~\textcolor{sthlmRed}{1} and \(\min\) for player~\textcolor{sthlmRed}{2}).

\begin{theorem}[Sequential Optimality Equations]
\label{thm:bellman_seq}
For any transient sub-stage state \(x_{\textcolor{sthlmRed}{i}}\in\mathcal{X}_{\mathrm{seq}}\),
\[
v^{*}_{\mathrm{seq}}(x_{\textcolor{sthlmRed}{i}})
=
\mathrm{opt}_{\textcolor{sthlmRed}{i},\,d_{\textcolor{sthlmRed}{i}}\in\mathcal{D}_{\textcolor{sthlmRed}{i}}}
\left[
\rho_{\mathrm{seq}}(x_{\textcolor{sthlmRed}{i}}, d_{\textcolor{sthlmRed}{i}})
+\gamma_{\textcolor{sthlmRed}{i}}\,
v^{*}_{\mathrm{seq}}(\tau_{\mathrm{seq}}(x_{\textcolor{sthlmRed}{i}}, d_{\textcolor{sthlmRed}{i}}))
\right],
\]
where \(\gamma_{\textcolor{sthlmRed}{1}}=1\) and \(\gamma_{\textcolor{sthlmRed}{2}}=\gamma\).
\end{theorem}

This theorem makes the structural change explicit: the coupled saddle-point update is replaced by two sub-stage optimisations. By composing these updates across \((\textcolor{sthlmRed}{1},t)\) and \((\textcolor{sthlmRed}{2},t)\), we recover the simultaneous backup at stage boundaries: \(v^{*}_{\mathrm{seq}}(x_{\textcolor{sthlmRed}{1}})=v^{*}_{\mathrm{sim}}(x_{\textcolor{sthlmRed}{1}})\).

\begin{theorem}[Sufficiency]
\label{thm:sufficiency:sequential:occupancy:state}
Sequential occupancy states are sufficient to compute optimal values in \(\mathcal{M}_{\mathrm{seq}}\) and hence in the simultaneous zs-POSG \(\mathcal{M}\).
\end{theorem}

\paragraph{Roadmap.} The remainder of the paper builds on this reformulation. In Section \ref{sec:geometry_poly_backup}, we exploit the geometry of \(v^{*}_{\mathrm{seq}}\) to derive a polynomial-time backup operator. In Section \ref{sec:sequential_ti_sg}, we ensure losslessness by extending the transition-independent reduction to this sequential setting.
\section{Sequential Backup Geometry}
\label{sec:geometry_poly_backup}

The sequential reformulation removes the minimax coupling by splitting each stage into two internal sub-stages.
This section explains why that split matters computationally, beyond being a mere modelling convenience.
In the simultaneous view, the Bellman backup hides a nested pessimistic computation inside the saddle point, and that hidden computation must be reconstructed from scratch every time the backup is applied.
The sequential view makes this pessimistic layer visible at the intermediate occupancy and turns it into a first-class object that can be computed once, stored, and reused across backups.
This caching viewpoint is the central message of the section.

Concretely, we establish two outcomes.
First, we make the shape of the optimal value explicit as a two-layer envelope: a local pessimistic envelope revealed at the intermediate occupancy, and an outer selection over non-dominated continuations.
Second, we leverage this explicit envelope structure to obtain a sequential backup that decomposes into two linear programs, one per sub-stage.
The point is not that the game becomes simpler, but that the expensive part of the saddle-point update becomes an object we can pre-compile.

\subsection{Saddle Geometry Unveiled}
\label{sec:geometry_two_stage}

The intermediate occupancy $x_{\textcolor{sthlmRed}{2}}$ is the key to making the hidden pessimistic layer explicit.
By transitioning from the stage boundary to $x_{\textcolor{sthlmRed}{2}}$, the planner effectively freezes player~\textcolor{sthlmRed}{1}'s mixed decision rule.
This creates the computational leverage: player~\textcolor{sthlmRed}{2}'s best response no longer needs to be chosen as a single global object.
Instead, it separates across private histories, allowing pessimism to be evaluated locally.
Crucially, this pessimistic evaluation is computed once per candidate envelope and \emph{reused} as a static constraint to evaluate any candidate decision rule \(d_{\textcolor{sthlmRed}{1}}\) (as formalised by \ref{eq:LP1} in \Cref{sec:disentangled_greedy}), effectively decoupling the decision variables within the backup.

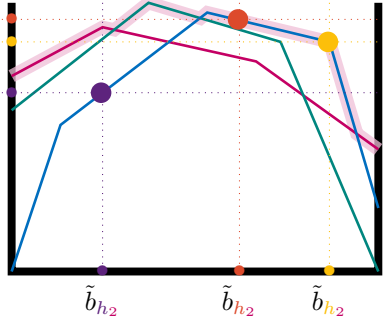
\begin{figure}[h]
\centering
\begin{tikzpicture}[
 scale=1.3,
 IS/.style={sthlmRed, thick},
 LM/.style={sthlmRed, thick},
 axis/.style={very thick, ->, >=stealth', line join=miter},
 important line/.style={thick},
 dashed line/.style={dashed, thin},
 every node/.style={color=black},
 dot/.style={circle,fill=red,minimum size=8pt,inner sep=0pt, outer sep=-1pt},
 ]

 %%% The axis
 \draw[axis, -, line width=.3em] (-6,2.75) -- (-6,0) -- (-2.25,0) -- (-2.25,2.75);

 %%% The exact upper envelope for the left chart
 { \draw[-, sthlmRed!25, line width=0.5em, opacity=0.75]
(-6,2) --
(-5.07,2.5) --
(-4.949,2.473) --
(-4.6, 2.75) --
(-4.07,2.59) --
(-4,2.65) --
(-2.75,2.35) --
(-2.47,1.41) --
(-2.25,1.25);
 }
 	
 %%% The piecewise linear concave functions
 \draw[-, sthlmRed, line width=.1em] (-6,2) -- (-5.07, 2.5) -- (-3.5, 2.15) -- (-2.25, 1.25);
 \draw[-, sthlmGreen, line width=.1em] (-6,1.65) -- (-4.6, 2.75) -- (-3.25, 2.35) -- (-2.25,0);
 \draw[-, sthlmBlue, line width=.1em] (-6,0) -- (-5.5,1.5) -- (-4,2.65) -- (-2.75,2.35) -- (-2.25,0.65);

 %%% The points on the marginal occupancy states
 % MODIFICATION HERE: Using \tilde{b} instead of b
 \node[scale=1, sthlmYellow] at (-2.75,-0.3) {${\color{black}\tilde{b}}_{h_{\textcolor{sthlmRed}{2}}}$};
 \node[scale=1, sthlmPurple] at (-5.07,-0.3) {${\color{black}\tilde{b}}_{h_{\textcolor{sthlmRed}{2}}}$};
 \node[scale=1, sthlmOrange] at (-3.67,-0.3) {${\color{black}\tilde{b}}_{h_{\textcolor{sthlmRed}{2}}}$};

 %%% The projections on the value function (left chart)
 \draw[-, dotted, sthlmYellow] (-2.75,2.75) -- (-2.75,0);
 \draw[-, dotted, sthlmYellow] (-6.1,2.35) -- (-2.25,2.35);
 \node[rotate=90, scale=2, sthlmYellow] at (-2.75,2.35) {$\bullet$};
 \node[scale=1, sthlmYellow] at (-2.75,0) {$\bullet$};
 \node[scale=1, sthlmYellow] at (-6,2.35) {$\bullet$};

 \draw[-, dotted, sthlmOrange] (-3.67,2.75) -- (-3.67,0);
 \draw[-, dotted, sthlmOrange] (-6.1,2.58) -- (-2.25,2.58);
 \node[rotate=90, scale=2, sthlmOrange] at (-3.67,2.58) {$\bullet$};
 \node[scale=1, sthlmOrange] at (-3.67,0) {$\bullet$};
 \node[scale=1, sthlmOrange] at (-6,2.58) {$\bullet$};

 \draw[-, dotted, sthlmPurple] (-5.07,2.75) -- (-5.07,0);
 \draw[-, dotted, sthlmPurple] (-6.1,1.83) -- (-2.25,1.83);
 \node[rotate=90, scale=2, sthlmPurple] at (-5.07,1.83) {$\bullet$};
 \node[scale=1, sthlmPurple] at (-5.07,0) {$\bullet$};
 \node[scale=1, sthlmPurple] at (-6,1.83) {$\bullet$};
 \end{tikzpicture}
\caption{\textbf{The Max-of-Concave Geometry.} This diagram visualizes the value function structure projected onto the probability simplex. The horizontal axis represents the space of \emph{normalized local beliefs} \(\tilde{b}_{h_{\textcolor{sthlmRed}{2}}} \doteq b_{h_{\textcolor{sthlmRed}{2}}} / \|b_{h_{\textcolor{sthlmRed}{2}}}\|_1\).
The thick highlighted curve corresponds to the upper envelope (optimism) over a family of concave functions (pessimism).}
\label{fig:uniform:continuity}
\end{figure}

\paragraph{Decomposition into local belief vectors.}
The separation above is driven by a simple statistic: for each private history \(h_{\textcolor{sthlmRed}{2}}\), we extract from the global occupancy the local weighted belief \(b_{h_{\textcolor{sthlmRed}{2}}}\) that governs player~\textcolor{sthlmRed}{2}'s pessimistic evaluation.
Crucially, this belief must capture the full latent state, which includes the opponent's private history \(h_{\textcolor{sthlmRed}{1}}\).
Let \(\Omega_{\textcolor{sthlmRed}{1}} \doteq \mathcal{S} \times \mathcal{H}_{\textcolor{sthlmRed}{1}}\) and \(\Omega_{\textcolor{sthlmRed}{2}} \doteq \mathcal{S} \times \mathcal{H}_{\textcolor{sthlmRed}{1}} \times \mathcal{A}_{\textcolor{sthlmRed}{1}}\) denote the domains of uncertainty relevant to player~\textcolor{sthlmRed}{2} at sub-stages \(1\) and \(2\) respectively.
For a given occupancy \(x_{\textcolor{sthlmRed}{i}}\), the local belief \(b_{h_{\textcolor{sthlmRed}{2}}} \in \mathbb{R}_+^{\Omega_{\textcolor{sthlmRed}{i}}}\) is defined as the \emph{unnormalized slice} of the occupancy:
\[
b_{h_{\textcolor{sthlmRed}{2}}}(\omega) \doteq x_{\textcolor{sthlmRed}{i}}(\omega, h_{\textcolor{sthlmRed}{2}}), \quad \text{for } \omega \in \Omega_{\textcolor{sthlmRed}{i}}.
\]
This vector encodes both the conditional probability and the relative weight of the history.
Consequently, summing the inner products \(\langle b_{h_{\textcolor{sthlmRed}{2}}}, \alpha \rangle\) correctly absorbs the marginal probability \(\Pr_{x_{\textcolor{sthlmRed}{i}}}(H_{\textcolor{sthlmRed}{2}}=h_{\textcolor{sthlmRed}{2}})\).

\paragraph{The Two-Layer Envelope.}
This decomposition reveals a universal geometry common to both sub-stages.
Fix an arbitrary sub-stage \(\textcolor{sthlmRed}{i} \in \{1, 2\}\).
The value function behaves as a \emph{two-layer envelope} reflecting the adversarial hierarchy.
At the \emph{inner layer}, any fixed future strategy for player~\textcolor{sthlmRed}{1} induces a collection \(\Gamma_{\!\textcolor{sthlmRed}{2}} \subset \mathbb{R}^{\Omega_{\textcolor{sthlmRed}{i}}}\) that forms a ``tent'' of hyperplanes—a piecewise-linear concave function representing player~\textcolor{sthlmRed}{2}'s \emph{pessimistic} response against local beliefs.
At the \emph{outer layer}, player~\textcolor{sthlmRed}{1} optimises over a family \(\Gamma_{\!\textcolor{sthlmRed}{1}}\) of such collections, establishing the global value as the \emph{optimistic} pointwise maximum over these concave tents.
The following theorem formalises this \emph{Max-of-Concave} saddle geometry (Fig. \ref{fig:uniform:continuity}).

\begin{theorem}
\label{thm:geometry}
Fix a sub-stage \(\textcolor{sthlmRed}{i} \in \{1, 2\}\).
Let \(\Gamma_{\!\textcolor{sthlmRed}{1}}\) be a family of finite sets representing non-dominated future strategies, where each \(\Gamma_{\!\textcolor{sthlmRed}{2}} \in \Gamma_{\!\textcolor{sthlmRed}{1}}\) is a subset of \(\mathbb{R}^{\Omega_{\textcolor{sthlmRed}{i}}}\).
For any occupancy \(x_{\textcolor{sthlmRed}{i}}\) at this sub-stage, the optimal sequential value is:
\begin{align*}
v^*_{\mathrm{seq}}(x_{\textcolor{sthlmRed}{i}})
&=
\textstyle
\max_{\Gamma_{\!\textcolor{sthlmRed}{2}} \in \Gamma_{\!\textcolor{sthlmRed}{1}}}
\mathtt{Val}_{\Gamma_{\!\textcolor{sthlmRed}{2}}}(x_{\textcolor{sthlmRed}{i}})
\\
\mathtt{Val}_{\Gamma_{\!\textcolor{sthlmRed}{2}}}(x_{\textcolor{sthlmRed}{i}})
&\doteq
\textstyle
\sum_{h_{\textcolor{sthlmRed}{2}}}
\min_{\alpha \in \Gamma_{\!\textcolor{sthlmRed}{2}}}
\langle b_{h_{\textcolor{sthlmRed}{2}}}, \alpha \rangle.
\end{align*}
\end{theorem}

\subsection{Disentangled Greedy Selection}
\label{sec:disentangled_greedy}

The envelope characterisations yield two decoupled LPs—one per sub-stage.
The sequential split turns the adversary’s inner pessimisation into reusable linear constraints; a variable/constraint count then exposes an exponential gap relative to the simultaneous matrix-game backup.

\paragraph{Optimizing \(d_{\textcolor{sthlmRed}{1}}\) against cached envelopes.}
\ref{eq:LP1} represents the greedy backup of the maximiser against a \emph{compiled} envelope.
Unlike simultaneous backups that must re-solve player~\textcolor{sthlmRed}{2}'s response inside every evaluation, the sequential formulation evaluates \(d_{\textcolor{sthlmRed}{1}}\) directly against the cached envelopes produced at the intermediate stage.
Since the immediate reward is zero, player~\textcolor{sthlmRed}{1} maximizes the continuation value of the induced intermediate occupancy.

\begin{lemma}
\label{lem:greedy_x1}
Fix a finite family \(\Gamma_{\!\textcolor{sthlmRed}{1}} \) of collections \(\Gamma_{\!\textcolor{sthlmRed}{2}} \) cached from the intermediate stage \((\textcolor{sthlmRed}{2},t)\).
The greedy decision rule of player~\textcolor{sthlmRed}{1} at occupancy \(x_{\textcolor{sthlmRed}{1}}\) is:
\begin{align*}
\textstyle
\argmax_{d_{\textcolor{sthlmRed}{1}}}
\max_{\Gamma_{\!\textcolor{sthlmRed}{2}} \in \Gamma_{\!\textcolor{sthlmRed}{1}}}
\mathtt{Val}_{\Gamma_{\!\textcolor{sthlmRed}{2}}}(\tau_{\mathrm{seq}}(x_{\textcolor{sthlmRed}{1}}, d_{\textcolor{sthlmRed}{1}})).
\end{align*}
\end{lemma}

Operationally, because \(\Gamma_{\!\textcolor{sthlmRed}{1}}\) is a finite collection, we solve one instance of \ref{eq:LP1} for each candidate envelope \(\Gamma_{\!\textcolor{sthlmRed}{2}} \in \Gamma_{\!\textcolor{sthlmRed}{1}}\), and retain the pair \((d_{\textcolor{sthlmRed}{1}}^*, \Gamma_{\!\textcolor{sthlmRed}{2}}^*)\) that yields the global maximum.
Recall that \(b_{h_{\textcolor{sthlmRed}{2}}}\) is an unnormalised weighted belief, so the objective correctly aggregates values according to the history probability.
\begin{equation}
\tag{LP$_1$}
\label{eq:LP1}
\begin{aligned}
&\text{maximise}_{d_{\textcolor{sthlmRed}{1}}, w} \quad \textstyle \sum_{h_{\textcolor{sthlmRed}{2}}} w_{h_{\textcolor{sthlmRed}{2}}} \\
&\text{subject to} \quad w_{h_{\textcolor{sthlmRed}{2}}} \le
\langle b_{h_{\textcolor{sthlmRed}{2}}}, \alpha \rangle,
\quad \forall h_{\textcolor{sthlmRed}{2}}, \forall \alpha \in \Gamma_{\!\textcolor{sthlmRed}{2}}.
\end{aligned}
\end{equation}
Crucially, the decision variables \(d_{\textcolor{sthlmRed}{1}}\) are embedded within the belief vector \(b_{h_{\textcolor{sthlmRed}{2}}}\).
The term \(\langle b_{h_{\textcolor{sthlmRed}{2}}}, \alpha \rangle\) expands to a linear expression in \(d_{\textcolor{sthlmRed}{1}}\):
\[
\langle b_{h_{\textcolor{sthlmRed}{2}}}, \alpha \rangle
=
\textstyle\sum_{s, h_{\textcolor{sthlmRed}{1}}, a_{\textcolor{sthlmRed}{1}}}
x_{\textcolor{sthlmRed}{1}}(s, h_{\textcolor{sthlmRed}{1}}, h_{\textcolor{sthlmRed}{2}}) \,
d_{\textcolor{sthlmRed}{1}}(a_{\textcolor{sthlmRed}{1}} \mid h_{\textcolor{sthlmRed}{1}}) \,
\alpha(s, a_{\textcolor{sthlmRed}{1}}).
\]
Consequently, \ref{eq:LP1} treats the adversary as a fixed set of hyperplanes evaluated at beliefs that vary linearly with \(d_{\textcolor{sthlmRed}{1}}\).

\paragraph{Optimizing \(d_{\textcolor{sthlmRed}{2}}\) against cached envelopes.}
Player~\textcolor{sthlmRed}{2} acts to minimise the immediate cost plus the expected future value. Because the future value is encoded by a family of envelopes \(\Gamma_{\!\textcolor{sthlmRed}{1}}'\), the optimal response involves selecting the most pessimistic envelope mixture.
We express this greedy backup in the dual form to avoid enumerating the decision-rule vertices.

\begin{lemma}
\label{lem:greedy_x2}
Fix a family \(\Gamma_{\!\textcolor{sthlmRed}{1}}'\) cached from the next stage \((\textcolor{sthlmRed}{1},t+1)\).
The greedy decision rule for player~\textcolor{sthlmRed}{2} corresponds to the solution of the dual maximisation:
\begin{align*}
\textstyle
\max_{\lambda \in \Delta(\Gamma_{\!\textcolor{sthlmRed}{1}}')}
\min_{d_{\textcolor{sthlmRed}{2}}}
\sum_{\Gamma_{\!\textcolor{sthlmRed}{2}}' \in \Gamma_{\!\textcolor{sthlmRed}{1}}'}
\lambda(\Gamma_{\!\textcolor{sthlmRed}{2}}') \,
\mathtt{Val}_{\Gamma_{\!\textcolor{sthlmRed}{2}}'}(\tau_{\mathrm{seq}}(x_{\textcolor{sthlmRed}{2}}, d_{\textcolor{sthlmRed}{2}})).
\end{align*}
\end{lemma}

The dual variable \(\lambda\) acts as a selector for active future envelopes.
To handle the branching over observations without exploding the complexity, we define an \emph{observation-shifted vector}.
For any future vector \(\alpha \in \Gamma_{\!\textcolor{sthlmRed}{2}}'\), action \(a_{\textcolor{sthlmRed}{2}}\), and observation \(z_{\textcolor{sthlmRed}{2}}\), let \(q_{a_{\textcolor{sthlmRed}{2}}, z_{\textcolor{sthlmRed}{2}}}^{\alpha}\) be the vector back-projected to the intermediate state:
\(
q_{a_{\textcolor{sthlmRed}{2}}, z_{\textcolor{sthlmRed}{2}}}^{\alpha}(s, h_{\textcolor{sthlmRed}{1}}, a_{\textcolor{sthlmRed}{1}})
\doteq
\frac{r(s, a_{\textcolor{sthlmRed}{1}}, a_{\textcolor{sthlmRed}{2}})}{|\mathcal{Z}_{\textcolor{sthlmRed}{2}}|} + 
  \gamma \sum_{s', z_{\textcolor{sthlmRed}{1}}} p(s', z_{\textcolor{sthlmRed}{1}}, z_{\textcolor{sthlmRed}{2}} \mid s, a_{\textcolor{sthlmRed}{1}}, a_{\textcolor{sthlmRed}{2}}) \, \alpha(s', h_{\textcolor{sthlmRed}{1}} \!\cdot\! a_{\textcolor{sthlmRed}{1}} \!\cdot\! z_{\textcolor{sthlmRed}{1}})
\).
Note the first term distributes the immediate reward uniformly across observation branches to ensure it is counted exactly once in the total summation.
\ref{eq:LP2} computes the optimal response by summing contributions over all observation outcomes.
The variables \(w_{h_{\textcolor{sthlmRed}{2}}, z_{\textcolor{sthlmRed}{2}}}^{\Gamma_{\!\textcolor{sthlmRed}{2}}'}\) represent the value contribution of a specific envelope \(\Gamma_{\!\textcolor{sthlmRed}{2}}'\) given observation \(z_{\textcolor{sthlmRed}{2}}\).

\begin{equation}
\tag{LP$_2$}
\label{eq:LP2}
\begin{aligned}
&\text{maximise}_{\lambda, w, v} \quad \textstyle \sum_{h_{\textcolor{sthlmRed}{2}}} v_{h_{\textcolor{sthlmRed}{2}}} \\
&\text{subject to} \\
% Constraint 1: Optimality of a2
&\quad \textstyle v_{h_{\textcolor{sthlmRed}{2}}} \le
\sum_{\Gamma_{\!\textcolor{sthlmRed}{2}}' \in \Gamma_{\!\textcolor{sthlmRed}{1}}'}
\sum_{z_{\textcolor{sthlmRed}{2}}}
w_{h_{\textcolor{sthlmRed}{2}}, z_{\textcolor{sthlmRed}{2}}}^{\Gamma_{\!\textcolor{sthlmRed}{2}}'},
\quad \textstyle \forall h_{\textcolor{sthlmRed}{2}}, \forall a_{\textcolor{sthlmRed}{2}}, \\
% Constraint 2: Concavity
&\quad \textstyle w_{h_{\textcolor{sthlmRed}{2}}, z_{\textcolor{sthlmRed}{2}}}^{\Gamma_{\!\textcolor{sthlmRed}{2}}'} \le
\lambda(\Gamma_{\!\textcolor{sthlmRed}{2}}') \,
\langle b_{h_{\textcolor{sthlmRed}{2}}}, q_{a_{\textcolor{sthlmRed}{2}}, z_{\textcolor{sthlmRed}{2}}}^{\alpha} \rangle, \\
&\quad  \forall h_{\textcolor{sthlmRed}{2}}, \forall a_{\textcolor{sthlmRed}{2}}, \forall z_{\textcolor{sthlmRed}{2}},
\forall \Gamma_{\!\textcolor{sthlmRed}{2}}' \in \Gamma_{\!\textcolor{sthlmRed}{1}}',
\forall \alpha \in \Gamma_{\!\textcolor{sthlmRed}{2}}', \\
% Constraint 3: Probability Simplex
&\quad \textstyle \sum_{\Gamma_{\!\textcolor{sthlmRed}{2}}' \in \Gamma_{\!\textcolor{sthlmRed}{1}}'} \lambda(\Gamma_{\!\textcolor{sthlmRed}{2}}') = 1,
\quad \lambda \ge 0.
\end{aligned}
\end{equation}
At optimum, the maximisation saturates the constraints, recovering the mixture value.
Crucially, this formulation avoids enumerating decision rules; the complexity is now polynomial in the game's representation parameters.

\begin{theorem}
\label{thm:complexity}
The sequential backup operator is polynomial in the explicit representation size. Computing the optimal update entails:
{(i) at sub-stage \(\textcolor{sthlmRed}{1}\)}, solving \(|\Gamma_{\!\textcolor{sthlmRed}{1}}|\) instances of \ref{eq:LP1} with \(\pmb{O}(|\mathcal{H}_{\textcolor{sthlmRed}{1}}||\mathcal{A}_{\textcolor{sthlmRed}{1}}|)\) variables and \(\pmb{O}(|\mathcal{H}_{\textcolor{sthlmRed}{2}}||\Gamma_{\!\textcolor{sthlmRed}{2}}|)\) constraints;
{(ii) at sub-stage \(\textcolor{sthlmRed}{2}\)}, solving one instance of \ref{eq:LP2} with \(\pmb{O}(|\mathcal{H}_{\textcolor{sthlmRed}{2}}| |\mathcal{Z}_{\textcolor{sthlmRed}{2}}| |\Gamma_{\!\textcolor{sthlmRed}{1}}'|)\) variables and \(\pmb{O}(|\mathcal{H}_{\textcolor{sthlmRed}{2}}| |\mathcal{A}_{\textcolor{sthlmRed}{2}}| |\mathcal{Z}_{\textcolor{sthlmRed}{2}}| |\Gamma_{\!\textcolor{sthlmRed}{1}}'| |\bar{\Gamma}_{\!\textcolor{sthlmRed}{2}}'|)\) constraints, where \(|\bar{\Gamma}_{\!\textcolor{sthlmRed}{2}}'| \doteq \max_{\Gamma'\in \Gamma_{\!\textcolor{sthlmRed}{1}}'} |\Gamma'|\).
\end{theorem}

In contrast, the simultaneous backup scales with the size of the decision-rule spaces \(|\mathcal{D}_{\textcolor{sthlmRed}{1}}|\) and \(|\mathcal{D}_{\textcolor{sthlmRed}{2}}|\), which grow exponentially with the size of the private-history spaces \(|\mathcal{H}_{\textcolor{sthlmRed}{i}}|\).
The sequential reformulation removes this strategy-enumeration bottleneck by never indexing the optimisation by opponent decision rules:
pessimism is enforced only against finitely many cached supporting vectors, so the runtime is governed by the explicit envelope representation that is actually stored and manipulated.
Appendix~\ref{app:sec:complexity} provides the detailed input-size analysis by counting scalar variables and linear constraints in LP$_{\textcolor{sthlmRed}{1}}$ and LP$_{\textcolor{sthlmRed}{2}}$ and appealing to polynomial-time solvability of linear programs.

Relative to the cached-envelope greedy update of \citet{escudie2025varepsilonoptimally}, the gain is structural rather than foundational: both approaches are polynomial in the same cached representation objects, but they pay for them differently.
\citet{escudie2025varepsilonoptimally} realise the greedy step as one monolithic LP in which the cache indices and the opponent’s action--observation branching are coupled throughout the constraint family.
Our sequential backup factorises this coupling into two stages: a cache-local optimisation (LP$_{\textcolor{sthlmRed}{1}}$) followed by a single global aggregation (LP$_{\textcolor{sthlmRed}{2}}$).
Operationally, this means we pay the action--observation branching once—at sub-stage~\(\textcolor{sthlmRed}{2}\)—rather than inside every per-envelope optimisation, yielding a lighter greedy update while preserving the same lossless guarantees.

\paragraph{Witness-based cache updates.}
Our value representation is a finite cache of envelopes, combined by pointwise maximisation.
Crucially, each greedy step produces a certificate (a witness vector at sub-stage~\textcolor{sthlmRed}{1}, or a witness envelope at sub-stage~\textcolor{sthlmRed}{2} via the LP dual) that can be appended to the cache.
This yields a monotone augmentation of the represented value function, with strict improvement whenever the new candidate improves the value at the generating occupancy.
See Appendix~\ref{sec:updating_value_function_representations} for the formal constructions and proof.

\subsection{Algorithmic Realisation: Sequential PBVI}

The geometric insights and polynomial backups developed above are not standalone; they serve as the computational engine for a Point-Based Value Iteration (PBVI) algorithm.
Since the sequential reduction maps the original problem to a sequential occupancy Markov game, we can directly apply the standard \texttt{PBVI} schema used in simultaneous settings.
In the sequential model, however, the dynamic program is indexed by \emph{sub-stages}:
we write $\nu=(i,t)$ with $i\in\{\textcolor{sthlmRed}{1},\textcolor{sthlmRed}{2}\}$ and $t\in\{0,\ldots,\ell\}$, and we let $\mathtt{next}(\nu)$ denote the successor sub-stage.
The algorithm therefore maintains (i) a finite sample set $\mathcal X_\nu$ of sequential occupancies for each sub-stage $\nu$, and (ii) a family of envelopes $\Gamma_{\!\textcolor{sthlmRed}{1},\nu}$ representing the value function at that sub-stage.
%For notational convenience, we also denote by $\Gamma_{\!\textcolor{sthlmRed}{1}}'$ the cache queried at the successor sub-stage, i.e., $\Gamma_{\!\textcolor{sthlmRed}{1}}' \equiv \Gamma_{\!\textcolor{sthlmRed}{1},\mathtt{next}(\nu)}$ in Algorithm~\ref{pbvi:zsposg}.

The integration is seamless: at each iteration, rather than solving a large matrix game, the planner executes the sequential backup \emph{at every sampled sub-stage}.
This is implemented by the backward sweep over sub-stage $\nu$ in reverse order, from $(\textcolor{sthlmRed}{1},\ell)$ down to $(\textcolor{sthlmRed}{1},0)$.
For each $\nu$ and each sampled occupancy $x_\nu\in\mathcal X_\nu$, the routine \texttt{Improve} solves the LP associated with the active player at \(\nu\) using the cached \emph{successor} envelope $\Gamma_{\!\textcolor{sthlmRed}{1},\mathtt{next}(\nu)}$ and then augments the current cache $\Gamma_{\!\textcolor{sthlmRed}{1},\nu}$ with the newly generated witness envelope (Theorem~\ref{thm:monotone_gamma_update}).
Concretely, when $\nu=(\textcolor{sthlmRed}{2},t)$, \texttt{Improve} solves the dual program~\eqref{eq:LP2} against the successor cache \(\Gamma_{\!\textcolor{sthlmRed}{1},\mathtt{next}(\nu)}\), which computes the optimal player~$\textcolor{sthlmRed}{2}$ response against $\Gamma_{\!\textcolor{sthlmRed}{1},\mathtt{next}(\nu)}$ and thereby generates an intermediate envelope $\Gamma_{\!\textcolor{sthlmRed}{2},\nu}$.
When $\nu=(\textcolor{sthlmRed}{1},t)$, \texttt{Improve} solves~\eqref{eq:LP1} against the corresponding intermediate successor envelope $\Gamma_{\!\textcolor{sthlmRed}{2},\mathtt{next}(\nu)}$, producing a witness vector (and its envelope) used to augment $\Gamma_{\!\textcolor{sthlmRed}{1},\nu}$ (Theorem~\ref{thm:monotone_gamma_update}).

After this backward improvement sweep, the algorithm performs a forward expansion sweep over sub-stages (second loop in Algorithm~\ref{pbvi:seq}).
Starting from the current sampled sets $\{\mathcal X_\nu\}_\nu$, the routine \texttt{expand} propagates reachability to populate the successor sample sets $\mathcal X_{\mathtt{next}(\nu)}$.
To ensure scalability, we employ standard bounded pruning techniques to discard envelopes (or vectors within envelopes) that are not selected on the sampled occupancies, preventing representation growth from exploding.
Overall, this sequential \texttt{PBVI} inherits the convergence guarantees of the simultaneous \texttt{PBVI} schema, while typically executing each backup substantially faster thanks to the polynomial LP factorisation. We refer the reader to Appendix~\ref{sec:algorithms-seq} for the complete pseudocode and detailed pruning specifications.

\section{Sequential Transition Independence}
\label{sec:sequential_ti_sg}

The previous sections show that the sequential split yields a polynomial-time Bellman update at the level of sequential occupancies.
What remains is to ensure that this computational device sits inside a game-theoretically sound model: values and equilibrium structure must be preserved exactly, as in the transition-independent reduction of \citet{escudie2025varepsilonoptimally}.
This section therefore introduces a \emph{sequential} variant of the TI-zs-SG and states the corresponding losslessness result.
All the structural and algorithmic consequences established in \citet{escudie2025varepsilonoptimally} then carry over verbatim; in particular, policy extraction is performed at the level of player occupancy sets (as in that work), not by the LPs used to implement the sequential backup.

\subsection{The Sequential TI-zs-SG}

Recall that, in the simultaneous TI-zs-SG of \citet{escudie2025varepsilonoptimally}, the game is organised around two layers:
(i) a central occupancy state used to evaluate values, and
(ii) player-specific occupancy sets used to represent each player’s feasible private-information evolution and to support equilibrium-consistent policy extraction.
The sequential split refines only the \emph{timing} of the central backup: each stage is decomposed into two sub-stages, so that player~\textcolor{sthlmRed}{1} updates first and player~\textcolor{sthlmRed}{2} completes the stage, exactly as in Sections~\ref{sec:sequential_approach} and \ref{sec:geometry_poly_backup}.

Formally, we define the sequential transition-independent game \(\mathcal{M}'_{\mathrm{seq}}\) by taking the TI-zs-SG construction of \citet{escudie2025varepsilonoptimally} and applying it to the sequential occupancy recursion:
at sub-stage \((1,t)\), only player~\textcolor{sthlmRed}{1} is active and only player~\textcolor{sthlmRed}{1}’s occupancy set is updated (the opponent set is held fixed);
at sub-stage \((2,t)\), only player~\textcolor{sthlmRed}{2} is active and only player~\textcolor{sthlmRed}{2}’s occupancy set is updated, while the stage reward is realised and the next stage-boundary state is produced.
The central sequential occupancies \(x_{\textcolor{sthlmRed}{1},t}\) and \(x_{\textcolor{sthlmRed}{2},t}\) remain the value-sufficient statistics manipulated by the uninformed recursion, and they are induced by the compatibility of the two current occupancy sets (cf.\ Figure~\ref{fig:sequential_ti_diagram}).

This interleaving is the only change relative to the simultaneous TI-zs-SG: the information structure, the meaning of occupancy sets, and the equilibrium-consistent extraction operators remain exactly those of \citet{escudie2025varepsilonoptimally}.

\subsection{Losslessness}

\begin{theorem}[Sequential lossless reduction]
\label{thm:sequential_lossless}
The sequential transition-independent zs-SG \(\mathcal{M}'_{\mathrm{seq}}\) is a lossless reduction of the original zs-POSG \(\mathcal{M}\), and composing its two sub-stages at each \(t\) recovers the simultaneous TI-zs-SG optimality equations of \citet{escudie2025varepsilonoptimally}.
\end{theorem}

By Theorem~\ref{thm:sequential_lossless}, every result established for TI-zs-SGs in \citet{escudie2025varepsilonoptimally} applies immediately to \(\mathcal{M}'_{\mathrm{seq}}\), including equilibrium correspondence, \emph{exploitability bounds}, and the associated policy extraction procedure at the level of player occupancy sets.
In the remainder of the paper, our sequential backups are used to implement the value recursion efficiently, while the extraction of executable strategies follows the TI-zs-SG machinery unchanged.

\section{Empirical evaluation}

\newcommand{\subfigsize}{120pt}
\newcommand{\subfigfontsize}{\normalsize}

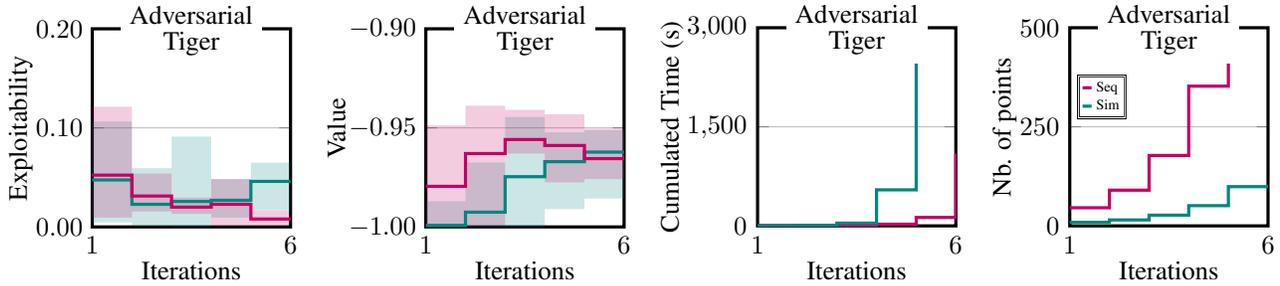
\begin{figure*}[htbp]
  \centering
  \begin{subfigure}
    \centering
    \begin{tikzpicture}
      \begin{axis}[
        title={\shortstack{Adversarial \\ Tiger}},
        title style={yshift=-20pt, fill=white},
        xlabel=Iterations,
        ylabel=Exploitability,
        width=\subfigsize,
        height=\subfigsize,
        xmin=1, xmax=6,
        xtick={1, 6},
        ymin=0, ymax=0.2,
        ytick={0, 0.1, 0.2},
        ylabel style={yshift=-6pt,font=\subfigfontsize},
        xlabel style={yshift=4pt,font=\subfigfontsize},
        grid=both,
        scaled y ticks=false,
        yticklabel style={
          /pgf/number format/fixed,
          /pgf/number format/fixed zerofill,
          /pgf/number format/precision=2,
         font=\subfigfontsize
        },
        style={very thick,font=\subfigfontsize},
        axis line style={very thick},
        legend style={draw=none}
      ]
        \addplot+[const plot, mark=none, very thick, sthlmGreen, opacity=1., solid] table [col sep=comma, x=iter, y=exp_value] {./results/adversarial_tiger/random/sim_mean.log};
        \addplot+[const plot, mark=none, very thick, sthlmRed, opacity=1., solid] table [col sep=comma, x=iter, y=exp_value] {./results/adversarial_tiger/random/seq_mean.log};
        
        \addplot+[name path=sim_max, const plot, mark=none, thick, sthlmGreen, opacity=0., solid] table [col sep=comma, x=iter, y=exp_value] {./results/adversarial_tiger/random/sim_max.log};
        \addplot+[name path=sim_min, const plot, mark=none, thick, sthlmGreen, opacity=0., solid] table [col sep=comma, x=iter, y=exp_value] {./results/adversarial_tiger/random/sim_min.log};
        \addplot[
            fill=sthlmGreen,
            fill opacity=0.2
        ]
        fill between[
            of=sim_max and sim_min
        ];

        \addplot+[name path=seq_max, const plot, mark=none, thick, sthlmRed, opacity=0., solid] table [col sep=comma, x=iter, y=exp_value] {./results/adversarial_tiger/random/seq_max.log};
        \addplot+[name path=seq_min, const plot, mark=none, thick, sthlmRed, opacity=0., solid] table [col sep=comma, x=iter, y=exp_value] {./results/adversarial_tiger/random/seq_min.log};
        \addplot[
            fill=sthlmRed,
            fill opacity=0.2
        ]
        fill between[
            of=seq_max and seq_min
        ];
      \end{axis}
    \end{tikzpicture}
  \end{subfigure}
  % \hfill
  \begin{subfigure}
    \centering
    \begin{tikzpicture}
      \begin{axis}[
        title={\shortstack{Adversarial \\ Tiger}},
        title style={yshift=-20pt, fill=white},
        xlabel=Iterations,
        ylabel=Value,
        width=\subfigsize,
        height=\subfigsize,
        xmin=1, xmax=6,
        xtick={1, 6},
        ymin=-1., ymax=-0.9,
        ytick={-1, -0.95, -0.9},
        ylabel style={yshift=-6pt,font=\subfigfontsize},
        xlabel style={yshift=4pt,font=\subfigfontsize},
        grid=both,
        scaled y ticks=false,
        yticklabel style={
          /pgf/number format/fixed,
          /pgf/number format/fixed zerofill,
          /pgf/number format/precision=2,
         font=\subfigfontsize
        },
        style={very thick,font=\subfigfontsize},
        axis line style={very thick},
        legend style={draw=none}
      ]
        \addplot+[const plot, mark=none, very thick, sthlmGreen, opacity=1., solid] table [col sep=comma, x=iter, y=value] {./results/adversarial_tiger/random/sim_mean.log};
        \addplot+[const plot, mark=none, very thick, sthlmRed, opacity=1., solid] table [col sep=comma, x=iter, y=value] {./results/adversarial_tiger/random/seq_mean.log};
        
        \addplot+[name path=sim_max, const plot, mark=none, thick, sthlmGreen, opacity=0., solid] table [col sep=comma, x=iter, y=value] {./results/adversarial_tiger/random/sim_max.log};
        \addplot+[name path=sim_min, const plot, mark=none, thick, sthlmGreen, opacity=0., solid] table [col sep=comma, x=iter, y=value] {./results/adversarial_tiger/random/sim_min.log};
        \addplot[
            fill=sthlmGreen,
            fill opacity=0.2
        ]
        fill between[
            of=sim_max and sim_min
        ];

        \addplot+[name path=seq_max, const plot, mark=none, thick, sthlmRed, opacity=0., solid] table [col sep=comma, x=iter, y=value] {./results/adversarial_tiger/random/seq_max.log};
        \addplot+[name path=seq_min, const plot, mark=none, thick, sthlmRed, opacity=0., solid] table [col sep=comma, x=iter, y=value] {./results/adversarial_tiger/random/seq_min.log};
        \addplot[
            fill=sthlmRed,
            fill opacity=0.2
        ]
        fill between[
            of=seq_max and seq_min
        ];
      \end{axis}
    \end{tikzpicture}
  \end{subfigure}
  % \hfill
  \begin{subfigure}
    \centering
    \begin{tikzpicture}
      \begin{axis}[
        title={\shortstack{Adversarial \\ Tiger}},
        title style={yshift=-20pt, fill=white},
        xlabel=Iterations,
        ylabel=Cumulated Time (s),
        width=\subfigsize,
        height=\subfigsize,
        xmin=1, xmax=6,
        xtick={1, 6},
        ymin=0, ymax=3000,
        ytick={0, 1500, 3000},
        ylabel style={yshift=-6pt,font=\subfigfontsize},
        xlabel style={yshift=4pt,font=\subfigfontsize},
        grid=both,
        scaled y ticks=false,
        yticklabel style={
          /pgf/number format/fixed,
          /pgf/number format/fixed zerofill,
          /pgf/number format/precision=0,
         font=\subfigfontsize
        },
        style={very thick,font=\subfigfontsize},
        legend style={draw=none},
        style={very thick,font=\subfigfontsize}
      ]
        \addplot+[const plot, mark=none, very thick, sthlmRed, opacity=1.] table [col sep=comma, x=iter, y=cumul_time] {./results/adversarial_tiger/random/seq_mean.log};
        \addplot+[const plot, mark=none, very thick, solid, sthlmGreen, opacity=1.] table [col sep=comma, x=iter, y=cumul_time] {./results/adversarial_tiger/random/sim_mean.log};
      \end{axis}
    \end{tikzpicture}
  \end{subfigure}
  \begin{subfigure}
    \centering
    \begin{tikzpicture}
      \begin{axis}[
        title={\shortstack{Adversarial \\ Tiger}},
        title style={yshift=-20pt, fill=white},
        xlabel=Iterations,
        ylabel=Nb. of points,
        width=\subfigsize,
        height=\subfigsize,
        xmin=1, xmax=6,
        xtick={1, 6},
        ymin=0, ymax=500,
        ytick={0, 250, 500},
        ylabel style={yshift=-6pt,font=\subfigfontsize},
        xlabel style={yshift=4pt,font=\subfigfontsize},
        grid=both,
        scaled y ticks=false,
        yticklabel style={
          /pgf/number format/fixed,
          /pgf/number format/fixed zerofill,
          /pgf/number format/precision=0,
         font=\subfigfontsize
        },
        style={very thick,font=\subfigfontsize},
        legend style={draw=none},
        style={very thick,font=\subfigfontsize},
        legend style={
            at={(0.05,0.75)},
            thin,
            legend cell align=left,
            anchor=north west,
            font=\small,
            inner sep=0.3pt,
            draw=none,
            legend columns=1,
            nodes={scale=0.6},
            name=leg
        },
        legend image post style={
            xscale=0.2
        }
      ]
        \addplot+[const plot, mark=none, very thick, sthlmRed, opacity=1.] table [col sep=comma, x=iter, y=S_size] {./results/adversarial_tiger/random/seq_mean.log};
        \addlegendentry{Seq}
        \addplot+[const plot, mark=none, very thick, solid, sthlmGreen, opacity=1.] table [col sep=comma, x=iter, y=S_size] {./results/adversarial_tiger/random/sim_mean.log};
        \addlegendentry{Sim}
      \end{axis}
      
        \begin{scope}[on background layer]
            \draw[black, thin] (leg.north west) rectangle (leg.south east);
            \draw[black, thin] ([xshift=-0.8pt,yshift=0.8pt]leg.north west) rectangle ([xshift=0.8pt,yshift=-0.8pt]leg.south east);
        \end{scope}
    \end{tikzpicture}
  \end{subfigure}

\caption{PBVI on Adversarial Tiger ($\ell{=}5$): exploitability, value, cumulative runtime, and sample count over iterations, averaged over random seeds. \textcolor{sthlmGreen}{Green}: simultaneous; \textcolor{sthlmRed}{red}: sequential.}
  %\caption{Exploitability, value, cumulative computation time, and number of sampled points of PBVI over iterations on the Adversarial Tiger benchmark ($\ell$=5), evaluated across different random seeds. The \textcolor{sthlmGreen}{green curves} correspond to the \textcolor{sthlmGreen}{simultaneous} variant, while the \textcolor{sthlmRed}{red curves} represent the \textcolor{sthlmRed}{sequential} variant.}
  \label{fig:main_plots}
\end{figure*}

We evaluate our approach on Adversarial Tiger, Recycling, and MABC (\url{http://masplan.org/}), which are standard benchmarks in the recent zs-POSG literature and are used as reference domains in \citet{Wiggers16}, \citet{DelBufDibSaf-DGAA-23}, and \citet{escudie2025varepsilonoptimally}. Following the same convention, we instantiate the competitive setting by using a zero-sum payoff: player~\(\textcolor{sthlmRed}{2}\)’s reward is the negative of player~\(\textcolor{sthlmRed}{1}\)’s reward. This construction preserves the original transition and observation structure and isolates the algorithmic effect of the backup operator. % Additional implementation details are provided in Appendix~A.

For each game, we compare two variants of the \texttt{PBVI} algorithm: \texttt{Sim} \citep{escudie2025varepsilonoptimally}, which implements the simultaneous approach and serves as the baseline, and \texttt{Seq}, which implements the sequential approach. For completeness, we also include \texttt{HSVI} \citep{DelBufDibSaf-DGAA-23} and a naive \texttt{CFR+} implementation \citep{Tammelin14}. Table~\ref{table:main_results:value_and_exploitability} summarizes the results for the most computationally demanding horizons, reporting both the final value reached by each algorithm and the exploitability of the resulting policy.

\paragraph{Simultaneous vs.\ Sequential \texttt{PBVI}.}
Figure~\ref{fig:main_plots} and Table~\ref{table:main_results:value_and_exploitability} compare \texttt{Sim} and \texttt{Seq} under the same \texttt{PBVI} outer loop (sampling, improvement, and pruning), differing only in the backup: \texttt{Sim} applies the simultaneous minimax backup, whereas \texttt{Seq} applies the sequential backup (LP$_{\textcolor{sthlmRed}{2}}$ followed by LP$_{\textcolor{sthlmRed}{1}}$).
We report two outcomes: (i) wall-clock time to reach the reported iterate, and (ii) exploitability $\varepsilon$ of the resulting focal policy.

Across all benchmarks and the largest horizons reported, \texttt{Seq} attains \emph{lower or equal exploitability} than \texttt{Sim}, while being \emph{significantly faster} in the most demanding instances (Speedup column in Table~\ref{table:main_results:value_and_exploitability}, and runtime curves in Figure~\ref{fig:main_plots}).
The reason is structural: the sequential backup replaces a single, highly coupled simultaneous optimisation by two smaller programs whose sizes scale with the cached-envelope representation (Appendix~\ref{app:sec:complexity}), so each improvement step is cheaper.
This lower per-backup cost allows \texttt{Seq} to complete more improvement--expansion cycles within the same time budget, which results in denser sampling and, empirically, better exploitability.

\begin{table}[!ht]
\centering
%\caption{Snapshot of empirical results. Games are ordered by ascending complexity, and by increasing planning horizon \(\ell\). For each setting, we report the runtime of each algorithm (in seconds), as well as the exploitability \(\varepsilon\). The Speedup column indicates the acceleration factor achieved by the sequential approach relative to the simultaneous approach. \textsc{oot} indicates a timeout, \textsc{oom} denotes out-of-memory runs, and ‘--’ means the exploitability budget was exceeded. Best results are highlighted in \highest{magenta}.}
\caption{Snapshot of results (time in seconds, exploitability $\varepsilon$). Speedup is \texttt{Sim} time divided by \texttt{Seq} time. \textsc{oot}: timeout; \textsc{oom}: out of memory; ‘--’: exploitability budget exceeded. Best values per row are highlighted in \highest{magenta}.}
\label{table:main_results:value_and_exploitability}
\scalebox{0.56}{%
\begin{tabular}{@{}l r rr rr rr rr@{}}
\toprule
\textbf{Game ($\ell$)} & \textbf{Speedup} & \multicolumn{2}{c}{\textbf{\texttt{Seq}}} & \multicolumn{2}{c}{\textbf{\texttt{Sim}}} & \multicolumn{2}{c}{\textbf{\texttt{HSVI}}} & \multicolumn{2}{c}{\textbf{\texttt{CFR+}}} \\
& & time & \(\varepsilon\) & time & \(\varepsilon\) & time & \(\varepsilon\) & time & \(\varepsilon\) \\
\cmidrule(lr){3-4}\cmidrule(lr){5-6}\cmidrule(lr){7-8}\cmidrule(lr){9-10}

kuhn poker & & 0.2 & \highest{0.00} & 0.1 & \highest{0.00} & \textsc{na} & \textsc{na} & \highest{0.01} & \highest{0.00}\\

\myrowcolour adversarial-tiger(3) & 7.33 & \highest{0.15} & \highest{0.00} & 1.1 & \highest{0.00} & 500 & \highest{0.00} & 1 & \highest{0.00} \\
adversarial-tiger(4) & 1.4 & \highest{10} & \highest{0.00} & 14 & \highest{0.00} & \textsc{oot} &  & 17 & 0.01 \\
\myrowcolour adversarial-tiger(5) & 1.7 & \highest{30} & \highest{0.00} & 51 & \highest{0.00} & \textsc{oot} &  & 181 & 0.01 \\
adversarial-tiger(7) & 1.83 & \highest{264} & \highest{0.00} & 485 & 0.02 & \textsc{oot} &  & \textsc{oom} & \\
\myrowcolour adversarial-tiger(10) & 26.6 & \highest{287} & \highest{0.02} & 7648 & 0.05 & \textsc{oot} &  & \textsc{oom} &  \\
adversarial-tiger(12) & & \highest{782} & \highest{0.02} & \textsc{oot} &  & \textsc{oot} &  & \textsc{oom} &  \\
\myrowcolour adversarial-tiger(14) & & \highest{1840} & \highest{0.18} & \textsc{oot} &  & \textsc{oot} &  & \textsc{oom} &  \\

mabc(3) & 20 & 0.85 & \highest{0.00} & 17 & \highest{0.00} & 70 & \highest{0.00} & \highest{0.5} & \highest{0.00} \\
\myrowcolour mabc(4) & 20.8 & 15 & \highest{0.00} & 312 & 0.01 & \textsc{oot} &  & \highest{4} & \highest{0.00} \\
mabc(5) & 3.1 & \highest{44} & \highest{0.00} & 136 & 0.02 & \textsc{oot} &  & 51 & 0.01 \\
\myrowcolour mabc(7) & 18.5 & \highest{178} & \highest{0.00} & 3298 & 0.04 & \textsc{oot} &  & \textsc{oom} &  \\
mabc(10) & & \highest{1377} & 0.05 & \textsc{oot} &  & \textsc{oot} &  & \textsc{oom} &  \\

\myrowcolour recycling(3) & 0.18 & 27 & \highest{0.00} & \highest{5} & \highest{0.00} & 430 & \highest{0.00} & 6 & \highest{0.00} \\
recycling(4) & 0.4 & 36 & \highest{0.01} & \highest{15} & \highest{0.01} & \textsc{oot} &  & 80 & 0.03 \\
\myrowcolour recycling(5) & 3.1 & \highest{72} & \highest{0.01} & 225 & 0.03 & \textsc{oot} &  & \textsc{oom} &  \\
recycling(7) & 23.4 & \highest{84}  & \highest{0.01} & 1969 & 0.02 & \textsc{oot} &  & \textsc{oom} &   \\
\myrowcolour recycling(10) & 14.7 & \highest{1616} & \highest{0.00} & 23798 & 0.05 & \textsc{oot} &  & \textsc{oom} &  \\

competitive-tiger(3) & 0.93 & 48 & 0.02 & 45 & 0.02 & 291 & \highest{0.00} & \highest{17} & \highest{0.00} \\
\myrowcolour competitive-tiger(4) & 1.98 & \highest{60} & 0.01 & 119 & \highest{0.00} & \textsc{oot} &  & \textsc{oom} &  \\
competitive-tiger(5) & 2.11 & \highest{152} & 0.03 & 322 & \highest{0.02} & \textsc{oot} &  & \textsc{oom} &  \\
\myrowcolour competitive-tiger(7) & 3.7 & \highest{450} & 0.08 & 1685 & \highest{0.06} & \textsc{oot} &  & \textsc{oom} &  \\
competitive-tiger(10) &  & \highest{847} & \highest{0.03} & \textsc{oot} &  & \textsc{oot} &  & \textsc{oom} &  \\

% \myrowcolour competitive-tiger(4)  & \highest{-0.03} & 0.03 & -0.07 & \highest{0.00} & -0.05 & 0.03 & \textsc{oot} &  & \textsc{oom} &  \\
% \myrowcolour competitive-tiger(5)  & \highest{-0.06} & 0.01 & -0.08 & \highest{0.00} & -0.10 & 0.02 & \textsc{oot} &  & \textsc{oom} &  \\
% competitive-tiger(7)  & \highest{-0.15} & 0.03 & -0.17 & 0.04 & \highest{-0.15} & \highest{0.02} & \textsc{oot} &  & \textsc{oom} &  \\
% \myrowcolour competitive-tiger(10) & \textsc{oot} &  & -0.29 & -- & \highest{-0.20} & -- & \textsc{oot} &  & \textsc{oom} &  \\

\bottomrule
\end{tabular}
}
\end{table}

Regarding values, \texttt{Seq} can appear slightly more conservative early on (lower value estimates) because it maintains a lower-bound representation built from cached envelopes.
However, as iterations proceed, the additional improvement cycles made possible by the cheaper backup translate into tighter envelopes at the sampled occupancies, and the value curves typically catch up to the simultaneous variant (e.g., Adv.\ Tiger and MABC in Figure~\ref{fig:value_and_exp_plots}).
Finally, the exploitability traces produced by \texttt{Sim} exhibit larger transient fluctuations in several settings (Figure~\ref{fig:main_plots}); in our experiments this correlates with slower sampling growth, so each update is informed by fewer newly reached occupancies before the next greedy step is computed.

%\paragraph{Other baselines.}
\texttt{HSVI} \citep{DelBufDibSaf-DGAA-23} and \texttt{CFR+} \citep{Tammelin14} are included for completeness. A detailed discussion of their behaviour on these benchmarks, including the practical limitations that lead to timeouts or memory exhaustion at larger horizons, is provided in \citet{escudie2025varepsilonoptimally}; we therefore focus here on isolating the effect of replacing the simultaneous backup (\texttt{Sim}) by the sequential one (\texttt{Seq}).

\section{Conclusion}
Lossless reductions give zs-POSGs a clean dynamic-programming interpretation, yet in practice they still leave the \emph{same} stumbling block: a Bellman backup that is a simultaneous minimax over history-dependent decision rules, and therefore grows explosively with private-history length. Our key idea is to remove this obstacle at its source by \emph{sequentialising the backup}: we replace the exponential stage game by two structured LP sub-stages whose size is polynomial in the explicit cache representation, so the cost is governed by the number of supporting hyperplanes actually stored rather than by the dimension of the strategy space. This structural shift compounds in practice—cheaper backups enable denser point sets, denser point sets yield tighter envelopes, and tighter envelopes translate into lower exploitability—so across benchmarks and horizons \texttt{Seq} consistently improves robustness over \texttt{Sim} while often delivering order-of-magnitude speedups. The message is simple: in zs-POSGs, scalability is won in the backup, and sequentialisation turns reduction-based planning from a principled formulation into a usable algorithm.

\bibliography{example_paper}
\bibliographystyle{icml2026}

%%%%%%%%%%%%%%%%%%%%%%%%%%%%%%%%%%%%%%%%%%%%%%%%%%%%%%%%%%%%%%%%%%%%%%%%%%%%%%%
%%%%%%%%%%%%%%%%%%%%%%%%%%%%%%%%%%%%%%%%%%%%%%%%%%%%%%%%%%%%%%%%%%%%%%%%%%%%%%%
% APPENDIX
%%%%%%%%%%%%%%%%%%%%%%%%%%%%%%%%%%%%%%%%%%%%%%%%%%%%%%%%%%%%%%%%%%%%%%%%%%%%%%%
%%%%%%%%%%%%%%%%%%%%%%%%%%%%%%%%%%%%%%%%%%%%%%%%%%%%%%%%%%%%%%%%%%%%%%%%%%%%%%%
\onecolumn
\appendix
\begin{center}\LARGE \textbf{Appendix — Detailed Proofs}\end{center}
\tableofcontents
\bigskip

\section{Sequential Optimality Equations --- Proof of Theorem \ref{thm:bellman_seq}}
\label{app:proof:bellman_seq}

This section provides a concise, self-contained proof by backward induction. 
The only non-trivial points are: 
(i) sequential occupancy states are sufficient (Markov) for one-step rewards and transitions, 
and (ii) the one-step optimisations attain \(\max/\min\). 

\paragraph{Termination conditions.} 
At stage boundary \(t=\ell\), the process terminates at \(x_{\textcolor{sthlmRed}{1},\ell}\).
We set \(v^*_{\mathrm{seq}}(x_{\textcolor{sthlmRed}{1},\ell})\doteq 0\).

\paragraph{Discount split.} Since the true stage reward is realised only at sub-stage \(\textcolor{sthlmRed}{2}\), discounting must apply only once per completed original stage.
We therefore set \(\gamma_{\textcolor{sthlmRed}{1}}\doteq 1\) and \(\gamma_{\textcolor{sthlmRed}{2}}\doteq \gamma \).

% ------------------------------------------------------------
\subsection{Sufficiency of Sequential Occupancy States --- Proof Theorem \ref{thm:sufficiency:sequential:occupancy:state}}
\label{app:proof:sufficiency:sequential:occupancy:state}
% ------------------------------------------------------------

\begin{lemma}
\label{lem:seq_sufficiency_zero_dot}
Fix \(t<\ell\).
Let \(\theta_{\textcolor{sthlmRed}{1},t}\) and \(\theta_{\textcolor{sthlmRed}{2},t}\) be the sequential planner's exhaustive information states, and let \(x_{\textcolor{sthlmRed}{1},t}\) and \(x_{\textcolor{sthlmRed}{2},t}\) be the induced sequential occupancies.
Private histories are updated by appending the realised pair \((a_{\textcolor{sthlmRed}{i}},z_{\textcolor{sthlmRed}{i}})\) to \(h_{\textcolor{sthlmRed}{i},t}\).
\begin{enumerate}[(i)]
\item For any \(d_{\textcolor{sthlmRed}{1},t}\in\mathcal{D}_{\textcolor{sthlmRed}{1},t}\), consider the resulting exhaustive history
\(\theta_{\textcolor{sthlmRed}{2},t}=(\theta_{\textcolor{sthlmRed}{1},t},d_{\textcolor{sthlmRed}{1},t})\).
Then, for all \((s,h_{\textcolor{sthlmRed}{1},t},h_{\textcolor{sthlmRed}{2},t},a_{\textcolor{sthlmRed}{1}})\),
\[
\Pr(s,h_{\textcolor{sthlmRed}{1},t},h_{\textcolor{sthlmRed}{2},t},a_{\textcolor{sthlmRed}{1}}\mid \theta_{\textcolor{sthlmRed}{2},t})
=
\tau_{\mathrm{seq}}(x_{\textcolor{sthlmRed}{1},t},d_{\textcolor{sthlmRed}{1},t})
(s,h_{\textcolor{sthlmRed}{1},t},h_{\textcolor{sthlmRed}{2},t},a_{\textcolor{sthlmRed}{1}})
\quad
\text{and} 
\quad
\rho_{\mathrm{seq}}(x_{\textcolor{sthlmRed}{1},t},d_{\textcolor{sthlmRed}{1},t})=0
.\]

\item For any \(d_{\textcolor{sthlmRed}{2},t}\in\mathcal{D}_{\textcolor{sthlmRed}{2},t}\), consider the resulting exhaustive history
\(\theta_{\textcolor{sthlmRed}{1},t+1}=(\theta_{\textcolor{sthlmRed}{2},t},d_{\textcolor{sthlmRed}{2},t})\).
Then, for all \((s',h_{\textcolor{sthlmRed}{1},t+1},h_{\textcolor{sthlmRed}{2},t+1})\),
\[
\Pr(s',h_{\textcolor{sthlmRed}{1},t+1},h_{\textcolor{sthlmRed}{2},t+1}\mid \theta_{\textcolor{sthlmRed}{1},t+1})
=
\tau_{\mathrm{seq}}(x_{\textcolor{sthlmRed}{2},t},d_{\textcolor{sthlmRed}{2},t})
(s',h_{\textcolor{sthlmRed}{1},t+1},h_{\textcolor{sthlmRed}{2},t+1})
\quad
\text{and}
\quad
\mathbb{E}\!\left[r(s,a_{\textcolor{sthlmRed}{1}},a_{\textcolor{sthlmRed}{2}})\mid \theta_{\textcolor{sthlmRed}{2},t},d_{\textcolor{sthlmRed}{2},t}\right]
=
\rho_{\mathrm{seq}}(x_{\textcolor{sthlmRed}{2},t},d_{\textcolor{sthlmRed}{2},t}).
\]
\end{enumerate}
In particular, the one-step law of the next occupancy and the one-step expected reward depend on the exhaustive past only through
\((x_{\textcolor{sthlmRed}{i},t},d_{\textcolor{sthlmRed}{i},t})\).
\end{lemma}

\begin{proof}
\textbf{(i)} Under \(\theta_{\textcolor{sthlmRed}{1},t}\), the triple \((s,h_{\textcolor{sthlmRed}{1},t},h_{\textcolor{sthlmRed}{2},t})\) has distribution \(x_{\textcolor{sthlmRed}{1},t}\).
Moreover, for each \(h_{\textcolor{sthlmRed}{1},t}\) and each \(a_{\textcolor{sthlmRed}{1}}\), the probability of sampling action \(a_{\textcolor{sthlmRed}{1}}\) equals
\(d_{\textcolor{sthlmRed}{1},t}(a_{\textcolor{sthlmRed}{1}}\mid h_{\textcolor{sthlmRed}{1},t})\).
Therefore, for every \((s,h_{\textcolor{sthlmRed}{1},t},h_{\textcolor{sthlmRed}{2},t},a_{\textcolor{sthlmRed}{1}})\),
\[
\Pr(s,h_{\textcolor{sthlmRed}{1},t},h_{\textcolor{sthlmRed}{2},t},a_{\textcolor{sthlmRed}{1}}\mid \theta_{\textcolor{sthlmRed}{2},t})
=
x_{\textcolor{sthlmRed}{1},t}(s,h_{\textcolor{sthlmRed}{1},t},h_{\textcolor{sthlmRed}{2},t})\,
d_{\textcolor{sthlmRed}{1},t}(a_{\textcolor{sthlmRed}{1}}\mid h_{\textcolor{sthlmRed}{1},t}),
\]
which equals \(\tau_{\mathrm{seq}}(x_{\textcolor{sthlmRed}{1},t},d_{\textcolor{sthlmRed}{1},t})(s,h_{\textcolor{sthlmRed}{1},t},h_{\textcolor{sthlmRed}{2},t},a_{\textcolor{sthlmRed}{1}})\)
by the sub-stage \((1,t)\) pushforward construction of Section~\ref{subsec:sequential_primitives}.
Moreover, the sub-stage reward at \((1,t)\) is defined to be \(0\).

\textbf{(ii)} Under \(\theta_{\textcolor{sthlmRed}{2},t}\) and the choice \(d_{\textcolor{sthlmRed}{2},t}\) (equivalently, under \(\theta_{\textcolor{sthlmRed}{1},t+1}\)),
for each \(h_{\textcolor{sthlmRed}{2},t}\) and each \(a_{\textcolor{sthlmRed}{2}}\), the probability of sampling action \(a_{\textcolor{sthlmRed}{2}}\) equals
\(d_{\textcolor{sthlmRed}{2},t}(a_{\textcolor{sthlmRed}{2}}\mid h_{\textcolor{sthlmRed}{2},t})\).
Then, for each \((s,a_{\textcolor{sthlmRed}{1}},a_{\textcolor{sthlmRed}{2}})\) and each \((s',z_{\textcolor{sthlmRed}{1}},z_{\textcolor{sthlmRed}{2}})\),
the probability of transitioning and observing equals
\(p(s',z_{\textcolor{sthlmRed}{1}},z_{\textcolor{sthlmRed}{2}}\mid s,a_{\textcolor{sthlmRed}{1}},a_{\textcolor{sthlmRed}{2}})\),
and histories are updated by appending \((a_{\textcolor{sthlmRed}{i}},z_{\textcolor{sthlmRed}{i}})\).
By the sub-stage \((2,t)\) transition construction of Section~\ref{subsec:sequential_primitives} (integration of \(p\) against \(x_{\textcolor{sthlmRed}{2},t}\) and \(d_{\textcolor{sthlmRed}{2},t}\) with history appending),
this induces exactly \(\tau_{\mathrm{seq}}(x_{\textcolor{sthlmRed}{2},t},d_{\textcolor{sthlmRed}{2},t})\), yielding the claimed componentwise identity.
Finally, the conditional expected stage reward is \(\rho_{\mathrm{seq}}(x_{\textcolor{sthlmRed}{2},t},d_{\textcolor{sthlmRed}{2},t})\) by definition of \(\rho_{\mathrm{seq}}\).
\end{proof}

\begin{proof}[Proof of \Cref{thm:sufficiency:sequential:occupancy:state}]
Fix a sub-stage \((\textcolor{sthlmRed}{i},t)\) and consider two exhaustive information states
\(\theta_{\textcolor{sthlmRed}{i},t}\) and \(\theta'_{\textcolor{sthlmRed}{i},t}\) that induce the same sequential occupancy \(x_{\textcolor{sthlmRed}{i},t}\).
By Lemma~\ref{lem:seq_sufficiency_zero_dot}, for any admissible one-step decision rule \(d_{\textcolor{sthlmRed}{i},t}\),
the conditional expected one-step reward and the conditional law of the next occupancy coincide under
\(\theta_{\textcolor{sthlmRed}{i},t}\) and \(\theta'_{\textcolor{sthlmRed}{i},t}\), and depend on the exhaustive past only through
\((x_{\textcolor{sthlmRed}{i},t},d_{\textcolor{sthlmRed}{i},t})\).

Fix any continuation policy \(\pi_{t:}\) in \(\mathcal{M}_{\mathrm{seq}}\) from sub-stage \((\textcolor{sthlmRed}{i},t)\) onward.
We show by induction on the remaining number of sub-stages that the expected cumulative discounted return induced by \(\pi_{t:}\)
is the same when starting from \(\theta_{\textcolor{sthlmRed}{i},t}\) or from \(\theta'_{\textcolor{sthlmRed}{i},t}\).
The base case holds at termination.
For the induction step, let \(d_{\textcolor{sthlmRed}{i},t}\) be the decision rule prescribed by \(\pi_{t:}\) at \((\textcolor{sthlmRed}{i},t)\)
as a function of \(h_{\textcolor{sthlmRed}{i},t}\).
Lemma~\ref{lem:seq_sufficiency_zero_dot} gives the same one-step expected reward and the same conditional law of the next occupancy under
\(\theta_{\textcolor{sthlmRed}{i},t}\) and \(\theta'_{\textcolor{sthlmRed}{i},t}\); conditioning on that next occupancy and applying the induction hypothesis to the
remaining continuation under \(\pi_{t:}\) yields equality of the full expected return.
Hence, for any continuation policy \(\pi_{t:}\), the induced expected return depends on the exhaustive past only through the current occupancy \(x_{\textcolor{sthlmRed}{i},t}\).

Taking the appropriate \(\max\)/\(\min\) over continuation policies preserves this invariance, so the optimal value depends on
\(\theta_{\textcolor{sthlmRed}{i},t}\) only through \(x_{\textcolor{sthlmRed}{i},t}\).
Therefore sequential occupancy states are sufficient to compute \(v^*_{\mathrm{seq}}\) in \(\mathcal{M}_{\mathrm{seq}}\).
Finally, at stage boundaries the two sub-stages compose to the original simultaneous stage, so these optimal boundary values coincide with
the optimal values in the simultaneous zs-POSG \(\mathcal{M}\) (Section~\ref{subsec:sequential_bellman}).
\end{proof}

% ------------------------------------------------------------
\subsection{Compactness and attainment}
% ------------------------------------------------------------

\begin{lemma}[Compactness of decision-rule sets]
\label{lem:compact_camera_ready}
For each \(t\), the set \(\mathcal{D}_{\textcolor{sthlmRed}{i},t}\) is compact.
\end{lemma}

\begin{proof}
Since \(\mathcal{H}_{\textcolor{sthlmRed}{i},t}\) and \(\mathcal{A}_{\textcolor{sthlmRed}{i}}\) are finite, a decision rule is a finite product of simplices
\(\Delta(\mathcal{A}_{\textcolor{sthlmRed}{i}})\), hence compact.
\end{proof}

% ------------------------------------------------------------
\subsection{Sequential Bellman equations by backward induction}
% ------------------------------------------------------------

\begin{proof}[Proof of \Cref{thm:bellman_seq}]
We prove the sequential Bellman equalities by backward induction on \(t\), using the standard \emph{one-step decomposition} of the finite-horizon objective.

\smallskip\noindent
\textbf{Base.}
At \(t=\ell\), no reward can be collected in the remaining horizon, hence
\(v^*_{\mathrm{seq}}(x_{\textcolor{sthlmRed}{1},\ell})=0\).

\smallskip\noindent
\textbf{Inductive hypothesis.}
Fix \(t\in\{0,\dots,\ell-1\}\) and assume that \(v^*_{\mathrm{seq}}(x_{\textcolor{sthlmRed}{1},t+1})\) is defined for every valid boundary occupancy at time \(t+1\),
and that it is continuous on that set.

\smallskip\noindent
\textbf{Step 1: intermediate sub-stage \((\textcolor{sthlmRed}{2},t)\).}
Fix a valid \(x_{\textcolor{sthlmRed}{2},t}\).
For any \(d_{\textcolor{sthlmRed}{2},t}\in\mathcal{D}_{\textcolor{sthlmRed}{2},t}\),
Lemma~\ref{lem:seq_sufficiency_zero_dot} ensures that
\(\tau_{\mathrm{seq}}(x_{\textcolor{sthlmRed}{2},t},d_{\textcolor{sthlmRed}{2},t})\) is a valid boundary occupancy at time \(t+1\),
so \(v^*_{\mathrm{seq}}\bigl(\tau_{\mathrm{seq}}(x_{\textcolor{sthlmRed}{2},t},d_{\textcolor{sthlmRed}{2},t})\bigr)\) is well-defined.
Consider the finite-horizon return from \((\textcolor{sthlmRed}{2},t)\) under a choice \(d_{\textcolor{sthlmRed}{2},t}\) and an arbitrary continuation policy \(\pi_{\textcolor{sthlmRed}{2},t+1:}\):
by the Markov property (Lemma~\ref{lem:seq_sufficiency_zero_dot}) and the definition of the sequential primitives,
it decomposes as the one-step reward plus the discounted continuation value starting from the next boundary occupancy.
Optimising over continuations therefore yields the \emph{principle of optimality}:
\[
v^*_{\mathrm{seq}}(x_{\textcolor{sthlmRed}{2},t})
=
\min_{d_{\textcolor{sthlmRed}{2},t}\in\mathcal{D}_{\textcolor{sthlmRed}{2},t}}
\left[
\rho_{\mathrm{seq}}(x_{\textcolor{sthlmRed}{2},t}, d_{\textcolor{sthlmRed}{2},t})
+\gamma_{\textcolor{sthlmRed}{2}} \,v^{*}_{\mathrm{seq}}\!\bigl(\tau_{\mathrm{seq}}(x_{\textcolor{sthlmRed}{2},t}, d_{\textcolor{sthlmRed}{2},t})\bigr)
\right].
\]
Moreover, the map \(d_{\textcolor{sthlmRed}{2},t}\mapsto \rho_{\mathrm{seq}}(x_{\textcolor{sthlmRed}{2},t},d_{\textcolor{sthlmRed}{2},t})\) is affine,
\(d_{\textcolor{sthlmRed}{2},t}\mapsto \tau_{\mathrm{seq}}(x_{\textcolor{sthlmRed}{2},t},d_{\textcolor{sthlmRed}{2},t})\) is affine,
and \(v^*_{\mathrm{seq}}\) is continuous by the inductive hypothesis; hence the bracketed objective is continuous.
Since \(\mathcal{D}_{\textcolor{sthlmRed}{2},t}\) is compact (Lemma~\ref{lem:compact_camera_ready}), the minimum is attained.

\smallskip\noindent
\textbf{Step 2: boundary sub-stage \((\textcolor{sthlmRed}{1},t)\).}
By Step~1, \(v^*_{\mathrm{seq}}(x_{\textcolor{sthlmRed}{2},t})\) is well-defined for every valid intermediate occupancy at time \(t\). % (+1 line)
Fix a valid \(x_{\textcolor{sthlmRed}{1},t}\).
For any \(d_{\textcolor{sthlmRed}{1},t}\in\mathcal{D}_{\textcolor{sthlmRed}{1},t}\),
Lemma~\ref{lem:seq_sufficiency_zero_dot} ensures that
\(\tau_{\mathrm{seq}}(x_{\textcolor{sthlmRed}{1},t},d_{\textcolor{sthlmRed}{1},t})\) is a valid intermediate occupancy at time \(t\),
so \(v^*_{\mathrm{seq}}\bigl(\tau_{\mathrm{seq}}(x_{\textcolor{sthlmRed}{1},t},d_{\textcolor{sthlmRed}{1},t})\bigr)\) is well-defined.
By construction \(\rho_{\mathrm{seq}}(x_{\textcolor{sthlmRed}{1},t},d_{\textcolor{sthlmRed}{1},t})\equiv 0\),
and by the same one-step decomposition/optimal-continuation argument as in Step~1,
\[
v^*_{\mathrm{seq}}(x_{\textcolor{sthlmRed}{1},t})
=
\max_{d_{\textcolor{sthlmRed}{1},t}\in\mathcal{D}_{\textcolor{sthlmRed}{1},t}}
\left[
\rho_{\mathrm{seq}}(x_{\textcolor{sthlmRed}{1},t}, d_{\textcolor{sthlmRed}{1},t})
+
\gamma_{\textcolor{sthlmRed}{1}}
v^{*}_{\mathrm{seq}}\!\bigl(\tau_{\mathrm{seq}}(x_{\textcolor{sthlmRed}{1},t}, d_{\textcolor{sthlmRed}{1},t})\bigr)
\right],
\]
with \(\gamma_{\textcolor{sthlmRed}{1}}=1\).
As above, the objective is continuous in \(d_{\textcolor{sthlmRed}{1},t}\) (affine \(\tau_{\mathrm{seq}}\) composed with a continuous \(v^*_{\mathrm{seq}}\)),
and \(\mathcal{D}_{\textcolor{sthlmRed}{1},t}\) is compact, so the maximum is attained.

\smallskip\noindent
\textbf{Conclusion.}
The base case and Steps~1--2 establish the claimed Bellman equalities for every \(t<\ell\) and every transient sub-stage occupancy.
\end{proof}

\section{Geometry of Optimal Sequential Value Functions --- Proof Theorem \ref{thm:geometry} }
\label{app:proof:geometry}

\begin{proof}[Proof of \Cref{thm:geometry}]
Fix a sub-stage \(\textcolor{sthlmRed}{i}\in\{1,2\}\) at some stage \(t\).
An occupancy \(x_{\textcolor{sthlmRed}{i}}\) is a distribution over \((\omega,h_{\textcolor{sthlmRed}{2},t})\) with \(\omega\in\Omega_{\textcolor{sthlmRed}{i}}\).
For each \(h_{\textcolor{sthlmRed}{2},t}\in\mathcal H_{\textcolor{sthlmRed}{2},t}\), define the unnormalised slice
\(b_{h_{\textcolor{sthlmRed}{2},t}}\in\mathbb{R}_+^{\Omega_{\textcolor{sthlmRed}{i}}}\) by
\(b_{h_{\textcolor{sthlmRed}{2},t}}(\omega)\doteq x_{\textcolor{sthlmRed}{i}}(\omega,h_{\textcolor{sthlmRed}{2},t})\).

\smallskip\noindent
\textbf{Inner pessimistic layer.}
Fix any player~\(\textcolor{sthlmRed}{1}\) continuation policy \(\pi_{\textcolor{sthlmRed}{1},t:}\).
For any player~\(\textcolor{sthlmRed}{2}\) continuation tree \(\pi_{\mathrm{tree},\textcolor{sthlmRed}{2},t}\) (rooted at the empty continuation history),
the definition of value-vector pairs induces a value vector
\(\alpha^{\pi_{\textcolor{sthlmRed}{1},t:},\pi_{\mathrm{tree},\textcolor{sthlmRed}{2},t}}\in\mathbb{R}^{\Omega_{\textcolor{sthlmRed}{i}}}\).
Since the horizon and all action/observation sets are finite, the set of deterministic player~\(\textcolor{sthlmRed}{2}\) continuation trees from stage \(t\) is finite; moreover, because the continuation value is linear in the player~\(\textcolor{sthlmRed}{2}\) continuation, the minimum over (possibly mixed) continuations is attained at a deterministic tree.
Let \(\Gamma_{\!\textcolor{sthlmRed}{2}}\subset\mathbb{R}^{\Omega_{\textcolor{sthlmRed}{i}}}\) be the (finite) set of vectors
\(\alpha^{\pi_{\textcolor{sthlmRed}{1},t:},\pi_{\mathrm{tree},\textcolor{sthlmRed}{2},t}}\) induced by these deterministic trees.
Then, for each \(h_{\textcolor{sthlmRed}{2},t}\),
\[
\min_{\pi_{\mathrm{tree},\textcolor{sthlmRed}{2},t}\ \emph{admissible}}
\Bigl\langle b_{h_{\textcolor{sthlmRed}{2},t}}, \alpha^{\pi_{\textcolor{sthlmRed}{1},t:},\pi_{\mathrm{tree},\textcolor{sthlmRed}{2},t}}\Bigr\rangle
=
\min_{\alpha\in\Gamma_{\!\textcolor{sthlmRed}{2}}}
\langle b_{h_{\textcolor{sthlmRed}{2},t}},\alpha\rangle.
\]
Moreover, since \(\pi_{\textcolor{sthlmRed}{2},t:}\) is a family of decision rules indexed by player~\(\textcolor{sthlmRed}{2}\) histories, it selects, for each realised \(h_{\textcolor{sthlmRed}{2},t}\), a re-rooted continuation tree
\(\pi_{\mathrm{tree},\textcolor{sthlmRed}{2},t}=\pi_{\textcolor{sthlmRed}{2},t:}(h_{\textcolor{sthlmRed}{2},t})\).
Thus the player~\(\textcolor{sthlmRed}{2}\)-optimal continuation value under \(\pi_{\textcolor{sthlmRed}{1},t:}\) equals
\[
\textstyle
\sum_{h_{\textcolor{sthlmRed}{2},t}}
\min_{\alpha\in\Gamma_{\!\textcolor{sthlmRed}{2}}}
\langle b_{h_{\textcolor{sthlmRed}{2},t}},\alpha\rangle
\;=\;
\mathtt{Val}_{\Gamma_{\!\textcolor{sthlmRed}{2}}}(x_{\textcolor{sthlmRed}{i}}).
\]

\smallskip\noindent
\textbf{Outer optimistic layer.}
Let \(\Gamma_{\!\textcolor{sthlmRed}{1}}\) be a family of such finite sets \(\Gamma_{\!\textcolor{sthlmRed}{2}}\) representing the non-dominated player~\(\textcolor{sthlmRed}{1}\) continuations.
Optimising over player~\(\textcolor{sthlmRed}{1}\) continuations is therefore equivalent to maximising over \(\Gamma_{\!\textcolor{sthlmRed}{2}}\in\Gamma_{\!\textcolor{sthlmRed}{1}}\), giving
\[
v^*_{\mathrm{seq}}(x_{\textcolor{sthlmRed}{i}})
=
\max_{\Gamma_{\!\textcolor{sthlmRed}{2}}\in\Gamma_{\!\textcolor{sthlmRed}{1}}}
\mathtt{Val}_{\Gamma_{\!\textcolor{sthlmRed}{2}}}(x_{\textcolor{sthlmRed}{i}}),
\]
as claimed.
\end{proof}

\section{Greedy Selection Operators --- Proofs of Lemmas \ref{lem:greedy_x1}--\ref{lem:greedy_x2}}
\label{app:proof:greedy_selection}

\subsection{Proof of Lemma \ref{lem:greedy_x1}}

\begin{proof}
Fix a boundary occupancy \(x_{\textcolor{sthlmRed}{1}}\) at sub-stage \((\textcolor{sthlmRed}{1},t)\).
By construction, \(\rho_{\mathrm{seq}}(x_{\textcolor{sthlmRed}{1}},d_{\textcolor{sthlmRed}{1}})\equiv 0\); hence any greedy decision rule of player~\(\textcolor{sthlmRed}{1}\) maximises
\begin{equation}
\label{eq:greedy-x1-start-icmlclean}
d_{\textcolor{sthlmRed}{1}}\ \longmapsto\ v^*_{\mathrm{seq}}\!\bigl(\tau_{\mathrm{seq}}(x_{\textcolor{sthlmRed}{1}},d_{\textcolor{sthlmRed}{1}})\bigr).
\end{equation}
Let \(x_{\textcolor{sthlmRed}{2}}\doteq \tau_{\mathrm{seq}}(x_{\textcolor{sthlmRed}{1}},d_{\textcolor{sthlmRed}{1}})\).

At the intermediate sub-stage \((\textcolor{sthlmRed}{2},t)\), the optimal value admits the envelope representation
\begin{equation}
\label{eq:geometry-max-icmlclean}
v^*_{\mathrm{seq}}(x_{\textcolor{sthlmRed}{2}})
=
\max_{\Gamma_{\!\textcolor{sthlmRed}{2}}\in\Gamma_{\!\textcolor{sthlmRed}{1}}^{*}}
\mathtt{Val}_{\Gamma_{\!\textcolor{sthlmRed}{2}}}(x_{\textcolor{sthlmRed}{2}}),
\qquad
\mathtt{Val}_{\Gamma_{\!\textcolor{sthlmRed}{2}}}(x_{\textcolor{sthlmRed}{2}})
\doteq
\sum_{h_{\textcolor{sthlmRed}{2}}}
\min_{\alpha\in\Gamma_{\!\textcolor{sthlmRed}{2}}}
\langle b_{h_{\textcolor{sthlmRed}{2}}},\alpha\rangle.
\end{equation}

We justify that the outer optimisation in \eqref{eq:geometry-max-icmlclean} is a \(\max\) (not merely a \(\sup\)).
Let \(\Pi_{\textcolor{sthlmRed}{1},t:}\) denote the set of player~\(\textcolor{sthlmRed}{1}\) continuation policies from stage \(t\) onward; since the horizon is finite and each decision rule is a simplex over a finite action set, \(\Pi_{\textcolor{sthlmRed}{1},t:}\) is a finite product of simplices and therefore compact.
For each \(\pi_{\textcolor{sthlmRed}{1},t:}\in\Pi_{\textcolor{sthlmRed}{1},t:}\), let \(\Gamma_{\!\textcolor{sthlmRed}{2},\pi_{\textcolor{sthlmRed}{1},t:}} \) be the associated finite set of \(\alpha\)-vectors induced by the finite index set \(\mathcal{K}\) of deterministic follower policy trees, and define
\(\Gamma_{\!\textcolor{sthlmRed}{1}}^{*}\doteq \{\Gamma_{\!\textcolor{sthlmRed}{2},\pi_{\textcolor{sthlmRed}{1},t:}} : \pi_{\textcolor{sthlmRed}{1},t:}\in\Pi_{\textcolor{sthlmRed}{1},t:}\}\).
For each fixed reachable \(x_{\textcolor{sthlmRed}{2}}\), the real-valued map
\begin{equation}
\label{eq:continuation-to-value-icmlclean}
\pi_{\textcolor{sthlmRed}{1},t:}\ \longmapsto\ \mathtt{Val}_{\Gamma_{\!\textcolor{sthlmRed}{2},\pi_{\textcolor{sthlmRed}{1},t:}}}(x_{\textcolor{sthlmRed}{2}})
\end{equation}
is continuous, because \(\Gamma_{\!\textcolor{sthlmRed}{2},\pi_{\textcolor{sthlmRed}{1},t:}}=\{\alpha^{\kappa}_{\pi_{\textcolor{sthlmRed}{1},t:}}:\kappa\in\mathcal{K}\}\) and, for each \(\kappa\in\mathcal{K}\),
\(\pi_{\textcolor{sthlmRed}{1},t:}\mapsto \langle b_{h_{\textcolor{sthlmRed}{2}}},\alpha^{\kappa}_{\pi_{\textcolor{sthlmRed}{1},t:}}\rangle\) is continuous, while the minimum over finitely many continuous functions preserves continuity.
Therefore, by Weierstrass, there exists \(\pi_{\textcolor{sthlmRed}{1},t:}^*\in\Pi_{\textcolor{sthlmRed}{1},t:}\) attaining the maximum of \eqref{eq:continuation-to-value-icmlclean}, and the corresponding
\(\Gamma_{\!\textcolor{sthlmRed}{2},\pi_{\textcolor{sthlmRed}{1},t:}^*}\in\Gamma_{\!\textcolor{sthlmRed}{1}}^{*}\)
attains the outer maximum in \eqref{eq:geometry-max-icmlclean}.

Substituting \eqref{eq:geometry-max-icmlclean} into \eqref{eq:greedy-x1-start-icmlclean} yields
\[
\argmax_{d_{\textcolor{sthlmRed}{1}}\in\mathcal{D}_{\textcolor{sthlmRed}{1}}}
v^*_{\mathrm{seq}}\!\bigl(\tau_{\mathrm{seq}}(x_{\textcolor{sthlmRed}{1}},d_{\textcolor{sthlmRed}{1}})\bigr)
=
\argmax_{d_{\textcolor{sthlmRed}{1}}\in\mathcal{D}_{\textcolor{sthlmRed}{1}}}
\max_{\Gamma_{\!\textcolor{sthlmRed}{2}}\in\Gamma_{\!\textcolor{sthlmRed}{1}}^{*}}
\mathtt{Val}_{\Gamma_{\!\textcolor{sthlmRed}{2}}}\!\bigl(\tau_{\mathrm{seq}}(x_{\textcolor{sthlmRed}{1}},d_{\textcolor{sthlmRed}{1}})\bigr),
\]
which is the claim.

Finally, the \(\argmax\) over \(d_{\textcolor{sthlmRed}{1}}\) is attained since \(\mathcal{D}_{\textcolor{sthlmRed}{1},t}\) is compact (Lemma~\ref{lem:compact_camera_ready}) and
\(d_{\textcolor{sthlmRed}{1}}\mapsto \tau_{\mathrm{seq}}(x_{\textcolor{sthlmRed}{1}},d_{\textcolor{sthlmRed}{1}})\) is affine, hence
\(d_{\textcolor{sthlmRed}{1}}\mapsto \mathtt{Val}_{\Gamma_{\!\textcolor{sthlmRed}{2}}}(\tau_{\mathrm{seq}}(x_{\textcolor{sthlmRed}{1}},d_{\textcolor{sthlmRed}{1}}))\) is continuous for each fixed \(\Gamma_{\!\textcolor{sthlmRed}{2}}\), and the pointwise maximum over \(\Gamma_{\!\textcolor{sthlmRed}{1}}^{*}\) is upper semicontinuous; therefore a maximiser exists by Weierstrass.
\end{proof}

\subsection{Proof of Lemma \ref{lem:greedy_x2}}

\begin{proof}
Fix \(x_{\textcolor{sthlmRed}{2}}\) at sub-stage \((\textcolor{sthlmRed}{2},t)\).
The greedy step minimises
\(
\rho_{\mathrm{seq}}(x_{\textcolor{sthlmRed}{2}},d_{\textcolor{sthlmRed}{2}})
+\gamma\,v^*_{\mathrm{seq}}(\tau_{\mathrm{seq}}(x_{\textcolor{sthlmRed}{2}},d_{\textcolor{sthlmRed}{2}}))
\)
over \(d_{\textcolor{sthlmRed}{2}}\in\mathcal{D}_{\textcolor{sthlmRed}{2},t}\).
Using the cached representation at \((\textcolor{sthlmRed}{1},t+1)\),
\[
v^*_{\mathrm{seq}}(x_{\textcolor{sthlmRed}{1}})=
\max_{\Gamma'_{\!\textcolor{sthlmRed}{2}}\in\Gamma'_{\!\textcolor{sthlmRed}{1}}}
\sum_{h_{\textcolor{sthlmRed}{2}}}\min_{\alpha\in\Gamma'_{\!\textcolor{sthlmRed}{2}}}
\langle b_{h_{\textcolor{sthlmRed}{2}}}(x_{\textcolor{sthlmRed}{1}}),\alpha\rangle,
\]
and unfolding the one-step recursion at \((\textcolor{sthlmRed}{2},t)\) (reward plus discounted continuation) yields exactly the
back-projected vectors \(q^{\alpha}_{a_{\textcolor{sthlmRed}{2}},z_{\textcolor{sthlmRed}{2}}}\) used in \Cref{eq:LP2}.
Introduce epigraph variables \(v_{h_{\textcolor{sthlmRed}{2}}}\) for the value of each history slice and
\(w^{\Gamma'_{\!\textcolor{sthlmRed}{2}}}_{h_{\textcolor{sthlmRed}{2}},z_{\textcolor{sthlmRed}{2}}}\) for each envelope/observation contribution.
Because \(d_{\textcolor{sthlmRed}{2}}\) enters linearly through \(\tau_{\mathrm{seq}}\), minimising over the simplex-valued rule
\(d_{\textcolor{sthlmRed}{2}}(\,\cdot\,\mid h_{\textcolor{sthlmRed}{2}})\) is enforced by the pointwise constraints
\(
v_{h_{\textcolor{sthlmRed}{2}}}\le \sum_{\Gamma'_{\!\textcolor{sthlmRed}{2}}}\sum_{z_{\textcolor{sthlmRed}{2}}}
w^{\Gamma'_{\!\textcolor{sthlmRed}{2}}}_{h_{\textcolor{sthlmRed}{2}},z_{\textcolor{sthlmRed}{2}}}
\ \forall a_{\textcolor{sthlmRed}{2}}
\),
and the inner \(\min_{\alpha\in\Gamma'_{\!\textcolor{sthlmRed}{2}}}\) is enforced by
\(
w^{\Gamma'_{\!\textcolor{sthlmRed}{2}}}_{h_{\textcolor{sthlmRed}{2}},z_{\textcolor{sthlmRed}{2}}}
\le
\lambda(\Gamma'_{\!\textcolor{sthlmRed}{2}})\langle b_{h_{\textcolor{sthlmRed}{2}}},q^{\alpha}_{a_{\textcolor{sthlmRed}{2}},z_{\textcolor{sthlmRed}{2}}}\rangle
\ \forall \alpha\in\Gamma'_{\!\textcolor{sthlmRed}{2}}.
\)
The simplex constraint \(\sum_{\Gamma'_{\!\textcolor{sthlmRed}{2}}}\lambda(\Gamma'_{\!\textcolor{sthlmRed}{2}})=1,\ \lambda\ge 0\)
makes \(\lambda\) a mixture (dual selector) over cached envelopes.
Thus the greedy problem is equivalently \(\text{LP}_2\) (maximising \(\sum_{h_{\textcolor{sthlmRed}{2}}}v_{h_{\textcolor{sthlmRed}{2}}}\)),
and by strong LP duality (feasible and bounded since \(\Gamma'_{\!\textcolor{sthlmRed}{1}}\) is finite and all payoffs/vectors are bounded),
its dual is precisely the stated max--min form in the lemma.
\end{proof}

\section{Updating Value Function Representations}
\label{sec:updating_value_function_representations}

\begin{proposition}[Sub-stage \(\textcolor{sthlmRed}{1}\): witness-vector update]
\label{prop:gamma_update_x1}
Fix an occupancy \(x_{\textcolor{sthlmRed}{1}}\) at sub-stage \(\textcolor{sthlmRed}{1}\) and a cached envelope
\(\Gamma_{\!\textcolor{sthlmRed}{2}}\in\Gamma_{\!\textcolor{sthlmRed}{1}}\),
where \(\Gamma_{\!\textcolor{sthlmRed}{2}}\subset \mathbb{R}^{\mathcal{S}\times\mathcal{H}_{\textcolor{sthlmRed}{1}}\times\mathcal{A}_{\textcolor{sthlmRed}{1}}}\)
indexes vectors by \((s,h_{\textcolor{sthlmRed}{1}},a_{\textcolor{sthlmRed}{1}})\) at the intermediate sub-stage.
For any decision rule \(d_{\textcolor{sthlmRed}{1}}\), define the induced intermediate slices
\[
b^{d_{\textcolor{sthlmRed}{1}}}_{h_{\textcolor{sthlmRed}{2}}}(s,h_{\textcolor{sthlmRed}{1}},a_{\textcolor{sthlmRed}{1}})
\doteq
x_{\textcolor{sthlmRed}{1}}(s,h_{\textcolor{sthlmRed}{1}},h_{\textcolor{sthlmRed}{2}})
\,d_{\textcolor{sthlmRed}{1}}(a_{\textcolor{sthlmRed}{1}}\mid h_{\textcolor{sthlmRed}{1}}).
\]
For each \(h_{\textcolor{sthlmRed}{2}}\), choose an arbitrary minimiser (ties broken arbitrarily)
\[
\alpha^{h_{\textcolor{sthlmRed}{2}}}\in
\argmin_{\alpha\in\Gamma_{\!\textcolor{sthlmRed}{2}}}
\sum_{s,h_{\textcolor{sthlmRed}{1}},a_{\textcolor{sthlmRed}{1}}}
b^{d_{\textcolor{sthlmRed}{1}}}_{h_{\textcolor{sthlmRed}{2}}}(s,h_{\textcolor{sthlmRed}{1}},a_{\textcolor{sthlmRed}{1}})
\,\alpha(s,h_{\textcolor{sthlmRed}{1}},a_{\textcolor{sthlmRed}{1}}),
\]
and define the generated witness vector \(\beta_{d_{\textcolor{sthlmRed}{1}},\Gamma_{\!\textcolor{sthlmRed}{2}}}\in
\mathbb{R}^{\mathcal{S}\times\mathcal{H}_{\textcolor{sthlmRed}{1}}\times\mathcal{H}_{\textcolor{sthlmRed}{2}}}\) by
\[
\beta_{d_{\textcolor{sthlmRed}{1}},\Gamma_{\!\textcolor{sthlmRed}{2}}}(s,h_{\textcolor{sthlmRed}{1}},h_{\textcolor{sthlmRed}{2}})
\doteq
\sum_{a_{\textcolor{sthlmRed}{1}}}
d_{\textcolor{sthlmRed}{1}}(a_{\textcolor{sthlmRed}{1}}\mid h_{\textcolor{sthlmRed}{1}})
\,\alpha^{h_{\textcolor{sthlmRed}{2}}}(s,h_{\textcolor{sthlmRed}{1}},a_{\textcolor{sthlmRed}{1}}).
\]
Then, for every \(d_{\textcolor{sthlmRed}{1}}\),
\[
\mathtt{Val}_{\Gamma_{\!\textcolor{sthlmRed}{2}}}\!\bigl(\tau_{\mathrm{seq}}(x_{\textcolor{sthlmRed}{1}},d_{\textcolor{sthlmRed}{1}})\bigr)
=
\sum_{s,h_{\textcolor{sthlmRed}{1}},h_{\textcolor{sthlmRed}{2}}}
x_{\textcolor{sthlmRed}{1}}(s,h_{\textcolor{sthlmRed}{1}},h_{\textcolor{sthlmRed}{2}})
\,\beta_{d_{\textcolor{sthlmRed}{1}},\Gamma_{\!\textcolor{sthlmRed}{2}}}(s,h_{\textcolor{sthlmRed}{1}},h_{\textcolor{sthlmRed}{2}})
=
\langle x_{\textcolor{sthlmRed}{1}},\beta_{d_{\textcolor{sthlmRed}{1}},\Gamma_{\!\textcolor{sthlmRed}{2}}}\rangle.
\]
In particular, if \((d_{\textcolor{sthlmRed}{1}}^*,\Gamma_{\!\textcolor{sthlmRed}{2}}^*)\) is the maximiser selected by \Cref{lem:greedy_x1},
then \(\beta_{d_{\textcolor{sthlmRed}{1}}^*,\Gamma_{\!\textcolor{sthlmRed}{2}}^*}\) is a valid witness vector to update the boundary envelope at \(x_{\textcolor{sthlmRed}{1}}\).
\end{proposition}

\begin{proof}
Fix \(d_{\textcolor{sthlmRed}{1}}\) and \(\Gamma_{\!\textcolor{sthlmRed}{2}}\).
By the definition of \(\mathtt{Val}_{\Gamma_{\!\textcolor{sthlmRed}{2}}}\) at the intermediate occupancy,
\begin{align*}
\mathtt{Val}_{\Gamma_{\!\textcolor{sthlmRed}{2}}}\!\bigl(\tau_{\mathrm{seq}}(x_{\textcolor{sthlmRed}{1}},d_{\textcolor{sthlmRed}{1}})\bigr)
&=
\sum_{h_{\textcolor{sthlmRed}{2}}}
\min_{\alpha\in\Gamma_{\!\textcolor{sthlmRed}{2}}}
\sum_{s,h_{\textcolor{sthlmRed}{1}},a_{\textcolor{sthlmRed}{1}}}
b^{d_{\textcolor{sthlmRed}{1}}}_{h_{\textcolor{sthlmRed}{2}}}(s,h_{\textcolor{sthlmRed}{1}},a_{\textcolor{sthlmRed}{1}})
\,\alpha(s,h_{\textcolor{sthlmRed}{1}},a_{\textcolor{sthlmRed}{1}}).
\end{align*}
Since \(\Gamma_{\!\textcolor{sthlmRed}{2}}\) is finite, each inner minimum is attained; choose
\(\alpha^{h_{\textcolor{sthlmRed}{2}}}\) as an arbitrary minimiser (ties broken arbitrarily).
Substituting the definition of \(b^{d_{\textcolor{sthlmRed}{1}}}_{h_{\textcolor{sthlmRed}{2}}}\) yields
\begin{align*}
\mathtt{Val}_{\Gamma_{\!\textcolor{sthlmRed}{2}}}\!\bigl(\tau_{\mathrm{seq}}(x_{\textcolor{sthlmRed}{1}},d_{\textcolor{sthlmRed}{1}})\bigr)
&=
\sum_{h_{\textcolor{sthlmRed}{2}}}
\sum_{s,h_{\textcolor{sthlmRed}{1}},a_{\textcolor{sthlmRed}{1}}}
x_{\textcolor{sthlmRed}{1}}(s,h_{\textcolor{sthlmRed}{1}},h_{\textcolor{sthlmRed}{2}})
\,d_{\textcolor{sthlmRed}{1}}(a_{\textcolor{sthlmRed}{1}}\mid h_{\textcolor{sthlmRed}{1}})
\,\alpha^{h_{\textcolor{sthlmRed}{2}}}(s,h_{\textcolor{sthlmRed}{1}},a_{\textcolor{sthlmRed}{1}}).
\end{align*}
Reordering the finite sums and using the definition of
\(\beta_{d_{\textcolor{sthlmRed}{1}},\Gamma_{\!\textcolor{sthlmRed}{2}}}\) gives
\[
\sum_{s,h_{\textcolor{sthlmRed}{1}},h_{\textcolor{sthlmRed}{2}}}
x_{\textcolor{sthlmRed}{1}}(s,h_{\textcolor{sthlmRed}{1}},h_{\textcolor{sthlmRed}{2}})
\,\beta_{d_{\textcolor{sthlmRed}{1}},\Gamma_{\!\textcolor{sthlmRed}{2}}}(s,h_{\textcolor{sthlmRed}{1}},h_{\textcolor{sthlmRed}{2}})
=
\langle x_{\textcolor{sthlmRed}{1}},\beta_{d_{\textcolor{sthlmRed}{1}},\Gamma_{\!\textcolor{sthlmRed}{2}}}\rangle,
\]
which proves the claim.
\end{proof}

\begin{proposition}[Sub-stage \((\textcolor{sthlmRed}{2},t)\): witness-envelope generation]
\label{prop:Gamma2_update}
Fix a reachable \(x_{\textcolor{sthlmRed}{2}}\) and a cached family \(\Gamma_{\!\textcolor{sthlmRed}{1}}'\) from \((\textcolor{sthlmRed}{1},t+1)\).
Let \((\lambda^*,w^*,v^*)\) be an optimal solution of \Cref{eq:LP2}.
For each \((h_{\textcolor{sthlmRed}{2}},a_{\textcolor{sthlmRed}{2}},z_{\textcolor{sthlmRed}{2}},\Gamma_{\!\textcolor{sthlmRed}{2}}')\),
choose any minimiser
\[
\alpha^{h_{\textcolor{sthlmRed}{2}},a_{\textcolor{sthlmRed}{2}},z_{\textcolor{sthlmRed}{2}},\Gamma_{\!\textcolor{sthlmRed}{2}}'}
\in
\argmin_{\alpha\in\Gamma_{\!\textcolor{sthlmRed}{2}}'}
\left\langle b_{h_{\textcolor{sthlmRed}{2}}},\, q^{\alpha}_{a_{\textcolor{sthlmRed}{2}},z_{\textcolor{sthlmRed}{2}}}\right\rangle,
\]
and define, for each \((h_{\textcolor{sthlmRed}{2}},a_{\textcolor{sthlmRed}{2}})\), the vector
\[
\beta^{h_{\textcolor{sthlmRed}{2}},a_{\textcolor{sthlmRed}{2}}}
\doteq
\sum_{z_{\textcolor{sthlmRed}{2}}}\;
\sum_{\Gamma_{\!\textcolor{sthlmRed}{2}}'\in\Gamma_{\!\textcolor{sthlmRed}{1}}'}
\lambda^*(\Gamma_{\!\textcolor{sthlmRed}{2}}')\,
q^{\alpha^{h_{\textcolor{sthlmRed}{2}},a_{\textcolor{sthlmRed}{2}},z_{\textcolor{sthlmRed}{2}},\Gamma_{\!\textcolor{sthlmRed}{2}}'}}_{a_{\textcolor{sthlmRed}{2}},z_{\textcolor{sthlmRed}{2}}}.
\]
Finally, define the (finite) witness envelope at \(x_{\textcolor{sthlmRed}{2}}\) as
\[
\Gamma_{\!\textcolor{sthlmRed}{2},x_{\textcolor{sthlmRed}{2}}}
\doteq
\Bigl\{\beta^{h_{\textcolor{sthlmRed}{2}},a_{\textcolor{sthlmRed}{2}}} \;:\;
h_{\textcolor{sthlmRed}{2}}\in\mathcal{H}_{\textcolor{sthlmRed}{2},t},\;
a_{\textcolor{sthlmRed}{2}}\in\mathcal{A}_{\textcolor{sthlmRed}{2}}\Bigr\}
\subset \mathbb{R}^{\Omega_{\textcolor{sthlmRed}{2}}}.
\]
Then, for all \(h_{\textcolor{sthlmRed}{2}}\) and \(a_{\textcolor{sthlmRed}{2}}\),
\begin{align*}
\left\langle b_{h_{\textcolor{sthlmRed}{2}}},\,\beta^{h_{\textcolor{sthlmRed}{2}},a_{\textcolor{sthlmRed}{2}}}\right\rangle
&=
\sum_{\Gamma_{\!\textcolor{sthlmRed}{2}}'\in\Gamma_{\!\textcolor{sthlmRed}{1}}'}
\lambda^*(\Gamma_{\!\textcolor{sthlmRed}{2}}')\,
\sum_{z_{\textcolor{sthlmRed}{2}}}
\min_{\alpha\in\Gamma_{\!\textcolor{sthlmRed}{2}}'}
\left\langle b_{h_{\textcolor{sthlmRed}{2}}},\, q^{\alpha}_{a_{\textcolor{sthlmRed}{2}},z_{\textcolor{sthlmRed}{2}}}\right\rangle,
\end{align*}
and therefore \(\Gamma_{\!\textcolor{sthlmRed}{2},x_{\textcolor{sthlmRed}{2}}}\) is a valid cached envelope for evaluating
\(\mathtt{Val}_{\Gamma_{\!\textcolor{sthlmRed}{2},x_{\textcolor{sthlmRed}{2}}}}(x_{\textcolor{sthlmRed}{2}})\), as used at predecessor boundary sub-stages (via \Cref{eq:LP1}).
\end{proposition}

\begin{proof}
Because each \(\Gamma_{\!\textcolor{sthlmRed}{2}}'\) is finite, the \(\argmin\) is non-empty and the chosen
\(\alpha^{h_{\textcolor{sthlmRed}{2}},a_{\textcolor{sthlmRed}{2}},z_{\textcolor{sthlmRed}{2}},\Gamma_{\!\textcolor{sthlmRed}{2}}'}\)
attains the minimum.
Substituting these minimisers into the definition of \(\beta^{h_{\textcolor{sthlmRed}{2}},a_{\textcolor{sthlmRed}{2}}}\) gives the equality by linearity in \(\lambda^*\) and summation over \(z_{\textcolor{sthlmRed}{2}}\).
\end{proof}

\begin{theorem}[Monotone cache augmentation]
\label{thm:monotone_gamma_update}
Fix a sub-stage \(\textcolor{sthlmRed}{i}\in\{1,2\}\) and let \(\Gamma_{\!\textcolor{sthlmRed}{1}}\) be the current cached family.
Define the represented value function
\[
\widehat v_{\mathrm{seq}}(x_{\textcolor{sthlmRed}{i}})
\doteq
\max_{\Gamma_{\!\textcolor{sthlmRed}{2}}\in \Gamma_{\!\textcolor{sthlmRed}{1}}}
\mathtt{Val}_{\Gamma_{\!\textcolor{sthlmRed}{2}}}(x_{\textcolor{sthlmRed}{i}}).
\]
Let \(\Gamma_{\!\textcolor{sthlmRed}{2},x_{\textcolor{sthlmRed}{i}}}\) be a witness envelope generated at \(x_{\textcolor{sthlmRed}{i}}\)
(by \Cref{prop:gamma_update_x1} when \(\textcolor{sthlmRed}{i}=1\), and by \Cref{prop:Gamma2_update} when \(\textcolor{sthlmRed}{i}=2\)),
and define the augmented family \(\Gamma_{\!\textcolor{sthlmRed}{1}}^{+}\doteq \Gamma_{\!\textcolor{sthlmRed}{1}}\cup\{\Gamma_{\!\textcolor{sthlmRed}{2},x_{\textcolor{sthlmRed}{i}}}\}\),
with the associated value \(\widehat v^{+}_{\mathrm{seq}}\) defined as above with \(\Gamma_{\!\textcolor{sthlmRed}{1}}^{+}\).
Then, for every occupancy \(x_{\textcolor{sthlmRed}{i}}\),
\[
\widehat v^{+}_{\mathrm{seq}}(x_{\textcolor{sthlmRed}{i}})\ \ge\ \widehat v_{\mathrm{seq}}(x_{\textcolor{sthlmRed}{i}}),
\]
and in particular
\[
\widehat v^{+}_{\mathrm{seq}}(x_{\textcolor{sthlmRed}{i}})
=
\max\!\Bigl\{\widehat v_{\mathrm{seq}}(x_{\textcolor{sthlmRed}{i}}),\ 
\mathtt{Val}_{\Gamma_{\!\textcolor{sthlmRed}{2},x_{\textcolor{sthlmRed}{i}}}}(x_{\textcolor{sthlmRed}{i}})\Bigr\}.
\]
Moreover, if \(\mathtt{Val}_{\Gamma_{\!\textcolor{sthlmRed}{2},x_{\textcolor{sthlmRed}{i}}}}(x_{\textcolor{sthlmRed}{i}})>\widehat v_{\mathrm{seq}}(x_{\textcolor{sthlmRed}{i}})\),
then \(\widehat v^{+}_{\mathrm{seq}}(x_{\textcolor{sthlmRed}{i}})>\widehat v_{\mathrm{seq}}(x_{\textcolor{sthlmRed}{i}})\).
\end{theorem}

\begin{proof}
By definition, \(\Gamma_{\!\textcolor{sthlmRed}{1}}\subseteq \Gamma_{\!\textcolor{sthlmRed}{1}}^{+}\), hence
\[
\widehat v^{+}_{\mathrm{seq}}(x_{\textcolor{sthlmRed}{i}})
=
\max_{\Gamma_{\!\textcolor{sthlmRed}{2}}\in \Gamma_{\!\textcolor{sthlmRed}{1}}^{+}}
\mathtt{Val}_{\Gamma_{\!\textcolor{sthlmRed}{2}}}(x_{\textcolor{sthlmRed}{i}})
\ \ge\
\max_{\Gamma_{\!\textcolor{sthlmRed}{2}}\in \Gamma_{\!\textcolor{sthlmRed}{1}}}
\mathtt{Val}_{\Gamma_{\!\textcolor{sthlmRed}{2}}}(x_{\textcolor{sthlmRed}{i}})
=
\widehat v_{\mathrm{seq}}(x_{\textcolor{sthlmRed}{i}}).
\]
The displayed identity follows because \(\Gamma_{\!\textcolor{sthlmRed}{1}}^{+}\) adds exactly one candidate envelope,
and strict improvement holds whenever the added candidate attains a strictly larger value at \(x_{\textcolor{sthlmRed}{i}}\).
\end{proof}

\section{From Exponential to Polynomial Complexity --- Proof of Theorem \ref{thm:complexity}}
\label{app:sec:complexity}

\begin{proof}[Proof of \Cref{thm:complexity}]
We count the number of decision variables and linear constraints in the two LPs and use that linear programs are solvable
in time polynomial in their input size (bit-length).

\paragraph{Sub-stage \((\textcolor{sthlmRed}{1},t)\): LP\(_1\).}
Fix a cached envelope \(\Gamma_{\!\textcolor{sthlmRed}{2}}\in\Gamma_{\!\textcolor{sthlmRed}{1}}\).
LP\(_1\) optimises over \((d_{\textcolor{sthlmRed}{1}},w)\).
A decision rule \(d_{\textcolor{sthlmRed}{1}}(\cdot\mid h_{\textcolor{sthlmRed}{1}})\) is a simplex vector for each
\(h_{\textcolor{sthlmRed}{1}}\in\mathcal{H}_{\textcolor{sthlmRed}{1}}\), hence has
\(\Theta(|\mathcal{H}_{\textcolor{sthlmRed}{1}}||\mathcal{A}_{\textcolor{sthlmRed}{1}}|)\) scalar coordinates.
The epigraph variables \(w_{h_{\textcolor{sthlmRed}{2}}}\) add \(|\mathcal{H}_{\textcolor{sthlmRed}{2}}|\) scalars, so the total variable count is
\[
\pmb{O}(|\mathcal{H}_{\textcolor{sthlmRed}{1}}||\mathcal{A}_{\textcolor{sthlmRed}{1}}|+|\mathcal{H}_{\textcolor{sthlmRed}{2}}|).
\]
The main inequalities are
\(w_{h_{\textcolor{sthlmRed}{2}}}\le \langle b_{h_{\textcolor{sthlmRed}{2}}},\alpha\rangle\)
for every \(h_{\textcolor{sthlmRed}{2}}\in\mathcal{H}_{\textcolor{sthlmRed}{2}}\) and every \(\alpha\in\Gamma_{\!\textcolor{sthlmRed}{2}}\),
hence \(\pmb{O}(|\mathcal{H}_{\textcolor{sthlmRed}{2}}||\Gamma_{\!\textcolor{sthlmRed}{2}}|)\) constraints (plus simplex constraints as well).
Since we solve one such LP for each \(\Gamma_{\!\textcolor{sthlmRed}{2}}\in\Gamma_{\!\textcolor{sthlmRed}{1}}\),
sub-stage \((\textcolor{sthlmRed}{1},t)\) requires \(|\Gamma_{\!\textcolor{sthlmRed}{1}}|\) LPs of this size.

\paragraph{Sub-stage \((\textcolor{sthlmRed}{2},t)\): LP\(_2\).}
LP\(_2\) uses the variables \(\lambda\in\Delta(\Gamma'_{\!\textcolor{sthlmRed}{1}})\),
\(v_{h_{\textcolor{sthlmRed}{2}}}\) for \(h_{\textcolor{sthlmRed}{2}}\in\mathcal{H}_{\textcolor{sthlmRed}{2}}\),
and \(w^{\Gamma'_{\!\textcolor{sthlmRed}{2}}}_{h_{\textcolor{sthlmRed}{2}},z_{\textcolor{sthlmRed}{2}}}\) for
\((h_{\textcolor{sthlmRed}{2}},z_{\textcolor{sthlmRed}{2}},\Gamma'_{\!\textcolor{sthlmRed}{2}})\in
\mathcal{H}_{\textcolor{sthlmRed}{2}}\times\mathcal{Z}_{\textcolor{sthlmRed}{2}}\times\Gamma'_{\!\textcolor{sthlmRed}{1}}\).
Therefore, the number of scalar variables is
\[
|\Gamma'_{\!\textcolor{sthlmRed}{1}}|
+|\mathcal{H}_{\textcolor{sthlmRed}{2}}|
+|\mathcal{H}_{\textcolor{sthlmRed}{2}}||\mathcal{Z}_{\textcolor{sthlmRed}{2}}||\Gamma'_{\!\textcolor{sthlmRed}{1}}|
=
\pmb{O}\!\bigl(|\mathcal{H}_{\textcolor{sthlmRed}{2}}||\mathcal{Z}_{\textcolor{sthlmRed}{2}}||\Gamma'_{\!\textcolor{sthlmRed}{1}}|\bigr).
\]
For constraints:
(i) the ``min over \(a_{\textcolor{sthlmRed}{2}}\)'' constraints
\(v_{h_{\textcolor{sthlmRed}{2}}}\le \sum_{\Gamma'}\sum_{z_{\textcolor{sthlmRed}{2}}} w^{\Gamma'}_{h_{\textcolor{sthlmRed}{2}},z_{\textcolor{sthlmRed}{2}}}\)
appear for each \((h_{\textcolor{sthlmRed}{2}},a_{\textcolor{sthlmRed}{2}})\), giving
\(\pmb{O}(|\mathcal{H}_{\textcolor{sthlmRed}{2}}||\mathcal{A}_{\textcolor{sthlmRed}{2}}|)\) constraints;
(ii) the envelope constraints
\(w^{\Gamma'}_{h_{\textcolor{sthlmRed}{2}},z_{\textcolor{sthlmRed}{2}}}
\le \lambda(\Gamma')\langle b_{h_{\textcolor{sthlmRed}{2}}},q^{\alpha}_{a_{\textcolor{sthlmRed}{2}},z_{\textcolor{sthlmRed}{2}}}\rangle\)
appear for each \((h_{\textcolor{sthlmRed}{2}},a_{\textcolor{sthlmRed}{2}},z_{\textcolor{sthlmRed}{2}},\Gamma',\alpha)\) with \(\alpha\in\Gamma'\).
Thus their count is
\[
\pmb{O}\!\Bigl(
|\mathcal{H}_{\textcolor{sthlmRed}{2}}||\mathcal{A}_{\textcolor{sthlmRed}{2}}||\mathcal{Z}_{\textcolor{sthlmRed}{2}}|
\sum_{\Gamma'\in\Gamma'_{\!\textcolor{sthlmRed}{1}}} |\Gamma'|
\Bigr)
=
\pmb{O}\!\bigl(
|\mathcal{H}_{\textcolor{sthlmRed}{2}}||\mathcal{A}_{\textcolor{sthlmRed}{2}}||\mathcal{Z}_{\textcolor{sthlmRed}{2}}|
|\Gamma'_{\!\textcolor{sthlmRed}{1}}|\,|\bar{\Gamma}'_{\!\textcolor{sthlmRed}{2}}|
\bigr),
\]
where \(|\bar{\Gamma}'_{\!\textcolor{sthlmRed}{2}}|=\max_{\Gamma'\in\Gamma'_{\!\textcolor{sthlmRed}{1}}}|\Gamma'|\).
The simplex constraints for \(\lambda\) add \(\pmb{O}(|\Gamma'_{\!\textcolor{sthlmRed}{1}}|)\) further constraints.

\paragraph{Polynomial-time conclusion.}
Both LP\(_1\) and LP\(_2\) have sizes polynomial in the explicit parameters
\(|\mathcal{H}_{\textcolor{sthlmRed}{1}}|,|\mathcal{H}_{\textcolor{sthlmRed}{2}}|,|\mathcal{A}_{\textcolor{sthlmRed}{1}}|,
|\mathcal{A}_{\textcolor{sthlmRed}{2}}|,|\mathcal{Z}_{\textcolor{sthlmRed}{2}}|,|\Gamma_{\!\textcolor{sthlmRed}{1}}|,
|\Gamma'_{\!\textcolor{sthlmRed}{1}}|,|\bar{\Gamma}'_{\!\textcolor{sthlmRed}{2}}|\).
Since sub-stage \((\textcolor{sthlmRed}{1},t)\) solves \(|\Gamma_{\!\textcolor{sthlmRed}{1}}|\) instances of LP\(_1\) and sub-stage \((\textcolor{sthlmRed}{2},t)\) solves one instance of LP\(_2\),
and since linear programming is solvable in time polynomial in its input size (bit-length),
computing the sequential backup is polynomial in the explicit representation size.
\end{proof}

\paragraph{LP complexity comparison to the state of the art.}
We compare the linear programs underlying (i) the OMG greedy-update view \citet{DelBufDibSaf-DGAA-23},
(ii) the cached-envelope greedy LP view \citet{escudie2025varepsilonoptimally},
and (iii) our sequential backup (LP$_{\textcolor{sthlmRed}{1}}$/LP$_{\textcolor{sthlmRed}{2}}$).
In the OMG view, pessimism is enforced against all deterministic opponent decision rules
$d_{\textcolor{sthlmRed}{2}}:\mathcal H_{\textcolor{sthlmRed}{2}}\to\mathcal A_{\textcolor{sthlmRed}{2}}$.
Since the number of such mappings is $|\mathcal A_{\textcolor{sthlmRed}{2}}|^{|\mathcal H_{\textcolor{sthlmRed}{2}}|}$, the resulting LP description carries an exponential constraint family.
For the polynomial LPs, we report input-size bounds in terms of the same explicit representation objects used throughout the paper:
a cached family $\Gamma_{\textcolor{sthlmRed}{1}}$ of envelopes $\Gamma_{\textcolor{sthlmRed}{2}}\in\Gamma_{\textcolor{sthlmRed}{1}}$, and the maximal envelope size
\[
|\Gamma_{\textcolor{sthlmRed}{2}}^*| \doteq \max_{\Gamma_{\textcolor{sthlmRed}{2}}\in\Gamma_{\textcolor{sthlmRed}{1}}}|\Gamma_{\textcolor{sthlmRed}{2}}|.
\]
Table~\ref{tab:lp-complexity-comparison} reports the asymptotic numbers of scalar variables and linear constraints, as well as the number of LP solves per backup step.

\begin{table}[t]
\centering
\small
\setlength{\tabcolsep}{6pt}
\renewcommand{\arraystretch}{1.18}
\begin{tabular}{p{2.8cm} c c c}
\toprule
\textbf{Update step} & \textbf{\# LPs} & \textbf{\# variables} & \textbf{\# constraints} \\
\midrule

\textbf{Ours: LP$_{\textcolor{sthlmRed}{1}}$ } &
$|\Gamma_{\textcolor{sthlmRed}{1}}|$ &
$O(|\mathcal H_{\textcolor{sthlmRed}{1}}||\mathcal A_{\textcolor{sthlmRed}{1}}|+|\mathcal H_{\textcolor{sthlmRed}{2}}|)$ &
$O(|\mathcal H_{\textcolor{sthlmRed}{2}}|\cdot|\Gamma_{\textcolor{sthlmRed}{2}}|)$ \ (+ simplex) \\

\textbf{Ours: LP$_{\textcolor{sthlmRed}{2}}$ } &
$1$ &
$O(|\mathcal H_{\textcolor{sthlmRed}{2}}||\mathcal Z_{\textcolor{sthlmRed}{2}}|\cdot|\Gamma'_{\textcolor{sthlmRed}{1}}|)$ &
$O\!\bigl(|\mathcal H_{\textcolor{sthlmRed}{2}}||\mathcal A_{\textcolor{sthlmRed}{2}}|
+|\mathcal H_{\textcolor{sthlmRed}{2}}||\mathcal A_{\textcolor{sthlmRed}{2}}||\mathcal Z_{\textcolor{sthlmRed}{2}}|\cdot|\Gamma_{\textcolor{sthlmRed}{1}}|\cdot|\Gamma^*_{\textcolor{sthlmRed}{2}}|\bigr)$ \ (+ simplex) \\

\midrule

\citep{DelBufDibSaf-DGAA-23} &
$1$ &
(rule variables $+$ epigraph) &
\textbf{Exponential:}\;\; $O\!\bigl(|\mathcal A_{\textcolor{sthlmRed}{2}}|^{|\mathcal H_{\textcolor{sthlmRed}{2}}|}\bigr)$ \\

\citep{escudie2025varepsilonoptimally} &
$1$ &
$O\!\bigl(|\Gamma_{\textcolor{sthlmRed}{1}}|\cdot|\mathcal H_{\textcolor{sthlmRed}{2}}|\cdot|\mathcal A_{\textcolor{sthlmRed}{2}}|\cdot|\mathcal Z_{\textcolor{sthlmRed}{2}}|\bigr)$ &
$O\!\bigl(|\Gamma_{\textcolor{sthlmRed}{1}}|\cdot|\Gamma^{*}_{\textcolor{sthlmRed}{2}}|\cdot|\mathcal H_{\textcolor{sthlmRed}{2}}|\cdot|\mathcal A_{\textcolor{sthlmRed}{2}}|\cdot|\mathcal Z_{\textcolor{sthlmRed}{2}}|\bigr)$ \\
\bottomrule
\end{tabular}
\caption{LP complexity comparison (input-size view).
The OMG greedy-update LP has constraints indexed by opponent decision rules, yielding an exponential dependence on $|\mathcal H_{\textcolor{sthlmRed}{2}}|$ \citep{DelBufDibSaf-DGAA-23}.
Both the cached-envelope greedy LP \citep{escudie2025varepsilonoptimally} and our sequential backup scale polynomially in the same explicit representation objects, but they expose different optimisation structure.}
\label{tab:lp-complexity-comparison}
\end{table}

\paragraph{Comparative analysis: \citet{DelBufDibSaf-DGAA-23} vs.\ our sequential backup.}
The OMG greedy-update formulation \citet{DelBufDibSaf-DGAA-23} enforces pessimism against \emph{all} deterministic opponent decision rules
$d_{\textcolor{sthlmRed}{2}}:\mathcal H_{\textcolor{sthlmRed}{2}}\to \mathcal A_{\textcolor{sthlmRed}{2}}$.
Because there are $|\mathcal A_{\textcolor{sthlmRed}{2}}|^{|\mathcal H_{\textcolor{sthlmRed}{2}}|}$ such mappings, the LP description contains
$O(|\mathcal A_{\textcolor{sthlmRed}{2}}|^{|\mathcal H_{\textcolor{sthlmRed}{2}}|})$ constraints (Table~\ref{tab:lp-complexity-comparison}), which is exponential in the size of the opponent’s private-history space.
Our sequential backup removes this bottleneck by never indexing constraints by opponent decision rules:
LP$_{\textcolor{sthlmRed}{1}}$ constrains only against finitely many vectors in a cached envelope $\Gamma_{\textcolor{sthlmRed}{2}}$, and LP$_{\textcolor{sthlmRed}{2}}$ aggregates only over finitely many cached envelopes.
Hence both LP$_{\textcolor{sthlmRed}{1}}$ and LP$_{\textcolor{sthlmRed}{2}}$ admit polynomial-size descriptions in the explicit parameters
$|\mathcal H_{\textcolor{sthlmRed}{1}}|,|\mathcal H_{\textcolor{sthlmRed}{2}}|,|\mathcal A_{\textcolor{sthlmRed}{1}}|,|\mathcal A_{\textcolor{sthlmRed}{2}}|,|\mathcal Z_{\textcolor{sthlmRed}{2}}|,|\Gamma_{\textcolor{sthlmRed}{1}}|,|\Gamma^*_{\textcolor{sthlmRed}{2}}|$.
In particular, for any $|\mathcal A_{\textcolor{sthlmRed}{2}}|\ge 2$, any polynomially bounded LP description is asymptotically strictly smaller than
$|\mathcal A_{\textcolor{sthlmRed}{2}}|^{|\mathcal H_{\textcolor{sthlmRed}{2}}|}$ as $|\mathcal H_{\textcolor{sthlmRed}{2}}|\to\infty$.
This is the precise sense in which our sequentialisation turns an exponentially described update into a polynomially described backup operator.

\paragraph{Comparative analysis: \citet{escudie2025varepsilonoptimally} vs.\ our sequential backup.}
Both approaches avoid the OMG exponential constraint family by optimising against the same explicit cached representation $(\Gamma_{\textcolor{sthlmRed}{1}},\Gamma_{\textcolor{sthlmRed}{2}})$.
The difference is how the optimisation couples indices.
In \citet{escudie2025varepsilonoptimally}, a single monolithic greedy LP couples, within one program,
the outer cache index $\Gamma_{\textcolor{sthlmRed}{1}}$, the inner value index $\alpha\in\Gamma_{\textcolor{sthlmRed}{2}}$, and the opponent action--observation indices $(a_{\textcolor{sthlmRed}{2}},z_{\textcolor{sthlmRed}{2}})$,
yielding a constraint count
$O(|\Gamma_{\textcolor{sthlmRed}{1}}|\cdot|\Gamma^{*}_{\textcolor{sthlmRed}{2}}|\cdot|\mathcal H_{\textcolor{sthlmRed}{2}}|\cdot|\mathcal A_{\textcolor{sthlmRed}{2}}|\cdot|\mathcal Z_{\textcolor{sthlmRed}{2}}|)$
(Table~\ref{tab:lp-complexity-comparison}).

Our sequential backup factorises this coupling into two sub-stages.
LP$_{\textcolor{sthlmRed}{1}}$ is local to a fixed envelope $\Gamma_{\textcolor{sthlmRed}{2}}$ and therefore eliminates the global $(a_{\textcolor{sthlmRed}{2}},z_{\textcolor{sthlmRed}{2}})$ coupling from the $\alpha\in\Gamma_{\textcolor{sthlmRed}{2}}$ enumeration:
its constraints scale only as $O(|\mathcal H_{\textcolor{sthlmRed}{2}}|\cdot|\Gamma_{\textcolor{sthlmRed}{2}}|)$, and it is solved once per candidate envelope.
The global $(a_{\textcolor{sthlmRed}{2}},z_{\textcolor{sthlmRed}{2}})$ coupling is paid only once, at the aggregation stage LP$_{\textcolor{sthlmRed}{2}}$, whose constraints scale as
$O(|\mathcal H_{\textcolor{sthlmRed}{2}}||\mathcal A_{\textcolor{sthlmRed}{2}}||\mathcal Z_{\textcolor{sthlmRed}{2}}|\cdot|\Gamma_{\textcolor{sthlmRed}{1}}|\cdot|\Gamma^*_{\textcolor{sthlmRed}{2}}|)$.
Consequently, relative to the monolithic formulation of \citet{escudie2025varepsilonoptimally}, our update strictly reduces the coupling burden by ensuring that the
$|\mathcal A_{\textcolor{sthlmRed}{2}}||\mathcal Z_{\textcolor{sthlmRed}{2}}|$ factor appears only in the single aggregation LP and not inside each per-envelope optimisation.
This is precisely the computational gain delivered by the sequential factorisation: it preserves the same explicit representation objects while decoupling local envelope selection (LP$_{\textcolor{sthlmRed}{1}}$) from global aggregation (LP$_{\textcolor{sthlmRed}{2}}$).

\section{Sequential Lossless Reduction --- Proof of Theorem \ref{thm:sequential_lossless}}
\label{app:proof:sequential_lossless}

\begin{proof}[Proof of \Cref{thm:sequential_lossless}]
Let \(\mathcal{M}'\) denote the (simultaneous) TI-zs-SG of \citet{escudie2025varepsilonoptimally},
which is a lossless reduction of the zs-POSG \(\mathcal{M}\).

\paragraph{Step 1: \(\mathcal{M}'_{\mathrm{seq}}\) is a pure stage-unrolling of \(\mathcal{M}'\).}
The construction of \(\mathcal{M}'_{\mathrm{seq}}\) only unzips each stage \(t\) of \(\mathcal{M}'\) into two internal sub-stages:
at \((\textcolor{sthlmRed}{1},t)\) player~\(\textcolor{sthlmRed}{1}\) selects \(d_{\textcolor{sthlmRed}{1},t}\), the central state is pushed forward deterministically, the reward is \(0\), and \(\gamma_{\textcolor{sthlmRed}{1}}=1\);
at \((\textcolor{sthlmRed}{2},t)\) player~\(\textcolor{sthlmRed}{2}\) selects \(d_{\textcolor{sthlmRed}{2},t}\), the original TI kernel and occupancy-set updates of \(\mathcal{M}'\) are applied, the stage reward is realised, and \(\gamma_{\textcolor{sthlmRed}{2}}=\gamma\).
Define \(\Phi\) as the path-embedding that inserts the intermediate node \(x_{\textcolor{sthlmRed}{2},t}\) after each boundary node \(x_{\textcolor{sthlmRed}{1},t}\),
and \(\Psi\) as the projection that composes the two sub-stages at each \(t\); then \(\Psi\circ\Phi=\mathrm{Id}\), and \(\Phi\circ\Psi\) is identity up to the inserted bookkeeping nodes.
Consequently, \(\Phi\) and \(\Psi\) induce a bijection between \emph{admissible} policy profiles in \(\mathcal{M}'\) and in \(\mathcal{M}'_{\mathrm{seq}}\),
where admissibility means measurability w.r.t.\ the same decision-rule histories/occupancy sets as in \(\mathcal{M}'\) (so \(x_{\textcolor{sthlmRed}{2},t}\) is not an additional signal);
under corresponding admissible profiles, the induced law on stage-boundary states and the total discounted return coincide, hence best responses and equilibrium values are preserved.

\paragraph{Step 2: Losslessness w.r.t.\ \(\mathcal{M}\).}
Since \(\mathcal{M}'\) is lossless for \(\mathcal{M}\) and \(\mathcal{M}'_{\mathrm{seq}}\) is equivalent to \(\mathcal{M}'\) over admissible profiles via \(\Phi,\Psi\),
\(\mathcal{M}'_{\mathrm{seq}}\) is a lossless reduction of \(\mathcal{M}\).

\paragraph{Step 3: Recovering the simultaneous TI-zs-SG optimality equations.}
For any boundary state \(x_{\textcolor{sthlmRed}{1},t}\),
\[
v^*_{\mathrm{seq}}(x_{\textcolor{sthlmRed}{1},t})
=
\max_{d_{\textcolor{sthlmRed}{1},t}}
\min_{d_{\textcolor{sthlmRed}{2},t}}
\Bigl[
\rho_{\mathcal{M}'}(x_{\textcolor{sthlmRed}{1},t},d_{\textcolor{sthlmRed}{1},t},d_{\textcolor{sthlmRed}{2},t})
+\gamma\,v^*_{\mathrm{seq}}\!\bigl(\tau_{\mathcal{M}'}(x_{\textcolor{sthlmRed}{1},t},d_{\textcolor{sthlmRed}{1},t},d_{\textcolor{sthlmRed}{2},t})\bigr)
\Bigr],
\]
because \(x_{\textcolor{sthlmRed}{2},t}=\tau_{\mathrm{seq}}(x_{\textcolor{sthlmRed}{1},t},d_{\textcolor{sthlmRed}{1},t})\) is a bookkeeping node with \(\rho_{\mathrm{seq}}(x_{\textcolor{sthlmRed}{1},t},d_{\textcolor{sthlmRed}{1},t})=0\) and \(\gamma_{\textcolor{sthlmRed}{1}}=1\), and the \((\textcolor{sthlmRed}{2},t)\) update applies exactly the one-stage primitives of \(\mathcal{M}'\).
This is precisely the simultaneous TI-zs-SG optimality equation of \citet{escudie2025varepsilonoptimally}.
\end{proof}

\section{Sequential Transition-Independent Stochastic Games}
\label{app:sec:seq:ti:zs:sg}

This appendix formalises the \emph{sequential} transition-independent reduction used throughout the paper.
It follows the simultaneous TI-zs-SG construction of \citet{escudie2025varepsilonoptimally} and refines only the
\emph{timing} of the central recursion: each original stage \(t\) is split into two sub-stages \((\textcolor{sthlmRed}{1},t)\) and \((\textcolor{sthlmRed}{2},t)\),
so that player~\textcolor{sthlmRed}{1} commits first and player~\textcolor{sthlmRed}{2} completes the stage.
We present the reduction through two statistics:
(i) a value-sufficient \emph{sequential occupancy} for the uninformed recursion, and
(ii) player-specific \emph{occupancy families} supporting equilibrium-consistent extraction.

\paragraph{Timing.}
Sub-stages are indexed by \((\textcolor{sthlmRed}{i},t)\) with \(\textcolor{sthlmRed}{i}\in\{\textcolor{sthlmRed}{1},\textcolor{sthlmRed}{2}\}\) and \(t\in\{0,\ldots,\ell-1\}\).
Define \(\mathrm{next}(\textcolor{sthlmRed}{1},t)\doteq (\textcolor{sthlmRed}{2},t)\) and \(\mathrm{next}(\textcolor{sthlmRed}{2},t)\doteq (\textcolor{sthlmRed}{1},t+1)\).
The terminal stage boundary is \((\textcolor{sthlmRed}{1},\ell)\).

\paragraph{Sequential occupancies and occupancy families.}
At each sub-stage \((\textcolor{sthlmRed}{i},t)\), the (hidden) global state is a sequential occupancy \(x_{\textcolor{sthlmRed}{i},t}\),
the value-sufficient statistic manipulated by the uninformed recursion.
Player \(\textcolor{sthlmRed}{i}\) does not observe \(x_{\textcolor{sthlmRed}{i},t}\); instead it reasons over a local state \(\mathbf{o}_{\textcolor{sthlmRed}{i},t}\),
called an occupancy family.
The family \(\mathbf{o}_{\textcolor{sthlmRed}{i},t}\) collects exactly the sequential occupancies \(x_{\textcolor{sthlmRed}{i},t}\) that are consistent with
player \(\textcolor{sthlmRed}{i}\)'s decision-rule history up to \((\textcolor{sthlmRed}{i},t)\) (with the fixed information structure), while ranging over all
opponent private evolutions.
Because \(\mathbf{o}_{\textcolor{sthlmRed}{i},t}\) and \(\mathbf{o}_{\mathrm{next}(\textcolor{sthlmRed}{i},t)}\) live at consecutive sub-stages, the successor sequential
occupancy is recovered by a \emph{propagate-and-match} operation: propagate the predecessor family through the
occupancy recursion under the active decision rule, then match against the successor family.

\paragraph{Local dynamics.}
At sub-stage \((\textcolor{sthlmRed}{1},t)\), only player~\textcolor{sthlmRed}{1} is active: it selects \(d_{\textcolor{sthlmRed}{1},t}\), the central occupancy updates
\(x_{\textcolor{sthlmRed}{1},t}\mapsto x_{\textcolor{sthlmRed}{2},t}\), and the stage reward is \(0\).
At sub-stage \((\textcolor{sthlmRed}{2},t)\), only player~\textcolor{sthlmRed}{2} is active: it selects \(d_{\textcolor{sthlmRed}{2},t}\), the central occupancy updates
\(x_{\textcolor{sthlmRed}{2},t}\mapsto x_{\textcolor{sthlmRed}{1},t+1}\), and the original stage reward is realised.
Each player-specific occupancy family evolves by appending the corresponding decision rule, exactly as in
\citet{escudie2025varepsilonoptimally}, but interleaved across sub-stages (Figure~\ref{fig:sequential_ti_diagram}).

\subsection{Formal definitions}

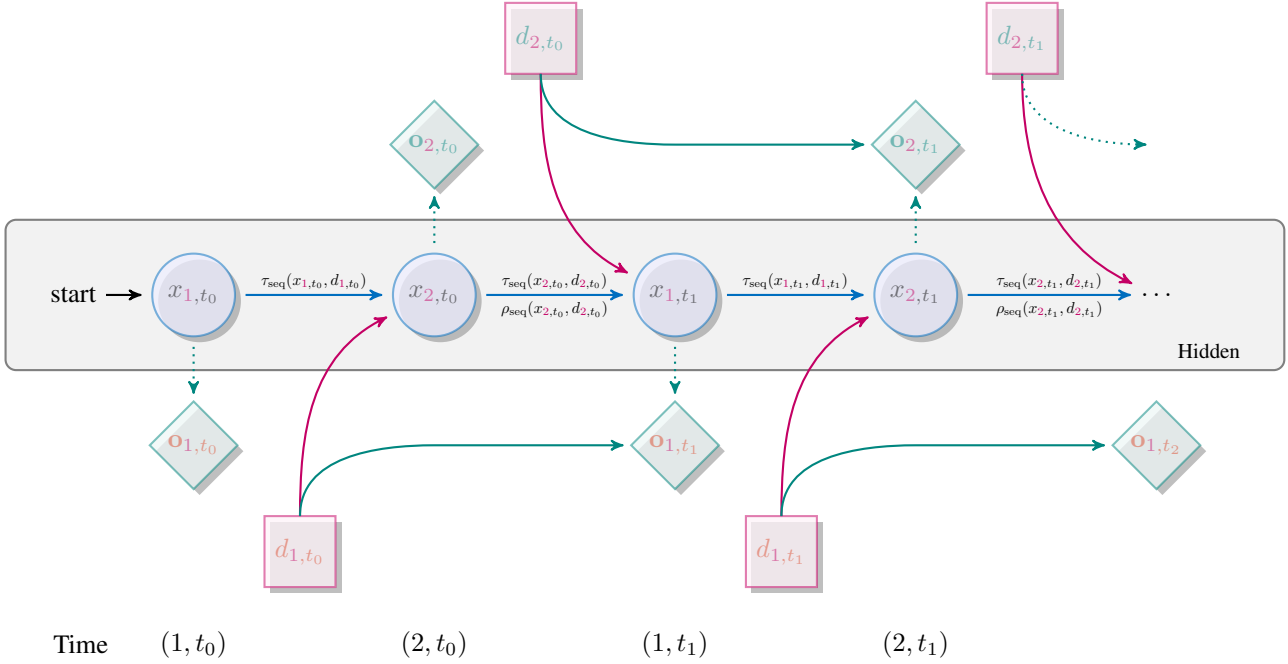
\begin{figure*}
\centering
\begin{tikzpicture}[->,>=stealth',auto,thick,node distance=3.2cm, square/.style={regular polygon,regular polygon sides=4}]
  % --- Styles ---
  \tikzstyle{every state}=[draw=black,text=black,inner color= white,outer color= white,draw= black,text=black, drop shadow,inner xsep=0pt,  inner ysep=0pt, outer sep=0pt,  minimum size=0pt]
  \tikzstyle{place}=[thick,draw=sthlmBlue,fill=blue!20,minimum size=11mm, opacity=.5, inner xsep=4pt,  inner ysep=4pt, outer sep=4pt,  minimum size=0pt]
  \tikzstyle{red place}=[square,place,draw=sthlmRed,fill=sthlmLightRed,minimum size=16mm,inner xsep=0pt,  inner ysep=0pt, outer sep=0pt,  minimum size=0pt]
  \tikzstyle{green place}=[diamond, place,draw=sthlmGreen,fill=sthlmGreen!20, inner xsep=2pt,  inner ysep=2pt, outer sep=2pt,  minimum size=0pt]
  \tikzstyle{hidden}=[draw=none, fill=none, opacity=0, text opacity=0, general shadow/.style={opacity=0}]

  % --- Bloc Central ---
  \draw[rounded corners, gray, fill=gray!10] (-2.5,-1) rectangle (14.5,1);
  \node[fill=gray!10,text=black,draw=none,scale=.8]  at (13.5,-.75) {Hidden};
  
  % --- États Globaux Séquentiels ---
  \node[initial, state, place] (S0) at (0,0) {$x_{\textcolor{sthlmRed}{1},t_0}$};
  \node[state, place] (S1) [right of=S0] {$x_{\textcolor{sthlmRed}{2},t_0}$};
  \node[state, place] (S2) [right of=S1] {$x_{\textcolor{sthlmRed}{1},t_1}$};
  \node[state, place] (S3) [right of=S2] {$x_{\textcolor{sthlmRed}{2},t_1}$};
  \node[] (S4) [right of=S3] {$\ldots$};
  
  % --- JOUEUR 1 (Bas, y=-2cm) ---
  \node[state,green place, text=sthlmOrange] (O0) [below of=S0,node distance=2cm] { $\mathbf{o}_{\textcolor{sthlmRed}{1},t_0}$ };
  \node[green place, text=sthlmOrange, hidden] (O_mid_1) [below of=S1,node distance=2cm] { $\mathbf{o}_{\textcolor{sthlmRed}{1},t_0}$ };
  \node[state,green place, text=sthlmOrange] (O1) [below of=S2,node distance=2cm] { $\mathbf{o}_{\textcolor{sthlmRed}{1},t_1}$ };
  \node[green place, text=sthlmOrange, hidden] (O_mid_2) [below of=S3,node distance=2cm] { $\mathbf{o}_{\textcolor{sthlmRed}{1},t_1}$ };
  \node[state,green place, text=sthlmOrange] (O2) [below of=S4,node distance=2cm] { $\mathbf{o}_{\textcolor{sthlmRed}{1},t_2}$ };

  % Actions P1
  \node[state,red place,text=sthlmOrange] (A0) [below right of=O0,node distance=2cm] {$d_{\textcolor{sthlmRed}{1},t_0}$};
  \node[state,red place,text=sthlmOrange] (A1) [right of=A0, node distance=6.4cm] {$d_{\textcolor{sthlmRed}{1},t_1}$};

  % --- JOUEUR 2 (Haut, y=+2cm) ---
  \node[green place, text=sthlmGreen, hidden] (O3) [above of=S0,node distance=2cm] { $\mathbf{o}_{\textcolor{sthlmRed}{2},t_0}$ };
  \node[red place,text=sthlmGreen, hidden] (A3_phantom) [above right of=O3,node distance=2cm] {$d_{\textcolor{sthlmRed}{2},t_0}$};
  \node[state,green place, text=sthlmGreen] (O4) [above of=S1,node distance=2cm] { $\mathbf{o}_{\textcolor{sthlmRed}{2},t_0}$ };
  \node[state,red place,text=sthlmGreen] (A4) [above right of=O4,node distance=2cm] {$d_{\textcolor{sthlmRed}{2},t_0}$};
  \node[green place, text=sthlmGreen, hidden] (O5) [above of=S2,node distance=2cm] { $\mathbf{o}_{\textcolor{sthlmRed}{2},t_1}$ };
  \node[state,green place, text=sthlmGreen] (O6) [above of=S3,node distance=2cm] { $\mathbf{o}_{\textcolor{sthlmRed}{2},t_1}$ };
  \node[state,red place,text=sthlmGreen] (A6) [above right of=O6,node distance=2cm] {$d_{\textcolor{sthlmRed}{2},t_1}$};
  \node[hidden] (O7) [above of=S4,node distance=2cm] {};

  % --- Axe Temporel ---
  \node[] (Time) at (-1.5,-4.65) {Time};
  \node[] (T0) [below of=S0,node distance=4.65cm] {$(1,t_0)$};
  \node[] (T1) [below of=S1,node distance=4.65cm] {$(2,t_0)$};
  \node[] (T2) [below of=S2,node distance=4.65cm] {$(1,t_1)$};
  \node[] (T3) [below of=S3,node distance=4.65cm] {$(2,t_1)$};
  
  % --- ARCS CENTRAUX ---
  \path (S0) edge[sthlmBlue] 
  node[pos=.5, fill=none,text=black,draw=none,scale=.6, above] {$\tau_{\mathrm{seq}}(x_{\textcolor{sthlmRed}{1},t_0}, d_{\textcolor{sthlmRed}{1},t_0})$} 
  (S1);
  \path (S1) edge[sthlmBlue] 
  node[pos=.5, fill=none,text=black,draw=none,scale=.6, above] {$\tau_{\mathrm{seq}}(x_{\textcolor{sthlmRed}{2},t_0}, d_{\textcolor{sthlmRed}{2},t_0})$} 
  node[pos=.5, fill=none,text=black,draw=none,scale=.6, below] {$\rho_{\mathrm{seq}}(x_{\textcolor{sthlmRed}{2},t_0}, d_{\textcolor{sthlmRed}{2},t_0})$} 
  (S2);
  \path (S2) edge[sthlmBlue] 
  node[pos=.5, fill=none,text=black,draw=none,scale=.6, above] {$\tau_{\mathrm{seq}}(x_{\textcolor{sthlmRed}{1},t_1}, d_{\textcolor{sthlmRed}{1},t_1})$} 
  (S3);
  \path (S3) edge[sthlmBlue] 
  node[pos=.5, fill=none,text=black,draw=none,scale=.6, above] {$\tau_{\mathrm{seq}}(x_{\textcolor{sthlmRed}{2},t_1}, d_{\textcolor{sthlmRed}{2},t_1})$} 
  node[pos=.5, fill=none,text=black,draw=none,scale=.6, below] {$\rho_{\mathrm{seq}}(x_{\textcolor{sthlmRed}{2},t_1}, d_{\textcolor{sthlmRed}{2},t_1})$} 
  (S4);

  % --- ARCS JOUEUR 1 ---
  \path (A0) edge [sthlmRed, out=90, in=-155] node {} (S1);
  \draw[sthlmGreen] (A0) to[out=90, in=180] (O_mid_1.center) -- (O1);

  \path (A1) edge [sthlmRed, out=90, in=-155] node {} (S3);
  \draw[sthlmGreen] (A1) to[out=90, in=180] (O_mid_2.center) -- (O2);

  % --- ARCS JOUEUR 2 ---
  \path (A4) edge [sthlmRed, out=-90, in=-205] node {} (S2);
  \draw[sthlmGreen] (A4) to[out=-90, in=180] (O5.center) -- (O6);

  \path (A6) edge [sthlmRed, out=-90, in=-205] node {} (S4);
  \draw[sthlmGreen, dotted] (A6) edge [ out=-90, in=-180] node {} (O7);

  % --- LIENS DE COMPATIBILITÉ ---
  \path[sthlmGreen, dotted] (S0) edge node {} (O0);
  \path[sthlmGreen, dotted, opacity=0] (S0) edge node {} (O3);
  
  \path[sthlmGreen, dotted, opacity=0] (S1) edge node {} (O_mid_1);
  \path[sthlmGreen, dotted] (S1) edge node {} (O4);
  
  \path[sthlmGreen, dotted] (S2) edge node {} (O1);
  \path[sthlmGreen, dotted, opacity=0] (S2) edge node {} (O5);
  
  \path[sthlmGreen, dotted, opacity=0] (S3) edge node {} (O_mid_2);
  \path[sthlmGreen, dotted] (S3) edge node {} (O6);

\end{tikzpicture}
\caption{\textbf{Sequential Transition-Independent Stochastic Game.} The diagram illustrates the interleaved dynamics extending over stages \(t_0\) and \(t_1\). 
At each sub-stage \((1,t)\), Player~\textcolor{sthlmRed}{1} is active, updating their local state \(\mathbf{o}_{\textcolor{sthlmRed}{1}}\) and driving the transition \(x_{\textcolor{sthlmRed}{1},t} \to x_{\textcolor{sthlmRed}{2},t}\). 
At sub-stage \((2,t)\), Player~\textcolor{sthlmRed}{2} acts, updating \(\mathbf{o}_{\textcolor{sthlmRed}{2}}\) and completing the stage transition \(x_{\textcolor{sthlmRed}{2},t} \to x_{\textcolor{sthlmRed}{1},t+1}\), while generating the sequential reward \(\rho_{\mathrm{seq}}\).}
\label{fig:sequential_ti_diagram}
\end{figure*}

\begin{definition}[Sequential focal planning process]
\label{def:seq:focal:planner}
For each player \(\textcolor{sthlmRed}{i}\in\{\textcolor{sthlmRed}{1},\textcolor{sthlmRed}{2}\}\), the sequential focal planning process is a tuple
\(
\mathcal{M}_{\mathrm{seq},\textcolor{sthlmRed}{i}}
=
(\mathbf{O}_{\textcolor{sthlmRed}{i}}, \mathcal{F}_{\textcolor{sthlmRed}{i}}, \mathcal{D}_{\textcolor{sthlmRed}{i}}, \tau_{\mathrm{seq},\textcolor{sthlmRed}{i}}, \rho_{\mathrm{seq},\textcolor{sthlmRed}{i}})
\),
where \(\mathbf{O}_{\textcolor{sthlmRed}{i}}\) is the set of occupancy families \(\mathbf{o}_{\textcolor{sthlmRed}{i},t}\), \(\mathcal{F}_{\textcolor{sthlmRed}{i}}\subset\mathbf{O}_{\textcolor{sthlmRed}{i}}\) is the set of terminal
families (at \((\textcolor{sthlmRed}{1},\ell)\) for \(\textcolor{sthlmRed}{i}=\textcolor{sthlmRed}{1}\), and after \((\textcolor{sthlmRed}{2},\ell-1)\) for \(\textcolor{sthlmRed}{i}=\textcolor{sthlmRed}{2}\)), \(\mathcal{D}_{\textcolor{sthlmRed}{i}}\) is the set of decision rules,
\(\tau_{\mathrm{seq},\textcolor{sthlmRed}{i}}\) maps \((\mathbf{o}_{\textcolor{sthlmRed}{i},t},d_{\textcolor{sthlmRed}{i},t})\) to the next family at the next activation of player \(\textcolor{sthlmRed}{i}\), and
\(\rho_{\mathrm{seq},\textcolor{sthlmRed}{i}}\) is \(0\) on nonterminal families and equals \(\operatorname{\mathtt{opt}}_{\neg \textcolor{sthlmRed}{i},\,x\in\mathbf{o}} g(x)\) on terminal families.
\end{definition}

\begin{definition}[Sequential TI-zs-SG]
\label{def:seq:ti:zs:sg}
The sequential transition-independent zero-sum stochastic game is
\(
\mathcal{M}'_{\mathrm{seq}}
=
(\mathcal{X}_{\mathrm{seq}}, \mathcal{F}_{\mathrm{seq}},
\mathcal{M}_{\mathrm{seq},\textcolor{sthlmRed}{1}}, \mathcal{M}_{\mathrm{seq},\textcolor{sthlmRed}{2}},
\tau_{\mathrm{seq}}, \varphi_{\mathrm{seq}}, \rho_{\mathrm{seq}}, \ell, \gamma)
\),
where \(\mathcal{X}_{\mathrm{seq}}=\{x_{\textcolor{sthlmRed}{1},t}\}_{t=0}^{\ell}\cup \{x_{\textcolor{sthlmRed}{2},t}\}_{t=0}^{\ell-1}\), \(\mathcal{F}_{\mathrm{seq}}\subseteq\mathcal{X}_{\mathrm{seq}}\) is terminal
(e.g., \(x_{\textcolor{sthlmRed}{1},\ell}\)), \(\tau_{\mathrm{seq}}\) maps \((x_{\textcolor{sthlmRed}{i},t},d_{\textcolor{sthlmRed}{i},t})\) to \(x_{\mathrm{next}(\textcolor{sthlmRed}{i},t)}\), \(\rho_{\mathrm{seq}}\) is \(0\) at \((\textcolor{sthlmRed}{1},t)\) and realises the
original stage reward at \((\textcolor{sthlmRed}{2},t)\), and the compatibility operator is
\[
\varphi_{\mathrm{seq}}(\mathbf{o}_{\textcolor{sthlmRed}{i},t}, \mathbf{o}_{\mathrm{next}(\textcolor{sthlmRed}{i},t)}, d_{\textcolor{sthlmRed}{i},t})
\doteq
\bigl\{\tau_{\mathrm{seq}}(x_{\textcolor{sthlmRed}{i},t},d_{\textcolor{sthlmRed}{i},t}) \mid x_{\textcolor{sthlmRed}{i},t}\in \mathbf{o}_{\textcolor{sthlmRed}{i},t}\bigr\}\cap \mathbf{o}_{\mathrm{next}(\textcolor{sthlmRed}{i},t)} ,
\]
which is a singleton under the standard construction of occupancy families from decision-rule histories.\footnote{%
Occupancy families are defined as the sets of sequential occupancies consistent with a fixed information structure and a given decision-rule history.
Given \(\mathbf{o}_{\textcolor{sthlmRed}{i},t}\) and \(d_{\textcolor{sthlmRed}{i},t}\), propagating any \(x_{\textcolor{sthlmRed}{i},t}\in\mathbf{o}_{\textcolor{sthlmRed}{i},t}\) through \(\tau_{\mathrm{seq}}(\cdot,d_{\textcolor{sthlmRed}{i},t})\) yields the unique successor
occupancy consistent with the successor decision-rule history; matching with \(\mathbf{o}_{\mathrm{next}(\textcolor{sthlmRed}{i},t)}\) therefore selects a unique element.}
\end{definition}

\section{Algorithms}
\label{sec:algorithms-seq}

\subsection{Point Based Value Iteration}

We now present a detailed description of the Point-Based Value Iteration (PBVI) algorithms for simultaneous zs-POSGs. We describe how PBVI \citep{pineau2003point} is adapted to solve either the original model $\mathcal{M}$ (Simultaneous PBVI, \cf Algorithm~\ref{pbvi:sim}) or its sequential reduction $\mathcal{M}_{\mathrm{seq}}$ (Sequential PBVI, \cf Algorithm~\ref{pbvi:seq}). In both cases, the algorithm alternates between two main phases: an \emph{expansion} phase, which incrementally builds a representative set of reachable occupancy states, and an \emph{improvement} phase, which updates the corresponding value function approximations over this sample.

\paragraph{Expansion phase -- \(\mathtt{expand}(\mathcal{X}_{\nu})\).}
Value updates are performed over a finite sample of occupancy states, rather than over the entire continuous occupancy space. This sample set is iteratively expanded by simulating likely trajectories from the initial belief state, typically by sampling decision rules according to a prescribed exploration strategy. By concentrating computation on reachable regions of the occupancy space, this procedure improves the accuracy of the value approximation. In the Sequential PBVI algorithm, the generated trajectories are twice as long as those of the simultaneous variant, since the effective planning horizon is doubled by the sequentialisation over the two players.

\paragraph{Improvement phase -- \(\mathtt{improve}(\mathcal{X}_{\nu}, \Gamma_{\textcolor{sthlmRed}{1},\mathtt{next}(\nu)})\).}
Value backups are carried out backward over the sampled occupancy states, from the final stage to the initial one, in a manner analogous to backward induction. This ordering enables an efficient propagation of value information across sub-stages and between players. For each occupancy state, a Bellman backup is obtained by solving a linear program that evaluates all decision rules available to the current player and selects the one that maximises the expected value. While the sequential formulation involves roughly twice as many sampled occupancy states, each linear program is substantially smaller, with significantly fewer variables and constraints.

\begin{algorithm}[H]
  \centering
  %------------------ \texttt{PBVI} procedure ------------------%
  \begin{minipage}[t]{.48\linewidth}
    \begin{algorithmic}
      \STATE ${\mathtt{function}~ \mathtt{SeqPBVI}(\ell)}$
      \STATE Initialise $\mathcal{X}_{\nu} \gets \emptyset$,  for all sub-stage $\nu$.
      \STATE Initialise $\Gamma_{\textcolor{sthlmRed}{1},\nu} \gets \emptyset$, for all sub-stage $\nu$.
      \WHILE{has not converged}
        \FOR{$\nu=(\textcolor{sthlmRed}{1},\ell),\ldots,(\textcolor{sthlmRed}{1},0)$}
          \STATE $\mathtt{improve}(\mathcal{X}_{\nu}, \Gamma_{\textcolor{sthlmRed}{1},\mathtt{next}(\nu)})$.
        \ENDFOR
        \FOR{$\nu=(\textcolor{sthlmRed}{1},0),\ldots,(\textcolor{sthlmRed}{2},\ell)$}
          \STATE $\mathcal{X}_{\mathtt{next}(\nu)} \gets \mathtt{expand}( \mathcal{X}_{\nu} )$.
        \ENDFOR
      \ENDWHILE
    \end{algorithmic}
  \end{minipage}%
  \hfill
  %---------------- improve procedure -------------------%
  \begin{minipage}[t]{.48\linewidth}
    \begin{algorithmic}
      \STATE ${\mathtt{function}~ \mathtt{improve}(\mathcal{X}_{\nu}, \Gamma_{\textcolor{sthlmRed}{1},\mathtt{next}(\nu)})}$
      \FOR{$x_{\nu} \in \mathcal{X}_{\nu}$}
        \STATE $\mathtt{sol}_{\nu} 
               \gets \texttt{LP}(\Gamma_{\textcolor{sthlmRed}{1},\mathtt{next}(\nu)}, x_{\nu})$  % 
        \STATE $\Gamma_{\textcolor{sthlmRed}{2},\nu} \gets \mathtt{Update}(\mathcal{X}_{\nu}, \Gamma_{\textcolor{sthlmRed}{1},\mathtt{next}(\nu)}, \mathtt{sol}_{\nu})$
        \STATE $\Gamma_{\textcolor{sthlmRed}{1},\nu} 
               \gets \Gamma_{\textcolor{sthlmRed}{1},\nu}
                  \cup \{\Gamma_{\textcolor{sthlmRed}{2},\nu}\}$  
                  % Update collection.
      \ENDFOR
    \end{algorithmic}
    \begin{algorithmic}
      \STATE ${\mathtt{function}~ \mathtt{expand}(\mathcal{X}_{\nu})}$
      \FOR{$x_{\nu} \in \mathcal{X}_{\nu}$}
        \STATE $\texttt{Sample randomly a set } \mathcal{D}_\nu \gets \{d_{\nu}\}$
        \STATE $\mathcal{X}_{\mathrm{next}(\nu)} 
               \gets \mathcal{X}_{\mathrm{next}(\nu)} \cup \{\tau_{\mathrm{seq}}(x_{\nu}, d_{\nu}) | d_\nu\in \mathcal{D}_\nu\}$ 
                  % Update collection.
      \ENDFOR
    \end{algorithmic}
  \end{minipage}
  \caption{Sequential \texttt{PBVI} for $\mathcal{M}_{\text{seq}}$.}
  \label{pbvi:seq}
\end{algorithm}

\begin{algorithm}[H]
  \centering
  %------------------ \texttt{PBVI} procedure ------------------%
  \begin{minipage}[t]{.48\linewidth}
    \begin{algorithmic}
      \STATE ${\mathtt{function}~ \mathtt{SimPBVI}(\ell)}$
      \STATE Initialise $\mathcal{X}_{t} \gets \emptyset$,  for all time-step $t$.
      \STATE Initialise $\Gamma_{\textcolor{sthlmRed}{1},t} \gets \emptyset$, for all time-steps $t$.
      \WHILE{has not converged}
        \FOR{$t=\ell-1,\ldots,0$}
          \STATE $\mathtt{improve}(\mathcal{X}_{t}, \Gamma_{\textcolor{sthlmRed}{1},t+1})$.
        \ENDFOR
        \FOR{$t=0,\dots,\ell$}
          \STATE $\mathcal{X}_{t+1} \gets \mathtt{expand}( \mathcal{X}_{t} )$.
        \ENDFOR
      \ENDWHILE
    \end{algorithmic}
  \end{minipage}%
  \hfill
  %---------------- improve procedure -------------------%
  \begin{minipage}[t]{.52\linewidth}
    \begin{algorithmic}
      \STATE ${\mathtt{function}~ \mathtt{improve}(\mathcal{X}_{t}, \Gamma_{\textcolor{sthlmRed}{1},t+1})}$
      \FOR{$x_{t} \in \mathcal{X}_{t}$}
        \STATE $\mathtt{sol}_{t} 
               \gets \texttt{LP}(\Gamma_{\textcolor{sthlmRed}{1},t+1}, x_{t})$  % 
        \STATE $\Gamma_{\textcolor{sthlmRed}{2},t} \gets \mathtt{Update}(\mathcal{X}_{t}, \Gamma_{\textcolor{sthlmRed}{1},t+1}, \mathtt{sol}_{t})$
        \STATE $\Gamma_{\textcolor{sthlmRed}{1},t} 
               \gets \Gamma_{\textcolor{sthlmRed}{1},t}
                  \cup \{\Gamma_{\textcolor{sthlmRed}{2},t}\}$  
                  % Update collection.
      \ENDFOR
    \end{algorithmic}
    \begin{algorithmic}
      \STATE ${\mathtt{function}~ \mathtt{expand}(\mathcal{X}_{t})}$
      \FOR{$x_{t} \in \mathcal{X}_{t}$}
        \STATE $\texttt{Sample randomly a set } \mathcal{D} \gets \{(d_{\textcolor{sthlmRed}{1},t}, d_{\textcolor{sthlmRed}{2},t})\}$
        \STATE $\mathcal{X}_{t+1} 
               \gets \mathcal{X}_{t+1} \cup \{\tau_{\mathrm{sim}}(x_{t}, (d_{\textcolor{sthlmRed}{1},t}, d_{\textcolor{sthlmRed}{2},t})) | (d_{\textcolor{sthlmRed}{1},t}, d_{\textcolor{sthlmRed}{2},t})\in \mathcal{D}\}$ 
                  % Update collection.
      \ENDFOR
    \end{algorithmic}
  \end{minipage}
  \caption{Simultaneous \texttt{PBVI} for $\mathcal{M}$.}
  \label{pbvi:sim}
\end{algorithm}

\subsection{Pruning strategy}

This section draws inspiration from \citep{escudie2025varepsilonoptimally}.  For any sub-stage~$\nu$, this procedure provides the subroutines required to prune unnecessary sets of hyperplanes $\Gamma_{\!\textcolor{sthlmRed}{2},\nu}$ from the collection $\Gamma_{\!\textcolor{sthlmRed}{1},\nu}$, as well as redundant points in $\mathcal{X}_\nu$. 

An envelope $\Gamma_{\!\textcolor{sthlmRed}{2},\nu} \in \Gamma_{\!\textcolor{sthlmRed}{1},\nu}$ is said to be dominated over $\mathcal{X}_\nu$ if, for every sequential occupancy state $x_\nu \in \mathcal{X}_\nu$, there exists another envelope $\Gamma'_{\!\textcolor{sthlmRed}{2},\nu} \in \Gamma_{\!\textcolor{sthlmRed}{1},\nu}$ such that
$\mathtt{Val}_{\Gamma_{\!\textcolor{sthlmRed}{2},\nu}}(x_\nu)
<
\mathtt{Val}_{\Gamma'_{\!\textcolor{sthlmRed}{2},\nu}}(x_\nu)
$.
According to Algorithm~\ref{pruning:zsposg}, a dominated envelope $\Gamma_{\!\textcolor{sthlmRed}{2},\nu}$ has no support point in $\mathcal{X}_\nu$, and therefore satisfies $\mathtt{refCount}(\Gamma_{\!\textcolor{sthlmRed}{2},\nu}) = 0$. The pruning routine consequently retains only envelopes with at least one support point in $\mathcal{X}_\nu$. That said, this procedure is not fully aggressive: a dominated envelope may still admit one or more support points in $\mathcal{X}_\nu$, although this situation appears unlikely in practice.

\begin{algorithm}[H]
        \caption{Bounded Pruning.}
        \label{pruning:zsposg}
        \begin{algorithmic}
            \STATE ${\mathtt{function}~ \mathtt{BoundedPruning}(\Gamma_{\!\!\textcolor{sthlmRed}{1},\nu},\mathcal{X}_\nu)}$
            \FOR{$\Gamma_{\!\!\textcolor{sthlmRed}{2},\nu}\in \Gamma_{\!\!\textcolor{sthlmRed}{1},\nu}$}
                \STATE $\mathtt{refCount}(\Gamma_{\!\!\textcolor{sthlmRed}{2},\nu}) \gets 0$.
            \ENDFOR
            \FOR{$x_\nu\in \mathcal{X}_\nu$}
                \STATE $\Gamma_{\!\!\textcolor{sthlmRed}{2},x_\nu} \gets \argmax_{\Gamma_{\!\!\textcolor{sthlmRed}{2}}\in \Gamma_{\!\!\textcolor{sthlmRed}{1}},\nu} \mathtt{Val}_{\Gamma_{\!\textcolor{sthlmRed}{2}}}(x_\nu)$
                \STATE  $\mathtt{refCount}(\Gamma_{\!\!\textcolor{sthlmRed}{2},x_\nu}) \gets \mathtt{refCount}(\Gamma_{\!\!\textcolor{sthlmRed}{2},x_\nu})+1$
            \ENDFOR
            \STATE \textbf{return} $\{\Gamma_{\!\!\textcolor{sthlmRed}{2},\nu}\in \Gamma_{\!\!\textcolor{sthlmRed}{1},\nu} \mid \mathtt{refCount}(\Gamma_{\!\!\textcolor{sthlmRed}{2},\nu})>0\}$
        \end{algorithmic}
\end{algorithm}

Similarly, we remove redundant points in \(\mathcal{X}_\nu\). Fix \(\epsilon>0\) and sequential occupancy state \(x_\nu\in \mathcal X_\nu\). Let \(\Gamma_{\!\!\textcolor{sthlmRed}{2},x_\nu}\in \Gamma_{\!\!\textcolor{sthlmRed}{1},\nu}\) denote an envelop within \(\Gamma_{\!\!\textcolor{sthlmRed}{1},\nu}\) that achieves the highest pessimistic value for player~\(\textcolor{sthlmRed}{2}\) at \(x_\nu\), \ie
\[ 
 \mathtt{Val}_{\Gamma_{\!\textcolor{sthlmRed}{2},x_\nu}}(x_\nu) 
 =
 \textstyle
 \max_{\Gamma_{\!\!\textcolor{sthlmRed}{2}}\in \Gamma_{\!\!\textcolor{sthlmRed}{1},\nu}}
 \mathtt{Val}_{\Gamma_{\!\textcolor{sthlmRed}{2}}}(x_\nu) 
 .
\]
We say that a point $x_\nu \in \mathcal{X}_\nu$ is $\epsilon$-redundant with respect to $(\mathcal{X}_\nu, \Gamma_{\!\textcolor{sthlmRed}{1},\nu})$ if and only if there exists another point $x'_\nu \in \mathcal{X}_\nu$ such that the gap
\(
|\mathtt{Val}_{\Gamma_{\!\textcolor{sthlmRed}{2},x_\nu}}(x_\nu) - \mathtt{Val}_{\Gamma_{\!\textcolor{sthlmRed}{2},x'_\nu}}(x_\nu)|
\)
is less than or equal to $\epsilon$. Algorithm~\ref{pruning:occupancystates} removes from $\mathcal{X}_\nu$ all $\epsilon$-redundant points.

\begin{algorithm}[H]
\caption{Redundant Sequential Occupancy State Pruning.}
\label{pruning:occupancystates}
\begin{algorithmic}
    \STATE ${\mathtt{function}~ \mathtt{PruneStates}(\mathcal{X}_\nu, \Gamma_{\!\!\textcolor{sthlmRed}{1},\nu}, \epsilon)}$
    \STATE Initialise $\mathcal{X}_\nu^{\circ} \gets \emptyset$
    \FOR{$x_\nu \in \mathcal{X}_\nu$}
        \STATE $\Gamma_{\!\!\textcolor{sthlmRed}{2},x_\nu} \gets \argmax_{\Gamma_{\!\!\textcolor{sthlmRed}{2}}\in \Gamma_{\!\!\textcolor{sthlmRed}{1},\nu}} \mathtt{Val}_{\Gamma_{\!\textcolor{sthlmRed}{2}}}(x_\nu)$
    \ENDFOR
    \FOR{$x_\nu \in \mathcal{X}_\nu$}
        \STATE $\texttt{isRedundant} \gets \texttt{false}$
        \FOR{$x'_\nu \in \mathcal{X}_\nu^{\circ}$}
            \IF{$|\mathtt{Val}_{\Gamma_{\!\textcolor{sthlmRed}{2},x_\nu}}(x_\nu) 
            -
            \mathtt{Val}_{\Gamma_{\!\textcolor{sthlmRed}{2},x'_\nu}}(x_\nu) | \leq \epsilon$}
                \STATE $\texttt{isRedundant} \gets \texttt{true}$ and \textbf{break}
            \ENDIF
        \ENDFOR
        \IF{$\neg \texttt{isRedundant}$}
            \STATE $\mathcal{X}_\nu^{\circ} \gets \mathcal{X}_\nu^{\circ} \cup \{x_\nu\}$
        \ENDIF
    \ENDFOR
    \STATE \textbf{return} $\mathcal{X}_\nu^{\circ}$
\end{algorithmic}
\end{algorithm}

\section{Experimental results}

\subsection{Benchmarks}

We evaluate our approach on several competitive benchmark problems, adapted from standard multi-agent settings.

% \paragraph{Matching Pennies.} Each player secretly chooses heads or tails. If the two choices match, Player 1 wins; otherwise, Player 2 wins. 

\paragraph{Multi-Agent Recycling.} In the original cooperative version, two robots must clean a room represented as a grid by emptying garbage cans. Each robot has limited battery life and a restricted view of the environment, including limited observability of the other robot. Coordination is therefore required. We adapt the task to a zero-sum setting by altering the reward function: each robot now aims to clean more efficiently than the other.

\paragraph{Multi-Agent Tiger.} The environment consists of two rooms—one containing a treasure and the other a tiger. Each agent stands before a door and may choose to either listen for cues or enter a room. Due to stochastic listening outcomes, agents receive noisy observations. Two competitive variants, \emph{Adversarial Tiger} and \emph{Competitive Tiger}, were introduced in~\citet{Wiggers16} to study adversarial behaviour under partial observability.

\paragraph{Multi-Agent Broadcast Channel (MABC).} This benchmark captures a communication scenario where two agents (nodes) must broadcast messages over a shared channel. To prevent collisions, only one node may broadcast at any time. While the original version is cooperative, maximising joint throughput, we consider a competitive variant by modifying the reward structure.

\paragraph{Kuhn Poker.} Kuhn Poker is a simple two-player, zero-sum poker game with three cards. Each player is dealt one hidden card and antes one chip. Players alternate betting or passing until one folds or the bets are equal. The higher card wins the pot.

% ADD TIME + VALUES

\subsection{Additional results}

\begin{figure}[htbp]
  \centering
  \begin{subfigure}
    \centering
    \begin{tikzpicture}
      \begin{axis}[
        title={\shortstack{Adversarial \\ Tiger}},
        title style={yshift=-20pt, fill=white, font=\subfigfontsize},
        xlabel=Iterations,
        ylabel=Exploitability,
        width=\subfigsize,
        height=\subfigsize,
        xmin=1, xmax=5,
        xtick={1, 3, 5},
        ymin=0, ymax=0.1,
        ytick={0, 0.05, 0.1},
        ylabel style={yshift=-6pt, font=\subfigfontsize},
        xlabel style={yshift=4pt, font=\subfigfontsize},
        grid=both,
        scaled y ticks=false,
        yticklabel style={
          /pgf/number format/fixed,
          /pgf/number format/fixed zerofill,
          /pgf/number format/precision=2,
          font=\subfigfontsize
        },
        style={very thick, font=\subfigfontsize},
        legend style={draw=none}
      ]
        \addplot+[const plot,mark=none,very thick,sthlmRed] table [col sep=comma, x=iter, y=exp_value] {./results/adversarial_tiger/seq2_adversarialtiger_H5_T1_S42.log};
        \addplot+[const plot,mark=none,very thick,sthlmGreen] table [col sep=comma, x=iter, y=exp_value] {./results/adversarial_tiger/sim_adversarialtiger_H5_T1.log};
      \end{axis}
    \end{tikzpicture}
  \end{subfigure}
  % \hfill
  \begin{subfigure}
    \centering
    \begin{tikzpicture}
      \begin{axis}[
        title={MABC},
        title style={yshift=-2.7ex, fill=white},
        xlabel=Iterations,
        % ylabel=Exploitability,
        width=\subfigsize,
        height=\subfigsize,
        xmin=1, xmax=4,
        xtick={1, 4},
        ymin=0, ymax=0.1,
        ytick={0, 0.05, 0.1},
        ylabel style={yshift=-6pt, font=\subfigfontsize},
        xlabel style={yshift=4pt, font=\subfigfontsize},
        grid=both,
        scaled y ticks=false,
        yticklabel style={
          /pgf/number format/fixed,
          /pgf/number format/fixed zerofill,
          /pgf/number format/precision=2,
          font=\subfigfontsize
        },
        style={very thick, font=\subfigfontsize},
        legend style={draw=none}
      ]
        \addplot+[const plot,mark=none,very thick,sthlmRed] table [col sep=comma, x=iter, y=exp_value] {./results/mabc/seq2_mabc_H5_T2_S42.log};
        \addplot+[const plot,mark=none,very thick,sthlmGreen] table [col sep=comma, x=iter, y=exp_value] {./results/mabc/sim_mabc_H5_T2.log};
      \end{axis}
    \end{tikzpicture}
  \end{subfigure}
  % \hfill
  \begin{subfigure}
    \centering
    \begin{tikzpicture}
      \begin{axis}[
        title={Recycling},
        title style={yshift=-2.7ex, fill=white},
        xlabel=Iterations,
        % ylabel=Exploitability,
        width=\subfigsize,
        height=\subfigsize,
        xmin=1, xmax=5,
        xtick={1, 3, 5},
        ymin=0, ymax=0.1,
        ytick={0, 0.05, 0.1},
        ylabel style={yshift=-6pt, font=\subfigfontsize},
        xlabel style={yshift=4pt, font=\subfigfontsize},
        grid=both,
        scaled y ticks=false,
        yticklabel style={
          /pgf/number format/fixed,
          /pgf/number format/fixed zerofill,
          /pgf/number format/precision=2,
          font=\subfigfontsize
        },
        style={very thick, font=\subfigfontsize},
        legend style={draw=none}
      ]
        \addplot+[const plot,mark=none,very thick,sthlmRed] table [col sep=comma, x=iter, y=exp_value] {./results/recycling/seq2_recycling_H5_T1_S42.log};
        \addplot+[const plot,mark=none,very thick,sthlmGreen] table [col sep=comma, x=iter, y=exp_value] {./results/recycling/sim_recycling_H5_T1.log};
      \end{axis}
    \end{tikzpicture}
  \end{subfigure}
  % \hfill
  % \begin{subfigure}[b]{0.22\textwidth}
  %   \centering
  %   \begin{tikzpicture}
  %     \begin{axis}[
  %       title={Kuhn Poker},
  %       title style={yshift=-2ex, fill=white},
  %       xlabel=Iterations,
  %       % ylabel=Exploitability,
  %       width=150pt,
  %       height=150pt,
  %       xmin=1, xmax=5,
  %       xtick={1, 3, 5},
  %       ymin=0, ymax=0.1,
  %       ytick={0, 0.05, 0.1},
  %       tick label style={font=\subfigfontsize},
  %       ylabel style={yshift=-6pt, font=\subfigfontsize},
  %       xlabel style={yshift=4pt, font=\subfigfontsize},
  %       grid=both,
  %       scaled y ticks=false,
  %       yticklabel style={
  %         /pgf/number format/fixed,
  %         /pgf/number format/fixed zerofill,
  %         /pgf/number format/precision=2,
  %         font=\subfigfontsize
  %       },
  %       style={very thick, font=\subfigfontsize},
  %       legend style={draw=none}
  %     ]
  %       % \addplot+[const plot,mark=none,very thick,sthlmRed] table [col sep=comma, x=iter, y=exp_value] {./results/};
  %       % \addplot+[const plot,mark=none,very thick,sthlmGreen] table [col sep=comma, x=iter, y=exp_value] {./results/};
  %     \end{axis}
  %   \end{tikzpicture}
  % \end{subfigure}

  \par

  \begin{subfigure}
    \centering
    \begin{tikzpicture}
      \begin{axis}[
        title={\shortstack{Adversarial \\ Tiger}},
        title style={yshift=-3.3ex, fill=white},
        xlabel=Iterations,
        ylabel=Nb. of points sampled,
        width=\subfigsize,
        height=\subfigsize,
        xmin=1, xmax=5,
        xtick={1, 3, 5},
        ymin=0, ymax=150,
        ytick={0, 75, 150},
        ylabel style={yshift=-6pt, font=\subfigfontsize},
        xlabel style={yshift=4pt, font=\subfigfontsize},
        grid=both,
        scaled y ticks=false,
        yticklabel style={
          /pgf/number format/fixed,
          /pgf/number format/fixed zerofill,
          /pgf/number format/precision=0,
          font=\subfigfontsize
        },
        style={very thick, font=\subfigfontsize},
        legend style={draw=none}
      ]
        \addplot+[const plot,mark=none,very thick,sthlmRed] table [col sep=comma, x=iter, y=S_size] {./results/adversarial_tiger/seq2_adversarialtiger_H5_T1_S42.log};
        \addplot+[const plot,mark=none,very thick,sthlmGreen] table [col sep=comma, x=iter, y=S_size] {./results/adversarial_tiger/sim_adversarialtiger_H5_T1.log};
      \end{axis}
    \end{tikzpicture}
  \end{subfigure}
  % \hfill
  \begin{subfigure}
    \centering
    \begin{tikzpicture}
      \begin{axis}[
        title={MABC},
        title style={yshift=-2.7ex, fill=white},
        xlabel=Iterations,
        % ylabel=Exploitability,
        width=\subfigsize,
        height=\subfigsize,
        xmin=1, xmax=4,
        xtick={1, 4},
        ymin=0, ymax=100,
        ytick={0, 50, 100},
        ylabel style={yshift=-6pt, font=\subfigfontsize},
        xlabel style={yshift=4pt, font=\subfigfontsize},
        grid=both,
        scaled y ticks=false,
        yticklabel style={
          /pgf/number format/fixed,
          /pgf/number format/fixed zerofill,
          /pgf/number format/precision=0,
          font=\subfigfontsize
        },
        style={very thick, font=\subfigfontsize},
        legend style={draw=none}
      ]
        \addplot+[const plot,mark=none,very thick,sthlmRed] table [col sep=comma, x=iter, y=S_size] {./results/mabc/seq2_mabc_H5_T2_S42.log};
        \addplot+[const plot,mark=none,very thick,sthlmGreen] table [col sep=comma, x=iter, y=S_size] {./results/mabc/sim_mabc_H5_T2.log};
      \end{axis}
    \end{tikzpicture}
  \end{subfigure}
  % \hfill
  \begin{subfigure}
    \centering
    \begin{tikzpicture}
      \begin{axis}[
        title={Recycling},
        title style={yshift=-2.7ex, fill=white},
        xlabel=Iterations,
        % ylabel=Exploitability,
        width=\subfigsize,
        height=\subfigsize,
        xmin=1, xmax=5,
        xtick={1, 3, 5},
        ymin=0, ymax=200,
        ytick={0,100,200},
        ylabel style={yshift=-6pt, font=\subfigfontsize},
        xlabel style={yshift=4pt, font=\subfigfontsize},
        grid=both,
        scaled y ticks=false,
        yticklabel style={
          /pgf/number format/fixed,
          /pgf/number format/fixed zerofill,
          /pgf/number format/precision=0,
          font=\subfigfontsize
        },
        style={very thick, font=\subfigfontsize},
        legend style={
            at={(0.05,0.85)},
            thin,
            legend cell align=left,
            anchor=north west,
            font=\small,
            draw=none,
            legend columns=1,
            nodes={scale=0.6},
            inner sep=0.3pt,
            name=leg
        },
        legend image post style={
            xscale=0.2
        }
      ]
        \addplot+[const plot,mark=none,very thick,sthlmRed] table [col sep=comma, x=iter, y=S_size] {./results/recycling/seq2_recycling_H4_T1_S42.log};
        \addlegendentry{Seq}
        \addplot+[const plot,mark=none,very thick,sthlmGreen] table [col sep=comma, x=iter, y=S_size] {./results/recycling/sim_recycling_H4_T1.log};
        \addlegendentry{Sim}
      \end{axis}
        \begin{scope}[on background layer]
            \draw[black, thin] (leg.north west) rectangle (leg.south east);
            \draw[black, thin] ([xshift=-0.8pt,yshift=0.8pt]leg.north west) rectangle ([xshift=0.8pt,yshift=-0.8pt]leg.south east);
        \end{scope}
    \end{tikzpicture}
  \end{subfigure}
  
  \caption{\textbf{Top:} Exploitability of PBVI across iterations on different benchmarks (\(\ell=5\)). The \textcolor{sthlmGreen}{green curve} corresponds to the simultaneous variant, while the \textcolor{sthlmRed}{red curve} represents the sequential variant. \textbf{Bottom:} Number of points sampled at each iteration by \textcolor{sthlmGreen}{simultaneous} and \textcolor{sthlmRed}{sequential} variants.}
  \label{fig:value_and_exp_plots}
\end{figure}
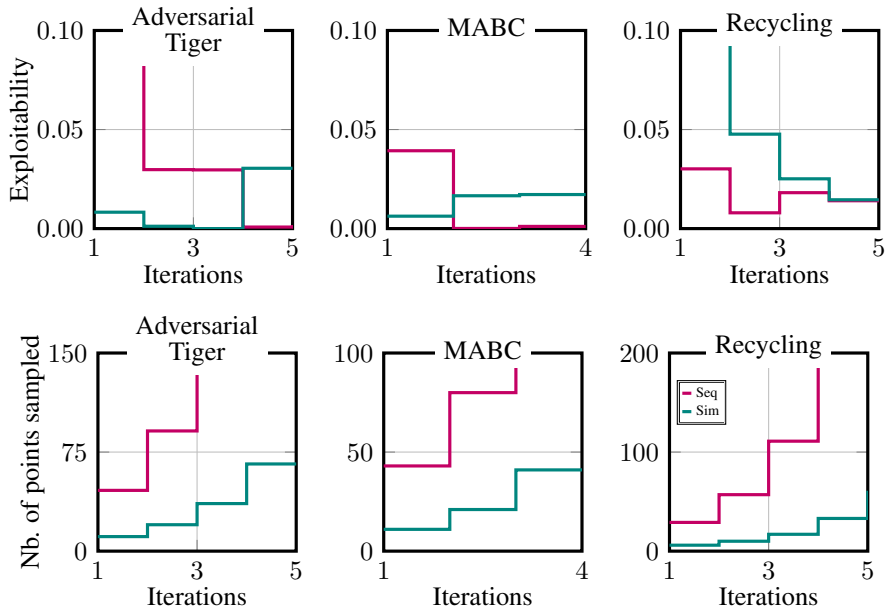
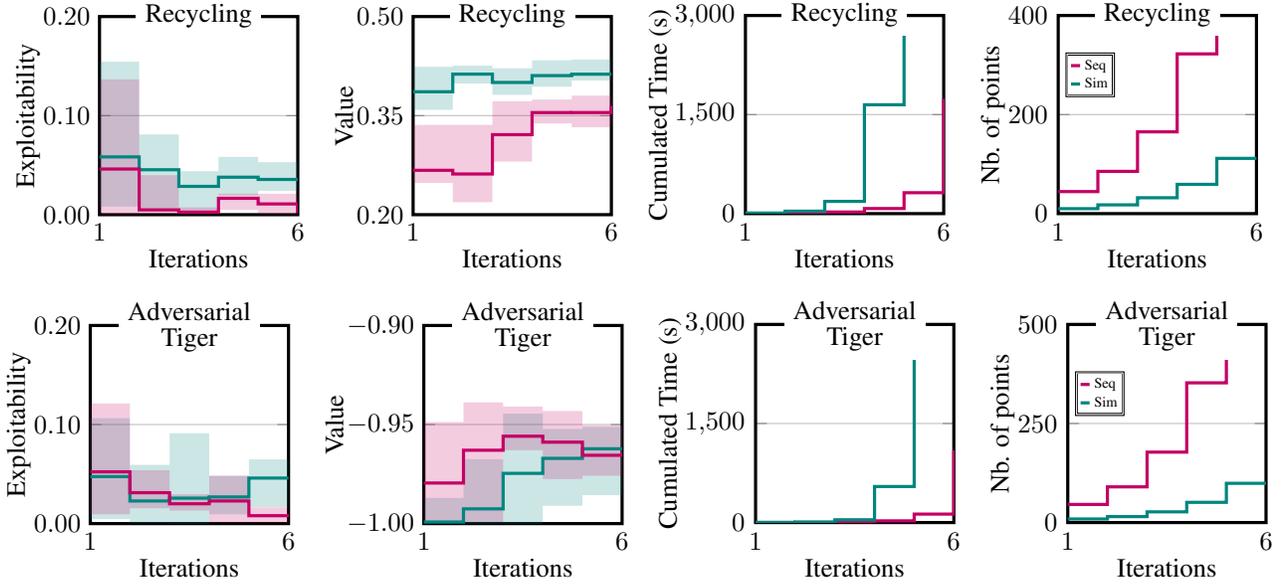
\begin{figure}[htbp]
  \centering

  \begin{subfigure}
    \centering
    \begin{tikzpicture}
      \begin{axis}[
        title={Recycling},
        title style={yshift=-14pt, fill=white},
        xlabel=Iterations,
        ylabel=Exploitability,
        width=\subfigsize,
        height=\subfigsize,
        xmin=1, xmax=6,
        xtick={1, 6},
        ymin=0, ymax=0.2,
        ytick={0, 0.1, 0.2},
        ylabel style={yshift=-6pt,font=\subfigfontsize},
        xlabel style={yshift=4pt,font=\subfigfontsize},
        grid=both,
        scaled y ticks=false,
        yticklabel style={
          /pgf/number format/fixed,
          /pgf/number format/fixed zerofill,
          /pgf/number format/precision=2,
         font=\subfigfontsize
        },
        style={very thick,font=\subfigfontsize},
        axis line style={very thick},
        legend style={draw=none}
      ]
        \addplot+[const plot, mark=none, very thick, sthlmGreen, opacity=1., solid] table [col sep=comma, x=iter, y=exp_value] {./results/recycling/random/sim_mean.log};
        \addplot+[const plot, mark=none, very thick, sthlmRed, opacity=1., solid] table [col sep=comma, x=iter, y=exp_value] {./results/recycling/random/seq_mean.log};
        
        \addplot+[name path=sim_max, const plot, mark=none, thick, sthlmGreen, opacity=0., solid] table [col sep=comma, x=iter, y=exp_value] {./results/recycling/random/sim_max.log};
        \addplot+[name path=sim_min, const plot, mark=none, thick, sthlmGreen, opacity=0., solid] table [col sep=comma, x=iter, y=exp_value] {./results/recycling/random/sim_min.log};
        \addplot[
            fill=sthlmGreen,
            fill opacity=0.2
        ]
        fill between[
            of=sim_max and sim_min
        ];

        \addplot+[name path=seq_max, const plot, mark=none, thick, sthlmRed, opacity=0., solid] table [col sep=comma, x=iter, y=exp_value] {./results/recycling/random/seq_max.log};
        \addplot+[name path=seq_min, const plot, mark=none, thick, sthlmRed, opacity=0., solid] table [col sep=comma, x=iter, y=exp_value] {./results/recycling/random/seq_min.log};
        \addplot[
            fill=sthlmRed,
            fill opacity=0.2
        ]
        fill between[
            of=seq_max and seq_min
        ];
      \end{axis}
    \end{tikzpicture}
  \end{subfigure}
  % \hfill
  \begin{subfigure}
    \centering
    \begin{tikzpicture}
      \begin{axis}[
        title={Recycling},
        title style={yshift=-14pt, fill=white},
        xlabel=Iterations,
        ylabel=Value,
        width=\subfigsize,
        height=\subfigsize,
        xmin=1, xmax=6,
        xtick={1, 6},
        ymin=0.2, ymax=0.5,
        ytick={0.2, 0.35, 0.5},
        ylabel style={yshift=-6pt,font=\subfigfontsize},
        xlabel style={yshift=4pt,font=\subfigfontsize},
        grid=both,
        scaled y ticks=false,
        yticklabel style={
          /pgf/number format/fixed,
          /pgf/number format/fixed zerofill,
          /pgf/number format/precision=2,
         font=\subfigfontsize
        },
        style={very thick,font=\subfigfontsize},
        axis line style={very thick},
        legend style={draw=none}
      ]
        \addplot+[const plot, mark=none, very thick, sthlmGreen, opacity=1., solid] table [col sep=comma, x=iter, y=value] {./results/recycling/random/sim_mean.log};
        \addplot+[const plot, mark=none, very thick, sthlmRed, opacity=1., solid] table [col sep=comma, x=iter, y=value] {./results/recycling/random/seq_mean.log};
        
        \addplot+[name path=sim_max, const plot, mark=none, thick, sthlmGreen, opacity=0., solid] table [col sep=comma, x=iter, y=value] {./results/recycling/random/sim_max.log};
        \addplot+[name path=sim_min, const plot, mark=none, thick, sthlmGreen, opacity=0., solid] table [col sep=comma, x=iter, y=value] {./results/recycling/random/sim_min.log};
        \addplot[
            fill=sthlmGreen,
            fill opacity=0.2
        ]
        fill between[
            of=sim_max and sim_min
        ];

        \addplot+[name path=seq_max, const plot, mark=none, thick, sthlmRed, opacity=0., solid] table [col sep=comma, x=iter, y=value] {./results/recycling/random/seq_max.log};
        \addplot+[name path=seq_min, const plot, mark=none, thick, sthlmRed, opacity=0., solid] table [col sep=comma, x=iter, y=value] {./results/recycling/random/seq_min.log};
        \addplot[
            fill=sthlmRed,
            fill opacity=0.2
        ]
        fill between[
            of=seq_max and seq_min
        ];
      \end{axis}
    \end{tikzpicture}
  \end{subfigure}
  % \hfill
  \begin{subfigure}
    \centering
    \begin{tikzpicture}
      \begin{axis}[
        title={Recycling},
        title style={yshift=-14pt, fill=white},
        xlabel=Iterations,
        ylabel=Cumulated Time (s),
        width=\subfigsize,
        height=\subfigsize,
        xmin=1, xmax=6,
        xtick={1, 6},
        ymin=0, ymax=3000,
        ytick={0, 1500, 3000},
        ylabel style={yshift=-6pt,font=\subfigfontsize},
        xlabel style={yshift=4pt,font=\subfigfontsize},
        grid=both,
        scaled y ticks=false,
        yticklabel style={
          /pgf/number format/fixed,
          /pgf/number format/fixed zerofill,
          /pgf/number format/precision=0,
         font=\subfigfontsize
        },
        style={very thick,font=\subfigfontsize},
        % legend style={draw=none},
        % legend style={
        %     at={(0.05,0.8)},
        %     thin,
        %     legend cell align=left,
        %     anchor=north west,
        %     font=\small,
        %     inner sep=0.3pt,
        %     draw=none,
        %     legend columns=1,
        %     nodes={scale=0.6},
        %     name=leg
        % },
        % legend image post style={
        %     xscale=0.2
        % }
      ]
        \addplot+[const plot, mark=none, very thick, sthlmRed, opacity=1.] table [col sep=comma, x=iter, y=cumul_time] {./results/recycling/random/seq_mean.log};
        % \addlegendentry{Seq}
        
        \addplot+[const plot, mark=none, very thick, solid, sthlmGreen, opacity=1.] table [col sep=comma, x=iter, y=cumul_time] {./results/recycling/random/sim_mean.log};
        % \addlegendentry{Sim}
      \end{axis}
    \end{tikzpicture}
  \end{subfigure}
  \begin{subfigure}
    \centering
    \begin{tikzpicture}
      \begin{axis}[
        title={Recycling},
        title style={yshift=-14pt, fill=white},
        xlabel=Iterations,
        ylabel=Nb. of points,
        width=\subfigsize,
        height=\subfigsize,
        xmin=1, xmax=6,
        xtick={1, 6},
        ymin=0, ymax=400,
        ytick={0, 200, 400},
        ylabel style={yshift=-6pt,font=\subfigfontsize},
        xlabel style={yshift=4pt,font=\subfigfontsize},
        grid=both,
        scaled y ticks=false,
        yticklabel style={
          /pgf/number format/fixed,
          /pgf/number format/fixed zerofill,
          /pgf/number format/precision=0,
         font=\subfigfontsize
        },
        style={very thick,font=\subfigfontsize},
        legend style={draw=none},
        legend style={
            at={(0.05,0.8)},
            thin,
            legend cell align=left,
            anchor=north west,
            font=\small,
            inner sep=0.3pt,
            draw=none,
            legend columns=1,
            nodes={scale=0.6},
            name=leg
        },
        legend image post style={
            xscale=0.2
        }
      ]
        \addplot+[const plot, mark=none, very thick, sthlmRed, opacity=1.] table [col sep=comma, x=iter, y=S_size] {./results/recycling/random/seq_mean.log};
        \addlegendentry{Seq}
        
        \addplot+[const plot, mark=none, very thick, solid, sthlmGreen, opacity=1.] table [col sep=comma, x=iter, y=S_size] {./results/recycling/random/sim_mean.log};
        \addlegendentry{Sim}
      \end{axis}
      
        \begin{scope}[on background layer]
            \draw[black, thin] (leg.north west) rectangle (leg.south east);
            \draw[black, thin] ([xshift=-0.8pt,yshift=0.8pt]leg.north west) rectangle ([xshift=0.8pt,yshift=-0.8pt]leg.south east);
        \end{scope}
    \end{tikzpicture}
  \end{subfigure}

  \par

  \begin{subfigure}
    \centering
    \begin{tikzpicture}
      \begin{axis}[
        title={\shortstack{Adversarial \\ Tiger}},
        title style={yshift=-20pt, fill=white},
        xlabel=Iterations,
        ylabel=Exploitability,
        width=\subfigsize,
        height=\subfigsize,
        xmin=1, xmax=6,
        xtick={1, 6},
        ymin=0, ymax=0.2,
        ytick={0, 0.1, 0.2},
        ylabel style={yshift=-6pt,font=\subfigfontsize},
        xlabel style={yshift=4pt,font=\subfigfontsize},
        grid=both,
        scaled y ticks=false,
        yticklabel style={
          /pgf/number format/fixed,
          /pgf/number format/fixed zerofill,
          /pgf/number format/precision=2,
         font=\subfigfontsize
        },
        style={very thick,font=\subfigfontsize},
        axis line style={very thick},
        legend style={draw=none}
      ]
        \addplot+[const plot, mark=none, very thick, sthlmGreen, opacity=1., solid] table [col sep=comma, x=iter, y=exp_value] {./results/adversarial_tiger/random/sim_mean.log};
        \addplot+[const plot, mark=none, very thick, sthlmRed, opacity=1., solid] table [col sep=comma, x=iter, y=exp_value] {./results/adversarial_tiger/random/seq_mean.log};
        
        \addplot+[name path=sim_max, const plot, mark=none, thick, sthlmGreen, opacity=0., solid] table [col sep=comma, x=iter, y=exp_value] {./results/adversarial_tiger/random/sim_max.log};
        \addplot+[name path=sim_min, const plot, mark=none, thick, sthlmGreen, opacity=0., solid] table [col sep=comma, x=iter, y=exp_value] {./results/adversarial_tiger/random/sim_min.log};
        \addplot[
            fill=sthlmGreen,
            fill opacity=0.2
        ]
        fill between[
            of=sim_max and sim_min
        ];

        \addplot+[name path=seq_max, const plot, mark=none, thick, sthlmRed, opacity=0., solid] table [col sep=comma, x=iter, y=exp_value] {./results/adversarial_tiger/random/seq_max.log};
        \addplot+[name path=seq_min, const plot, mark=none, thick, sthlmRed, opacity=0., solid] table [col sep=comma, x=iter, y=exp_value] {./results/adversarial_tiger/random/seq_min.log};
        \addplot[
            fill=sthlmRed,
            fill opacity=0.2
        ]
        fill between[
            of=seq_max and seq_min
        ];
      \end{axis}
    \end{tikzpicture}
  \end{subfigure}
  % \hfill
  \begin{subfigure}
    \centering
    \begin{tikzpicture}
      \begin{axis}[
        title={\shortstack{Adversarial \\ Tiger}},
        title style={yshift=-20pt, fill=white},
        xlabel=Iterations,
        ylabel=Value,
        width=\subfigsize,
        height=\subfigsize,
        xmin=1, xmax=6,
        xtick={1, 6},
        ymin=-1., ymax=-0.9,
        ytick={-1, -0.95, -0.9},
        ylabel style={yshift=-6pt,font=\subfigfontsize},
        xlabel style={yshift=4pt,font=\subfigfontsize},
        grid=both,
        scaled y ticks=false,
        yticklabel style={
          /pgf/number format/fixed,
          /pgf/number format/fixed zerofill,
          /pgf/number format/precision=2,
         font=\subfigfontsize
        },
        style={very thick,font=\subfigfontsize},
        axis line style={very thick},
        legend style={draw=none}
      ]
        \addplot+[const plot, mark=none, very thick, sthlmGreen, opacity=1., solid] table [col sep=comma, x=iter, y=value] {./results/adversarial_tiger/random/sim_mean.log};
        \addplot+[const plot, mark=none, very thick, sthlmRed, opacity=1., solid] table [col sep=comma, x=iter, y=value] {./results/adversarial_tiger/random/seq_mean.log};
        
        \addplot+[name path=sim_max, const plot, mark=none, thick, sthlmGreen, opacity=0., solid] table [col sep=comma, x=iter, y=value] {./results/adversarial_tiger/random/sim_max.log};
        \addplot+[name path=sim_min, const plot, mark=none, thick, sthlmGreen, opacity=0., solid] table [col sep=comma, x=iter, y=value] {./results/adversarial_tiger/random/sim_min.log};
        \addplot[
            fill=sthlmGreen,
            fill opacity=0.2
        ]
        fill between[
            of=sim_max and sim_min
        ];

        \addplot+[name path=seq_max, const plot, mark=none, thick, sthlmRed, opacity=0., solid] table [col sep=comma, x=iter, y=value] {./results/adversarial_tiger/random/seq_max.log};
        \addplot+[name path=seq_min, const plot, mark=none, thick, sthlmRed, opacity=0., solid] table [col sep=comma, x=iter, y=value] {./results/adversarial_tiger/random/seq_min.log};
        \addplot[
            fill=sthlmRed,
            fill opacity=0.2
        ]
        fill between[
            of=seq_max and seq_min
        ];
      \end{axis}
    \end{tikzpicture}
  \end{subfigure}
  % \hfill
  \begin{subfigure}
    \centering
    \begin{tikzpicture}
      \begin{axis}[
        title={\shortstack{Adversarial \\ Tiger}},
        title style={yshift=-20pt, fill=white},
        xlabel=Iterations,
        ylabel=Cumulated Time (s),
        width=\subfigsize,
        height=\subfigsize,
        xmin=1, xmax=6,
        xtick={1, 6},
        ymin=0, ymax=3000,
        ytick={0, 1500, 3000},
        ylabel style={yshift=-6pt,font=\subfigfontsize},
        xlabel style={yshift=4pt,font=\subfigfontsize},
        grid=both,
        scaled y ticks=false,
        yticklabel style={
          /pgf/number format/fixed,
          /pgf/number format/fixed zerofill,
          /pgf/number format/precision=0,
         font=\subfigfontsize
        },
        style={very thick,font=\subfigfontsize},
        legend style={draw=none},
        style={very thick,font=\subfigfontsize}
      ]
        \addplot+[const plot, mark=none, very thick, sthlmRed, opacity=1.] table [col sep=comma, x=iter, y=cumul_time] {./results/adversarial_tiger/random/seq_mean.log};
        \addplot+[const plot, mark=none, very thick, solid, sthlmGreen, opacity=1.] table [col sep=comma, x=iter, y=cumul_time] {./results/adversarial_tiger/random/sim_mean.log};
      \end{axis}
    \end{tikzpicture}
  \end{subfigure}
  \begin{subfigure}
    \centering
    \begin{tikzpicture}
      \begin{axis}[
        title={\shortstack{Adversarial \\ Tiger}},
        title style={yshift=-20pt, fill=white},
        xlabel=Iterations,
        ylabel=Nb. of points,
        width=\subfigsize,
        height=\subfigsize,
        xmin=1, xmax=6,
        xtick={1, 6},
        ymin=0, ymax=500,
        ytick={0, 250, 500},
        ylabel style={yshift=-6pt,font=\subfigfontsize},
        xlabel style={yshift=4pt,font=\subfigfontsize},
        grid=both,
        scaled y ticks=false,
        yticklabel style={
          /pgf/number format/fixed,
          /pgf/number format/fixed zerofill,
          /pgf/number format/precision=0,
         font=\subfigfontsize
        },
        style={very thick,font=\subfigfontsize},
        legend style={draw=none},
        style={very thick,font=\subfigfontsize},
        legend style={
            at={(0.05,0.75)},
            thin,
            legend cell align=left,
            anchor=north west,
            font=\small,
            inner sep=0.3pt,
            draw=none,
            legend columns=1,
            nodes={scale=0.6},
            name=leg
        },
        legend image post style={
            xscale=0.2
        }
      ]
        \addplot+[const plot, mark=none, very thick, sthlmRed, opacity=1.] table [col sep=comma, x=iter, y=S_size] {./results/adversarial_tiger/random/seq_mean.log};
        \addlegendentry{Seq}
        \addplot+[const plot, mark=none, very thick, solid, sthlmGreen, opacity=1.] table [col sep=comma, x=iter, y=S_size] {./results/adversarial_tiger/random/sim_mean.log};
        \addlegendentry{Sim}
      \end{axis}
      
        \begin{scope}[on background layer]
            \draw[black, thin] (leg.north west) rectangle (leg.south east);
            \draw[black, thin] ([xshift=-0.8pt,yshift=0.8pt]leg.north west) rectangle ([xshift=0.8pt,yshift=-0.8pt]leg.south east);
        \end{scope}
    \end{tikzpicture}
  \end{subfigure}

  \caption{Exploitability, value, cumulative computation time, and number of sampled points of PBVI over iterations on the Recycling and Adversarial Tiger benchmarks ($\ell$=5), evaluated across different random seeds. The \textcolor{sthlmGreen}{green curves} correspond to the \textcolor{sthlmGreen}{simultaneous} variant, while the \textcolor{sthlmRed}{red curves} represent the \textcolor{sthlmRed}{sequential} variant.}
  \label{fig:seed_robustness}
\end{figure}
\begin{table}[H]
\centering
\caption{Snapshot of results (time in seconds, exploitability $\varepsilon$). For each setting, we report the runtime of each algorithm (in seconds), as well as the value \(v(b)\) and the exploitability \(\varepsilon\). Speedup is \texttt{Sim} time divided by \texttt{Seq} time. \textsc{oot}: timeout; \textsc{oom}: out of memory; ‘--’: exploitability budget exceeded. Best values per row are highlighted in \highest{magenta} 
}
\label{table:complete_results:value_and_exploitability}
\scalebox{0.75}{%
\begin{tabular}{@{}r r rrr rrr rrr rrr@{}}
\toprule
\textbf{Game ($\ell$)} & \textbf{Speedup} & \multicolumn{3}{c}{\textbf{\texttt{Seq}}} & \multicolumn{3}{c}{\textbf{\texttt{Sim}}} & \multicolumn{3}{c}{\textbf{\texttt{HSVI}}} & \multicolumn{3}{c}{\textbf{\texttt{CFR+}}} \\
& & time & \(v(b)\) & \(\varepsilon\) & time & \(v(b)\) & \(\varepsilon\) & time & \(v(b)\) & \(\varepsilon\) & time & \(v(b)\) & \(\varepsilon\) \\
\cmidrule(lr){1-1}\cmidrule(lr){2-2}\cmidrule(lr){3-5}\cmidrule(lr){6-8}\cmidrule(lr){9-11}\cmidrule(lr){12-14}

\myrowcolour kuhn poker & 0.55 & 0.2 & \highest{0.055} & \highest{0.00} & 0.11 & \highest{0.055} & \highest{0.00} & \textsc{na} & \textsc{na} & \textsc{na} & 0.01 & \highest{0.055} & \highest{0.00}\\

% \hline

% matchingpennies(3) &  & 2.51 & \highest{0.6} & \highest{0.00} & 26 & \highest{0.6} & \highest{0.00} & 1 & \highest{0.6} & \highest{0.00} & \highest{0.1} & \highest{0.6} & \highest{0.00}\\
% \myrowcolour matchingpennies(4) &  & 10 & 0.79 & \highest{0.00} & 4 & \highest{0.8} & \highest{0.00} & 5 & \highest{0.8} & 0.00 & \highest{1} & \highest{0.8} & \highest{0.00}\\
% matchingpennies(5) &  & 5 & 0.98 & \highest{0.00} & 19 & \highest{1.} & 0.01 & 70 & 1. & 0.00 & \highest{13} & 1. & \highest{0.00}\\
% \myrowcolour matchingpennies(7) &  & 104 & \highest{1.34} & \highest{0.00} & 50 & \highest{1.49} & \highest{} & \textsc{oot} &  &  & \highest{55} & \highest{1.4} & \highest{0.00}\\
% matchingpennies(10) &  & 107 & \highest{1.95} & \highest{0.01} &  & \highest{} & \highest{} & \textsc{oot} &  &  & \highest{86} & \highest{2} & \highest{0.00}\\
% \myrowcolour matchingpennies(12) &  & 46 & \highest{2.23} & \highest{0.08} &  & \highest{} & \highest{} & \textsc{oot} &  &  & & \highest{} & \highest{0.00}\\
% matchingpennies(14) &  & 76 & \highest{2.61} & \highest{0.27} &  & \highest{} & \highest{} & \textsc{oot} &  &  & & \highest{} & \highest{0.00}\\

\hline

\myrowcolour adversarial-tiger(3) & 7.33 & \highest{0.15} & \highest{-0.56} & \highest{0.00} & 1.1 & \highest{-0.56} & \highest{0.00} & 500 & \highest{-0.56} & \highest{0.00} & 1 & \highest{-0.56} & \highest{0.00} \\
adversarial-tiger(4) & 1.4 & \highest{10} & -0.76 & \highest{0.00} & 14 & -0.76 & \highest{0.00} & \textsc{oot} &  &  & 17 & \highest{-0.75} & 0.01 \\
\myrowcolour adversarial-tiger(5) & 1.7 & \highest{30} & -0.97 & \highest{0.00} & 51 & -0.97 & \highest{0.00} & \textsc{oot} &  &  & 181 & \highest{-0.95} & 0.01 \\
adversarial-tiger(7) & 1.83 & \highest{264} & \highest{-1.36} & \highest{0.00} & 485 & -1.37 & 0.02 & \textsc{oot} &  &  & \textsc{oom} &  & \\
\myrowcolour adversarial-tiger(10) & 26.6 & \highest{287} & \highest{-1.95} & \highest{0.02} & 7648 & -1.98 & 0.05 & \textsc{oot} &  &  & \textsc{oom} &  & \\
adversarial-tiger(12) & & \highest{782} & \highest{-2.38} & \highest{0.02} & \textsc{oot} &  &  & \textsc{oot} &  &  & \textsc{oom} &  &  \\
\myrowcolour adversarial-tiger(14) & & \highest{1840} & \highest{-2.87} & \highest{0.18} & \textsc{oot} &  &  & \textsc{oot} &  &  & \textsc{oom} &  & \\

\hline

mabc(3) & 20 & 0.85 & \highest{0.096} & \highest{0.00} & 17 & \highest{0.096} & \highest{0.00} & 70 & \highest{0.096} & \highest{0.00} & \highest{0.5} & \highest{0.096} & \highest{0.00} \\
\myrowcolour mabc(4) & 20.8 & 15 & \highest{0.11} & \highest{0.00} & 312 & 0.10 & 0.01 & \textsc{oot} &  &  & \highest{4} & \highest{0.11} & \highest{0.00} \\
mabc(5) & 3.1 & \highest{44} & 0.11 & \highest{0.00} & 136 & \highest{0.12} & 0.02 & \textsc{oot} &  &  & 51 & 0.12 & 0.01 \\
\myrowcolour mabc(7) & 18.5 & \highest{178} & \highest{0.13} & \highest{0.00} & 3298 & \highest{0.13} & 0.04 & \textsc{oot} &  &  & \textsc{oom} &  & \\
mabc(10) & & \highest{1377} & \highest{0.17} & 0.05 & \textsc{oot} &  &  & \textsc{oot} &  &  & \textsc{oom} &  & \\

\hline

\myrowcolour recycling(3) & 0.18 & 27 & \highest{0.32} & \highest{0.00} & \highest{5} & \highest{0.32} & \highest{0.00} & 430 & 0.32 & \highest{0.00} & 6 & \highest{0.32} & \highest{0.00} \\
recycling(4) & 0.4 & 36 & 0.35 & \highest{0.01} & \highest{15} & \highest{0.36} & \highest{0.01} & \textsc{oot} &  &  & 80 & \highest{0.36} & 0.03 \\
\myrowcolour recycling(5) & 3.1 & \highest{72} & 0.38 & \highest{0.01} & 225 & \highest{0.40} & 0.03 & \textsc{oot} &  &  & \textsc{oom} &  &  \\
recycling(7) & 23.4 & \highest{84} & 0.46 & \highest{0.01} & 1969 & \highest{0.48} & 0.02 & \textsc{oot} &  &  & \textsc{oom} &  &  \\
\myrowcolour recycling(10) & 14.7 & \highest{1616} & 0.6 & \highest{0.00} & 23798 & \highest{0.62} & 0.05 & \textsc{oot} &  &  & \textsc{oom} &  &  \\

\hline

\myrowcolour competitive-tiger(3) & 0.93 & 48 & -0.04 & 0.02 & 45 & -0.04 & 0.02 & 291 & \highest{0.00} & \highest{0.00} & \highest{17} & \highest{0.00} & \highest{0.00} \\
competitive-tiger(4) & 1.98 & \highest{60} & -0.07 & 0.01 & 119 & \highest{-0.06} & \highest{0.00} & \textsc{oot} &  &  & \textsc{oom} &  &  \\
\myrowcolour competitive-tiger(5) & 2.11 & \highest{152} & -0.09 & 0.03 & 322 & \highest{-0.08} & \highest{0.02} & \textsc{oot} &  &  & \textsc{oom} &  &  \\
competitive-tiger(7) & 3.7 & \highest{450} & \highest{-0.13} & 0.08 & 1685 & -0.15 & \highest{0.06} & \textsc{oot} &  &  & \textsc{oom} &  &  \\
\myrowcolour competitive-tiger(10) &  & \highest{847} & \highest{-0.23} & \highest{0.03} & \textsc{oot} &  &  & \textsc{oot} &  &  & \textsc{oom} &  &  \\

\bottomrule
\end{tabular}
}
\end{table}

\subsection{Implementation Details}

The results presented in this paper were run on CPU (AMD Ryzen 7 PRO).

\begin{table}[H]
\centering
\caption{Hyperparameters employed by SeqPBVI.}
\label{table:complete_results:value_and_exploitability}
\scalebox{0.75}{%
\begin{tabular}{@{}r r r r@{}}
\toprule
\textbf{Game} & \textbf{Pruning threshold} & \textbf{Bound. Pruning} & \textbf{Sampling distance}\\

\hline 

\myrowcolour kuhn poker & 1e-5 & 1 & 0.0\\
adversarial tiger & 1e-5 & 1 & 0.1\\
\myrowcolour mabc & 1e-5 & 1 & 0.1\\
recycling & 1e-5 & 1 & 0.1\\
\myrowcolour competitive tiger & 1e-5 & 1 & 0.2\\

\bottomrule
\end{tabular}
}
\end{table}

%%%%%%%%%%%%%%%%%%%%%%%%%%%%%%%%%%%%%%%%%%%%%%%%%%%%%%%%%%%%%%%%%%%%%%%%%%%%%%%
%%%%%%%%%%%%%%%%%%%%%%%%%%%%%%%%%%%%%%%%%%%%%%%%%%%%%%%%%%%%%%%%%%%%%%%%%%%%%%%

\end{document}